\documentclass[11pt]{article}
\RequirePackage{amsthm,amsmath,amsfonts,amssymb}
\usepackage{mathrsfs,braket,enumitem,subcaption,float,xcolor,bbm,algorithm, algpseudocode,tabularx,setspace}
\usepackage{caption}
\RequirePackage[colorlinks,linkcolor=blue,citecolor=blue,urlcolor=blue]{hyperref}
\usepackage{graphicx}
\usepackage{natbib}
\usepackage{url} 

\newcommand{\R}{\mathbb{R}}
\newcommand{\N}{\mathbb{N}}
\newcommand{\Z}{\mathbb{Z}}

\newcommand{\rinit}{r_M}

\newcommand{\prinit}{\Pi_\delta(\rinit)}
\newcommand{\m}{\hat{\mu}}

\newcommand{\sqs}{\ell_2(\mathbb{Z})}
\newcommand{\Mld}[1]{\mathscr{M}_{\infty}({#1})\cap\ell_2(\mathbb{Z})}
\newcommand{\Ml}{\Mld{\delta}}
\newcommand{\ralpha}{{[-1+\delta,1-\delta]}} 

\newtheorem{lem}{Lemma}
\newtheorem{thm}{Theorem}
\newtheorem{prop}{Proposition}
\newtheorem{cor}{Corollary}
\newtheorem{definition}{Definition}
\newtheorem{exmp}{Example}[section]

\usepackage[sectionbib]{bibunits}
\defaultbibliographystyle{unsrtnat}
\defaultbibliography{bib}

\def\pageoption{3}
\def\Title{Efficient shape-constrained inference for the autocovariance sequence from a reversible Markov chain}
\begin{document}

	\def\spacingset#1{\renewcommand{\baselinestretch}%
		{#1}\small\normalsize} \spacingset{1}

	\ifnum\pageoption=2\begin{bibunit}\fi
		
		\title{\bf \Title}
		\author{Stephen Berg\thanks{Both authors contributed equally.} \, and Hyebin Song\footnotemark[1]  \thanks{Corresponding author: hps5320@psu.edu} \\
			Department of Statistics, Pennsylvania State University}
		\maketitle
		\bigskip
		\begin{abstract}
			In this paper, we study the problem of estimating the autocovariance sequence resulting from a reversible Markov chain. A motivating application for studying this problem is the estimation of the asymptotic variance in central limit theorems for Markov chains. 
			We propose a novel shape-constrained estimator of the autocovariance sequence, which is based on the key observation that the representability of the autocovariance sequence as a moment sequence imposes certain shape constraints.
			We examine the theoretical properties of the proposed estimator and provide strong consistency guarantees for our estimator. In particular, for geometrically ergodic reversible Markov chains, we show that our estimator is strongly consistent for the true autocovariance sequence with respect to an $\ell_2$ distance, and that our estimator leads to strongly consistent estimates of the asymptotic variance.
			Finally, we perform empirical studies to illustrate the theoretical properties of the proposed estimator as well as to demonstrate the effectiveness of our estimator in comparison with other current state-of-the-art methods for Markov chain Monte Carlo variance estimation, including batch means, spectral variance estimators, and the initial convex sequence estimator.
		\end{abstract}
		
		\noindent%
		\vfill
		
		\newpage
		\spacingset{1.5} 
		\section{Introduction}
		
		Markov chain Monte Carlo (MCMC) is a routinely used tool for approximating intractable integrals of the form $\mu=\int g(x)\pi(dx)$, where $\pi$ is an intractable probability measure on a measurable space $(\mathsf{X},\mathscr{X})$ and $g:\mathsf{X}\to\mathbb{R}$ is a $\pi$-integrable function. In MCMC, a Markov chain $X_0,X_1,X_2,...$ with transition kernel $Q$ and stationary probability measure $\pi$ is simulated for some finite number of iterations $M$, possibly after an initial burn-in period, and $\mu$ can then be estimated by the empirical average \begin{align*}
			Y_{M}=M^{-1}\sum_{t=0}^{M-1}g(X_t).
		\end{align*}
		
		In general, $g(X_t)$ from a Markov chain may have nonzero covariance. For a Markov chain transition kernel $Q$ with a unique stationary probability measure $\pi$, define the autocovariance sequence $\gamma=\{\gamma(k)\}_{k=-\infty}^{\infty}$  \begin{align*}
			\gamma(k)= E_\pi [(g(X_0)-\mu)(g(X_{|k|})-\mu)],\;\;\;\;\; k\in \mathbb{Z}.
		\end{align*} 
		In this work, we study the problem of estimating the autocovariance sequence $\gamma \in \R^\Z$ from a reversible Markov chain by exploiting shape constraints satisfied by the autocovariance sequence $\gamma$. It is a well known result that for a reversible Markov chain, the autocovariance sequence $\gamma$ admits the following representation [e.g., \citealp{Rudin1973-ry}]:
		\begin{align}
			\gamma(k) = \int x^{|k|} F(dx) \;\;\;\;\; k\in\mathbb{Z}\label{eq:spectralMeasure}
		\end{align}
		for a unique positive measure $F$ supported on $[-1,1]$. For a function on $\R$ or $\Z$, admitting a certain mixture representation has an implication in its global shape \citep{Hausdorff1921-kl, Feller1939-fp,Steutel1969-mn,jewell1982mixtures,Balabdaoui2020-nn}. For instance, if the support of $F$ in \eqref{eq:spectralMeasure} is contained in $[0,1]$, $\gamma$ is \textit{completely monotone}, meaning the inequalities $(-1)^n \Delta^n \gamma(j) \geq 0$ are satisfied for all $j, n \in \mathbb{N}$ where $\Delta^n \gamma(j)=\Delta^{n-1} \gamma(j+1)-\Delta^{n-1} \gamma(j)$ is a difference operator with $\Delta^0 \gamma=\gamma$. While $\gamma$ is not, in general, completely monotone because the support of $F$ may extend outside of $[0,1]$, the representation \eqref{eq:spectralMeasure} still imposes an infinite number of shape constraints on $\gamma$ (see Proposition~\ref{prop:momentSeq}). To exploit such structure in $\gamma$, in this work, we propose an estimator of the autocovariance sequence based on the $\ell_2$ projection of an initial input autocovariance sequence estimate, such as the ordinary empirical autocovariance sequence, onto the set of sequences admitting a representation as in~\eqref{eq:spectralMeasure}.
		
		\subsection{Main application: asymptotic variance estimation for MCMC estimates}
		There are several motivations for the estimation of the autocovariance sequence. As a main application, we consider the problem of estimating the asymptotic variance in a Markov chain central limit theorem. This problem has practical importance, as the asymptotic variance quantifies uncertainties in the MCMC estimate $Y_M$. 
		Under mild conditions~\citep{meyn2009markov}, a central limit theorem can be established for $Y_M$ such that \begin{align}
			\sqrt{M}(Y_{M}-\mu)\overset{d}{\to} N(0,\sigma^2)\label{eq:clt}
		\end{align} where \begin{align}
			\sigma^2&=\sum_{k=-\infty}^{\infty}\gamma(k). \label{eq:avar}
		\end{align} The infinite sum in~\eqref{eq:avar} arises from covariance between terms in the sum in the definition of $Y_{M}$. From \eqref{eq:clt}, the variance of the empirical mean $Y_M$ from an MCMC simulation as an estimator of $\mu$ is quantified, in an asymptotic sense, by the asymptotic variance $\sigma^2$. In turn, from \eqref{eq:avar}, $\sigma^2$ can be estimated based on an estimate of the autocovariance sequence $\gamma$. Fixed width stopping rules for MCMC, as in~\citet{jones2006fixed}, \citet{bednorz2007few}, \citet{flegal2008markov}, \citet{latuszynski}, \citet{flegal2015relative}, and \citet{vats2019multivariate}, depend on an estimate of $\sigma^2$. 
		
		One natural estimate for $\gamma(k)$ based on the first $M$ iterates $X_0,X_1,...,X_{M-1}$ is the empirical autocovariance $\tilde{r}_M(k)$, defined by
		\begin{align}
			\tilde{r}_M(k) = \begin{cases}
				\frac{1}{M}\sum_{t=0}^{M-1-|k|} (g(X_t)-Y_{M}) (g(X_{t+|k|})-Y_{M}) & |k|\leq M-1\\
				0 & |k|>(M-1).
			\end{cases}\label{eq:emp1_example}
		\end{align} It is well known that some natural estimators of $\sigma^2$ based on the $\tilde{r}_{M}(k)$ sequence are inconsistent. For the empirical autocovariances with mean centering based on the empirical mean $Y_{M}$ as in~\eqref{eq:emp1_example}, an elementary calculation shows that $\sum_{k=-(M-1)}^{M-1}\tilde{r}_{M}(k)=0$, and the estimator $\hat{\sigma}^2_{M,\textrm{emp}}=\sum_{k=-(M-1)}^{M-1}\tilde{r}_M(k)=0$ is thus inconsistent as an estimator of $\sigma^2$. With centering based on the true mean $\mu$ rather than $Y_{M}$ in~\eqref{eq:emp1_example}, the corresponding estimator converges in distribution to a scaled $\chi^2$ random variable~\citep{anderson1971statistical,flegal2010batch}, and is thus also inconsistent. These difficulties have led to methods for estimating $\sigma^2$ with better statistical properties. These methods include spectral variance estimators~\citep{anderson1971statistical,flegal2010batch}, estimators based on batch means~\citep{priestley1981spectral,flegal2010batch,chakraborty2022estimating}, and a class of methods for reversible Markov chains called initial sequence estimators~\citep{geyer1992practical,kosorok2000error,dai2017multivariate}.

		The batch means and spectral variance estimators have known consistency properties. In particular, they are a.s. consistent for $\sigma^2$, and have $M^{1/3}$ rate of convergence with an optimal choice of batch or window size \citep{Damerdji1991-oj, flegal2010batch}. Practically, they involve tuning parameters which are known in advance only up to a constant of proportionality. For instance, the batch means, overlapping batch means, and spectral variance estimators in~\citet{flegal2010batch} require the selection of a batch size $b_M$ depending on the Markov chain sample length $M$. The optimal setting is $b_M=CM^{1/3}$, but the constant of proportionality depends on problem-dependent parameters that will typically be unknown.
		
		\citet{geyer1992practical}, on the other hand, introduces initial sequence estimators for estimating $\sigma^2$. The initial sequence estimators exploit positivity, monotonicity, and convexity constraints satisfied for reversible Markov chains by the sequence $\Gamma=\{\Gamma(k)\}_{k=0}^{\infty}$ defined by
		\begin{align}\label{eq:geyer_Gammak}
			\Gamma(k) := \gamma(2k)+\gamma(2k+1)\;\;\;\;\;k=0,1,2,...
		\end{align} 
		More specifically, to impose such constraints, first the initial positive sequence estimator is obtained by truncating the empirical $\hat{\Gamma}_{M}(k) = \tilde{r}_M(2k)+\tilde{r}_M(2k+1)$ sequence at the first $k$ such that $\hat{\Gamma}_{M}(k)<0$, to obtain $\hat{\Gamma}^{\rm (pos)}_{M}=\{\hat{\Gamma}_{M}(k)\}_{k=0}^{T-1}$ where $T:=\min\{k \in \N; \hat{\Gamma}(k) <0\}$. The argument given in \citet{geyer1992practical} for truncating the sequence at $T$ is that $T$ is the estimated time point when the autocovariance curve falls below the noise level. In addition to the initial positive sequence estimator, \citet{geyer1992practical} introduces the initial monotone sequence and initial convex sequence estimators. The initial monotone sequence and initial convex sequence estimators can then be calculated by replacing each $\hat{\Gamma}^{\rm (pos)}(k)$ with the minimum of the preceding ones and with the $k$th element of the greatest convex minorant of the initial positive sequence, respectively.
		
		Despite their simplicity, initial sequence estimators have very strong empirical performance and do not require the choice of a tuning parameter value, making them very useful in practice. For example, the widely used Stan software \citep{stan2019} employs the initial sequence estimators to estimate the effective sample size of Markov chain simulations. However, the statistical guarantees of the initial sequence estimators are somewhat lacking compared to the batch means and spectral variance estimators. To our knowledge, the only consistency results for the initial sequence estimators are that the initial sequence estimates asymptotically do not underestimate $\sigma^2$, that is,
		\begin{align}
			\underset{M\to\infty}{\lim\inf }\; \widehat{\sigma^2}_{M,\textrm{init}}\geq \sigma^2\;a.s.\label{eq:liminf}, 
		\end{align} as in~\citet{geyer1992practical,kosorok2000error,brooks2011handbook}, and \citet{dai2017multivariate}, rather than $\underset{M\to\infty}{\lim}\widehat{\sigma^2}_{M,\textrm{init}}=\sigma^2$ almost surely.
		
		\subsection{Review on estimation with shape constraints and connection to autocovariance sequence estimation}
		The work of \cite{geyer1992practical} can be viewed as an example of shape constrained inference, where the sequence $\{\Gamma_k\}_{k\in\N}$ is estimated in such a way that various shape constraints (positivity, monotonicity, and convexity) are enforced. Shape constrained inference has a long history in statistics. One of the standard examples is the isotonic regression, where in the most basic scenario one observes $n$ independent random variables $Y_i$ which are assumed to be noisy observations of some monotone increasing signal, i.e., $E[Y_1]\le E[Y_2] \le\dots E[Y_n]$. The goal is to estimate the underlying $n$-dimensional signal \citep{ Barlow1972-ej,robertson1988order}. However, shape constrained inference is not limited to the estimation of a finite dimensional vector and to monotonicity constraints. In fact, shape constrained inference has also been applied to infinite-dimensional problems where the quantity of interest is an infinite-dimensional vector or a function on $\R$ with different shape constraints. Examples include nonparametric estimation of monotone sequences or functions, the estimation of a convex or log-convex density, etc. \citep{Grenander1956-wl,Jankowski2009-tm, Dumbgen2011-wv, Balabdaoui2015-yc, Kuchibhotla2017-id}.

		Among such constraints, $k$-monotonicity, which is a refinement of the monotonicity property, has been studied by several authors \citep{Balabdaoui2007-na,Lefevre2013-lr, Durot2015-sb,Chee2016-ae, Giguelay2017-sy}. A sequence $m$ is called a $k$-monotone decreasing sequence if its successive differences up to order $k$ are alternatively nonnegative and nonpositive, i.e., 
		\begin{align}\label{eq:k-monotonicity}
			(-1)^n \Delta^n m(j) \ge 0 \mbox{ for } \;j\in\mathbb{N},\;n=0,\dots,k
		\end{align}
		where $\Delta^n m(j) = \Delta^{n-1} m(j+1) - \Delta^{n-1}m(j)$ is a difference operator with $\Delta^0 m = m$. The case of $k=0$ corresponds to nonnegativity, so that $(-1)^0\Delta^0m(j)=m(j)\geq 0$. The case $k=1$ corresponds to monotonicity $m(j+1)-m(j)\leq 0$ in addition to nonnegativity, and $k=2$ corresponds to convexity $m(j+2)-m(j+1) \ge m(j+1)-m(j)$ in addition to nonnegativity and monotonicity. 
		
		When $(-1)^n \Delta^n m(j) \ge 0$ for all $j,n\in\N$, the sequence $m$ is called completely monotone. For functions on the real line, analogous versions of complete monotonicity involving derivatives rather than differences have been considered. Complete monotonicity conditions have been investigated by various authors. One prominent feature of prior results is an equivalence between satisfying a complete monotonicity constraint and admitting a mixture representation. For instance, \cite{Hausdorff1921-kl} proved that 
		a sequence $m$ is completely monotone if and only if the sequence $m$ admits a moment representation, namely, if there exists a nonnegative measure $F$ supported on $[0,1]$ such that $m(k)$ is the $k$th moment of $F$, i.e., $m(k) =\int x^k F(dx)$. Similarly, completely monotone functions on $\mathbb{R}^{+}\cup \{0\}$ can be represented as a scale mixture of exponentials \citep{Feller1939-fp,jewell1982mixtures}, and a completely monotone probability mass function (pmf) can be represented as a mixture of geometric pmfs \citep{Steutel1969-mn}. The latter fact was used in the recent work by \cite{Balabdaoui2020-nn} for the estimation of a completely monotone pmf using nonparametric least squares estimation.
		
		In the context of asymptotic variance estimation, the result of \cite{geyer1992practical} on the $\Gamma$ sequence can be refined using the concept of complete monotonicity. Recall that \cite{geyer1992practical} showed that the sequence $\Gamma$ obtained as the rolling sum of $\gamma$ with window size $2$, i.e., $\Gamma(k) = \gamma(2k) + \gamma(2k+1)$, is $2$-monotone. In this paper, we show that the $\Gamma$ sequence is not only $2$-monotone but completely monotone (Proposition \ref{prop:gamma_moment_seq}). This suggests that higher order shape structure could be exploited in the estimation of $\Gamma(k)$ and, consequently, the asymptotic variance. However, while the $\Gamma$ sequence is completely monotone, the set of completely monotone sequences is not entirely satisfactory to work with for our purpose of estimating the \textit{entire} autocovariance sequence $\gamma$, since $\gamma$ may not be a completely monotone sequence. 
		
		\paragraph{Our contribution and organization of the paper}
		
		To our knowledge, this is the first work in which the moment representation of the autocovariance sequence~\eqref{eq:spectralMeasure} is directly exploited in this manner to carry out shape-constrained inference for the estimation of the autocovariance sequence and asymptotic variance. 
		Our work is the first to use shape-constrained inference methods to provide a provably consistent estimator for the asymptotic variance for a Markov chain. The work of \cite{Balabdaoui2020-nn} on estimating a completely monotone pmf is the most similar to ours of which we are aware. However, \citet{Balabdaoui2020-nn} consider a substantially different setting involving the estimation of a completely monotone probability mass function (pmf) from iid samples. In our setting, the dependence between observations necessitates the use of different tools for the statistical analysis. To the best of our knowledge, this is the first work in which shape-constrained inference is used to alter the convergence property of input sequences as well.
		
		The remainder of the paper is organized as follows. In Section~\ref{sec:mcmcIntro}, we introduce background on Markov chains and prove Proposition~\ref{prop:gamma_moment_seq} on the representation of $\gamma$ and $\Gamma$ as moment sequences. In Section~\ref{sec:momentSequences}, we introduce our proposed estimator, the moment least squares estimator, and study some basic properties of the proposed estimator. In Section~\ref{sec:guarantee}, we provide statistical convergence results for the moment least squares estimator. In particular, we prove the almost sure convergence in the $\ell_2$ norm of the estimated autocovariance sequence (Theorem~\ref{thm:as_l2conv}), the almost sure vague convergence of the representing measure for the moment least squares estimator to the representing measure for $\gamma$ (Proposition~\ref{prop:vague_convergence_muhat}),  and the almost sure convergence of the estimated asymptotic variance (Theorem~\ref{thm:as_conv_avar}). In Section~\ref{sec:emp}, we show the results of our empirical study, in which the moment least squares estimator performs well relative to other state-of-the-art estimators for MCMC asymptotic variance and autocovariance sequence estimation.
		
		\section{Problem set-up\label{sec:mcmcIntro}}
		
		We now describe our setup in detail and fix some notation. We consider a $\psi$-irreducible, aperiodic Markov chain $X=\{X_t\}_{t=0}^{\infty}$ evolving over $t$ on a measurable space $(\mathsf{X},\mathscr{X})$, where the state space $\mathsf{X}$ is a complete separable metric space and $\mathscr{X}$ is the associated Borel $\sigma$-algebra. We let $\pi$ denote a probability measure defined on $(\mathsf{X},\mathscr{X})$ with respect to which we would like to compute expectations. We use $g:\mathsf{X}\to\mathbb{R}$ to denote a function for which it is of interest to obtain $\mu=\int g(x)\pi(dx)$. We define a transition kernel as a function $Q:\mathsf{X}\times\mathscr{X}\to [0,1]$ such that $Q(\cdot,A):\mathsf{X}\to [0,1]$ is an $\mathscr{X}$-measurable function for each $A\in\mathscr{X}$ and $Q(x,\cdot):\mathscr{X}\to [0,1]$ is a probability measure on $(\mathsf{X},\mathscr{X})$ for each $x\in \mathsf{X}$. For a probability measure $\pi$ on $(\mathsf{X},\mathscr{X})$, a probability kernel $Q$ is said to be $\pi$-stationary if $\pi (A)=\int Q(x,A)\pi(dx)$ for all $A\in\mathscr{X}$. 
		An initial measure $\nu$ on $\mathcal{X}$ and a transition kernel $Q$ define a Markov chain probability measure $P_\nu$ for $X=(X_0,X_1,X_2,\dots)$ on the canonical sequence space $(\Omega, \mathcal{F})$. We write $E_\nu$ to denote expectation with respect to $P_{\nu}$.

		For a function $f:\mathsf{X}\to\mathbb{R}$ and a transition kernel $Q$, we define the linear operator $Q$ by 
		\begin{align}
			Qf(x)=\int Q(x,dy) f(y)\label{eq:Qf}
		\end{align}We define $Q^0f(x)=f(x)$, $Q^1f(x)=Q f(x)$, and $Q^tf(x)=Q(Q^{t-1}f)(x)$ for $t>1$, and we define $Q^t(x,A)=Q^tI_A(x)$, where $I_A(\cdot)$ is the indicator function for the set $A$. 
		We let $L^2(\pi)$ be the space of functions which are square integrable with respect to $\pi$, i.e., $L_2(\pi) = \{f:\mathsf{X} \to \R; \int f(x)^2 \pi(dx)< \infty\}$. 
		For functions $f, g\in L^2(\pi)$, we define an inner product
		\begin{align}\label{def:fn_inner_product}
			\braket{f,g}_\pi = \int f(x) g(x) \pi(dx) .
		\end{align}
		We note that $L^2(\pi)$ is a Hilbert space equipped with the inner product \eqref{def:fn_inner_product}. For $f \in L^2(\pi)$, we define $\|f\|_{L^2(\pi)} = \sqrt{\braket{f,f}_\pi}$. Also, for an operator $Q$ on $L^2(\pi)$, we define $\|Q\|_{L^2(\pi)} = \sup_{f; \|f\|_{L^2(\pi)} \le 1} \|Qf\|_{L^2(\pi)}$ and we say $Q$ is bounded if $\|Q\|_{L^2(\pi)}<\infty$. 
		
		We say that a transition kernel $Q$ satisfies the reversibility property with respect to $\pi$ if \begin{align}
			\braket{f_1,Qf_2}_\pi = \braket{Qf_1,f_2}_\pi
			\label{eq:reversible}
		\end{align} for any functions $f_1,f_2 \in L^2(\pi)$, i.e., if $Q$ is a self-adjoint operator. 
		Reversibility with respect to $\pi$ is a sufficient condition for $\pi$-stationarity of $Q$, since for a reversible transition kernel $Q$, we have \begin{align*}
			\pi(A)=\int I_A(x)QI_X(x)\pi(dx)=\int I_X(x)QI_A(x)\pi(dx).
		\end{align*} 
		
		The spectrum of the operator $Q$ plays a key role in determining the mixing properties of a Markov chain with transition kernel $Q$. Recall that for an operator $T$ on the Hilbert space $L^2(\pi)$, the spectrum of $T$ is defined as
		\begin{align}\label{def:spectrum_of_Q}
			\sigma(T) = \{\lambda \in \mathbb{C}; (T-\lambda I)^{-1} \mbox{does not exist or is unbounded }\}.
		\end{align}
		For Markov operators $Q$, we define the spectral gap $\delta$ of $Q$ as $\delta = 1- \sup\{ |\lambda| ; \lambda \in \sigma(Q_0) \}$ where $Q_0$ is defined as
		\begin{align}\label{def:Q0}
			Q_0f = Qf - E_\pi[f(X_0)] f_0
		\end{align}
		and $f_0\in L^2(\pi)$ is the constant function such that $f_0(x) = 1$ for all $x\in \mathsf{X}$. It is easy to check that $Q_0$ is self-adjoint and bounded.
		If $Q$ is reversible, $Q$ has a positive spectral gap $(\delta>0)$ if and only if the chain is geometrically ergodic \citep{Roberts1997-qk, kontoyiannis2012geometric}. In addition, $(1-\delta)^k$ is the maximal lag $k$ correlation of any two functions, and therefore for any function $f$ and $Y_{fM} = M^{-1} \sum_{t=0}^{M-1} f(X_t)$, the asymptotic variance $\sigma_f^2$ of $ \sqrt{M}(Y_{fM}-E_\pi[f(X_0)])$ is bounded above by
		\begin{align*}
			\sigma^2_f = \gamma_f(0) + 2\sum_{k \ge 1}\gamma_f(k) \le \gamma_f(0) + 2\sum_{k \ge 1} (1-\delta)^k \gamma_f(0) = \frac{2-\delta}{\delta} \gamma_f(0) 
		\end{align*}
		where $\gamma_f(k) = {\rm Cov}_\pi(f(X_0), f(X_k))$.

		In the remainder, we consider a discrete time Markov chain $X=\{X_t\}_{t=0}^{\infty}$ with stationary distribution $\pi$ and $\pi$-reversible transition kernel $Q$ with a positive spectral gap. We let $g$ be a square integrable function with respect to $\pi$, and use $\gamma(k)$ defined by
		\begin{align*}
			\gamma(k)= {\rm Cov}_\pi\{g(X_0), g(X_{|k|})\} = \braket{g, Q_0^{|k|}g}_\pi  \quad\mbox{for } k\in\mathbb{Z}
		\end{align*}
		to denote the lag $k$ autocovariance of the stationary time series $\{g(X_t)\}_{t=0}^{\infty}$ obtained with $X_0\sim \pi$. We use $\gamma = \{\gamma(k)\}_{k\in\Z}$ to denote the autocovariance sequence on $\Z$.
		We summarize our assumptions on the Markov chain $X$ as follows for future reference:
		\begin{enumerate}[label=(A.\arabic*)]
			\item \label{cond:harris_ergodicity}(Harris ergodicity) $X$ is $\psi$-irreducible, aperiodic, and positive Harris recurrent.
			\item \label{cond:piReversible}(Reversibility) The transition kernel $Q$ is $\pi$-reversible for a probability measure $\pi$ on $(\mathsf{X},\mathscr{X})$.
			\item \label{cond:geometric_ergodicity}(Geometric ergodicity) There exists a real number $\rho<1$ and a non-negative function $M$ on the state space $\mathsf{X}$ such that
			\begin{align*}
				\|Q^n(x,\cdot) - \pi(\cdot)\|_{\rm TV} \le M(x)\rho^n, \mbox{  for all } x\in \mathsf{X},
			\end{align*}
			where $\|\cdot\|_{\rm TV}$ is the total variation norm.
		\end{enumerate}
		Throughout the paper, we assume that the function of interest $g:\mathsf{X}\to\mathbb{R}$ is in $L^2(\pi)$, i.e,
		\begin{enumerate}[label=(B.\arabic*)]
			\item \label{cond:integrability} (Square integrability)$\int g(x)^2\pi(dx)<\infty$.
		\end{enumerate}
		
		For the definitions of $\psi$-irreducibility, aperiodicity, and positive Harris recurrence, see e.g., \citet{meyn2009markov}. 
		Reversibility is a key requirement for our estimator because it allows us to use the shape constraints implied by the spectral decomposition of the Markov chain kernel (see Proposition \ref{prop:gamma_moment_seq}). Many practical transition kernels satisfy $\pi$-reversibility. Notably, all Metropolis-Hastings transition kernels satisfy reversibility. Additionally, all Gibbs component update kernels are reversible. 
		As noted by a referee, in practice, it is common to combine a set of reversible transition kernels $\{Q_k\}_{k=1}^K$, such as those from Metropolis-Hastings or Gibbs updates, to form a joint transition mechanism $Q$. The reversibility of the combined mechanism $Q$ depends on the way in which the individual kernels $Q_k$ are combined. For example, in deterministic scan sampling, where each update consists of sequentially applying $Q_k$, $k=1,\dots,K$, the resulting kernel $Q(x,A)=Q_1Q_2\cdots Q_K(x,A)$ is generally non-reversible. On the other hand, there are schemes for combining reversible kernels $Q_k$ in such a way that the resulting $Q$ is reversible. For example, the random scan transition kernel $Q=K^{-1}\sum_{k=1}^{K}Q_k$, corresponding to randomly selecting the transition kernel at each iteration, is reversible.
		Additionally, random permutation scans, in which at each iteration the reversible $Q_k$ are composed in a randomly permuted order, and palindromic scan updates, in which $Q=Q_1\dots Q_{K-1}Q_KQ_{K-1}\dots Q_{1}$, lead to reversible Markov chains~\citep[see, e.g., page 376 of][]{robert2004monte}. 
		Finally, we note that in data augmentation Gibbs sampling, the marginal chains are reversible~\citep[see, e.g.,][]{wongKongLiu,robert2004monte}.

		Geometric ergodicity implies exponential convergence of the Markov chain $X$ to its target distribution $\pi$. When the state space $\mathsf{X}$ is finite, all irreducible and aperiodic Markov chains are geometrically ergodic. While this is no longer true for infinite state space, geometric ergodicity remains a theoretically and practically important condition for Markov chains \citep[e.g.][]{roberts1998markov,jones2022markov}. For example, geometric ergodicity provides one of the simplest sufficient conditions for the Markov chain central limit theorem (CLT) to hold. In fact, for a reversible geometrically ergodic Markov chain, a finite second moment of the function of interest $g$ is sufficient to establish a CLT (e.g., \citealp{jones2004central}). The establishment of geometric ergodicity is usually done on a case-by-case analysis, and many works have studied geometric ergodicity of popular samplers (e.g., \citealp{mengersen1996rates,roberts1996geometric, jarner2000geometric, jarner2003necessary,johnson2012variable,chakraborty2017convergence,livingstone2019geometric,durmus2023convergence}).
		
		The following proposition shows that both the autocovariance sequence $\gamma$ and rolling sum $\Gamma$ of the sequence $\gamma$ with a window size of $2$ from a reversible chain have the following \textit{moment} representations, namely there exist measures $F$ and $G$ supported on $[-1,1]$ and $[0,1]$, respectively, such that $\gamma(k)$ and $\Gamma(k)$ are the $k$th moments of $F$ and $G$, respectively. Let $\mathcal{M}_\R$ denote the set of finite regular measures on $\R$.  
		
		\begin{prop}\label{prop:gamma_moment_seq}
			Assume \ref{cond:piReversible} and~\ref{cond:integrability}.  
			\begin{enumerate}
				\item The true autocovariance sequence $\gamma(k) = \braket{g, Q_0^k g}_\pi$, $k\in \Z$, has the following representation for some $F\in \mathcal{M}_\R$
				\begin{align}\label{eq:gamma_k}
					\gamma(k) = \int_{\sigma(Q_0)} x^{|k|} F(dx),
				\end{align}
				where $\sigma(Q_0)$ is the spectrum of the linear operator $Q_0$ defined as in \eqref{def:Q0}. Moreover, $\sigma(Q_0)$ lies on the real axis, and $\sigma(Q_0) \subseteq [-1,1]$.
				
				\item The sequence $\Gamma=\{\Gamma(k)\}_{k\in\mathbb{N}}$ defined by $\Gamma(k) = \gamma(2k)+\gamma(2k+1)$, $k\in\mathbb{N}$, has the following representation for some $G \in \mathcal{M}_\R$
				\begin{align}\label{eq:Gamma_k}
					\Gamma(k) = \int_{\sigma(Q_0^2)} x^{k} G(dx),
				\end{align}
				and $\sigma(Q_0^2) \subseteq [0,1]$.
				
				\item If we additionally assume \ref{cond:harris_ergodicity} and \ref{cond:geometric_ergodicity} in addition to \ref{cond:piReversible} and \ref{cond:integrability},  there exists $0<\delta_0\leq 1$ such that $\sigma(Q_0) \subseteq [-1+\delta_0,1-\delta_0] $ and  $\sigma(Q_0^2) \subseteq [0,(1-\delta_0)^2] $.
			\end{enumerate}
			
		\end{prop}
		The proof of Proposition \ref{prop:gamma_moment_seq} is deferred to Supplementary Material S3.1~\citep{berg2023efficientsuppl}. In the example below, the moment representation of the autocovariance sequence from a reversible Markov chain is illustrated using an AR(1) chain.
		
		\begin{exmp}{(Autoregressive chain example)\label{exmp:ar1exmp}}
			Consider an AR(1) autoregressive process with  
			$X_{t+1}=\rho X_t+\epsilon_{t+1}, \,\, t=0,1,2,\dots$
			, where $\epsilon_t\overset{iid}{\sim}N(0,\tau^2)$ and $\rho\in (-1,1)$. The stationary measure $\pi$ for the $X_t$ chain is the measure corresponding to a $N(0,\tau^2/(1-\rho^2))$ random variable, and the $X_t$ chain can be shown to be reversible with respect to $\pi$. Consider the autocovariance sequence $\gamma(k)=E_{\pi}[\bar{g}(X_0)\bar{g}(X_{|k|})]$ with the identity function $g(x) = x$. Since $E_\pi[g(X_0)] = 0$, we have
			\begin{align*}
				\gamma(k) = {\rm Cov}_\pi (g(X_0),g(X_k)) = {\rm Var}_\pi (g(X_0)) \rho^{|k|} = \frac{\tau^2}{1-\rho^2}\rho^{|k|}.
			\end{align*}
			Then $\gamma(k)$ can be represented as $\gamma(k) =\int x^{|k|}F(dx)$ for all $k \in \Z$ by letting $F=\frac{\tau^2}{1-\rho^2}\delta_{\rho}$, where $\delta_\rho$ denotes a unit point mass measure at $\rho$.
		\end{exmp}
		
		We note that the second statement of Proposition \ref{prop:gamma_moment_seq} implies that the $\Gamma(k)$ sequence is \textit{completely monotone}, and therefore is a refinement of the result in \cite{geyer1992practical} which showed that $\Gamma(k)$ is $2$-monotone. This is due to the theorem of \citet{Hausdorff1921-kl} below, in which an equivalence is shown between $[0,1]$-\textit{moment} sequences (sequences with the representation $m(k) = \int x^k F(dx)$ for some $F$ with ${\rm Supp}(F) \subseteq [0,1]$; see Definition \ref{def:representingMeasure} for the formal definition) and \textit{completely monotone} sequences satisfying inequalities \eqref{eq:k-monotonicity} for all $k,n \in \mathbb{N}$. The relationship between sequences admitting certain moment representations and their shape constraints will be further explored in the following Section \ref{sec:momentSequences}.
		
		\begin{thm}[Hausdorff moment theorem \citep{Hausdorff1921-kl}]\label{thm:hausdorff_moment}
			There exists a representing measure $\mu$ supported on $[0,1]$ for $m\in\R^\N$ if and only if $m$ is a completely monotone sequence. Additionally, if $m$ is a completely monotone sequence, the representing measure $\mu$ for $m$ is unique.
		\end{thm}
		
		We have from Proposition \ref{prop:gamma_moment_seq} that $\gamma$ is a $[-1,1]$-moment sequence. In general, $\gamma$ is not a completely monotone sequence as its representing measure can have mass in $[-1,0)$. A simple example is the autocovariance sequence from an AR(1) stationary chain with a negative AR(1) coefficient. The autocovariances oscillate between positive and negative values as $k\to \infty$ and therefore cannot decrease monotonically. 
		
		\subsection*{Notations} 
		We let $\mathbb{N}$ be the set of non-negative integers $\{0,1,2,...\}$ and $\mathbb{Z}$ the set of integers $\{...,-1,0,1,...\}$. For a sequence $m$ on $\N$ or $\Z$, we define an $\ell_p$ norm for $m$ by $\|m\|_p = (\sum_{k} |m(k)|^p)^{1/p}$ for $p=1,2,\dots,$ and $\|m\|_\infty =\max_k |m(k)|$. In addition, when $p=2$, we omit the subscript and write $\|\cdot\| = \|\cdot\|_2$.
		We use $\ell_p(\N)$ (or $\ell_p(\Z))$ to denote the space of sequences on $\N$ (or $\Z$) with finite $\ell_p$ norms. In particular, $\ell_1(\mathbb{Z})$ is the space of absolutely summable sequences on $\Z$, i.e., $\ell_1(\mathbb{Z})= \{m\in\R^\Z;\sum_{k=-\infty}^{\infty}|m(k)|<\infty\}$ and $\sqs$ is the space of square summable sequences on $\Z$, i.e., $\ell_2(\mathbb{Z})= \{m\in\R^\Z;\sum_{k=-\infty}^{\infty}m^2(k)<\infty\}$. We equip $\sqs$ with a usual inner product  $\braket{x,y}=\sum_{k=-\infty}^{\infty}x(k)y(k)$ for $x,y\in\sqs$. Then $\|x\| = \sqrt{\braket{x,x}} = \|x\|_2$. Also, for $\alpha \in [-1,1]$, we define $x_\alpha =\{x_\alpha(k)\}_{k\in \Z}$ such that $x_\alpha(k) = \alpha^{|k|}$ for $k \in \Z$. Note that for $\alpha \in(-1,1)$, $x_\alpha \in \sqs$. Finally, for a measure $\mu$, we let ${\rm Supp}(\mu)$ denote the support of $\mu$.

		\section{Moment least squares estimator (Moment LSE)} \label{sec:momentSequences}
		
		We now introduce the moment least squares estimator. We first formally define moment sequences and moment spaces. 
		\begin{definition}[moment sequence and representing measure]\label{def:representingMeasure}
			We say that a sequence $m$ is an $[a,b]$-moment sequence if there exists a positive measure $\mu$ supported on $[a,b]$ for some $-\infty< a \le b < \infty$ such that the equation
			\begin{align}\label{eq:momentRep}
				m(k) = \int x^{|k|} \mu(dx)
			\end{align}
			holds for any $k\in \N$ (if $m=\{m(k)\}_{k=0}^{\infty}$ is a sequence defined on $\N$) or any $k \in \Z$ (if $m=\{m(k)\}_{k=-\infty}^{\infty}$ is a sequence defined on $\Z$). We say that $\mu$ is a representing measure for the sequence $m$.
		\end{definition}
		
		For a closed set $C \subseteq \R$, we write $\mathscr{M}_\infty (C)$ to denote the the set of sequences on $\R^\Z$ with a moment representation with a measure supported on $C$. For example, $\mathscr{M}_\infty ([a,b])$ is the set of $[a,b]$-moment sequences. By definition, we have $\mathscr{M}_\infty(I_1) \subseteq \mathscr{M}_\infty (I_2)$ if $I_1\subseteq I_2$ for two closed intervals $I_1,I_2\subseteq\R$. The support $[a,b]$ has a close relationship with the shape constraints satisfied by sequences $m \in \mathscr{M}_\infty([a,b])$. When $[a,b]=[0,1]$, $\mathscr{M}_\infty ([0,1])$ is the space of completely monotone sequences. In general, the true autocovariance $\gamma$ does not belong to $\mathscr{M}_\infty ([0,1])$, but does belong to $\mathscr{M}_\infty([-1,1])$. Additionally, for a geometrically ergodic chain, Proposition~\ref{prop:gamma_moment_seq} shows $\gamma\in\mathscr{M}_{\infty}([-1+\delta,1-\delta])$ for any $\delta\ge 0$ such that $\delta \le \delta_0$, where $\delta_0$ is the spectral gap of $Q$ in Proposition \ref{prop:gamma_moment_seq}. Throughout the remainder of the paper, we will consider projections onto the set $\mathscr{M}_\infty([-1+\delta,1-\delta])$, and thus we let $\mathscr{M}_\infty(\delta) = \mathscr{M}_\infty([-1+\delta,1-\delta])$ for notational simplicity.
		
		Now we define the moment least squares estimator $\Pi_\delta(\rinit)$ resulting from an initial autocovariance sequence estimator $\rinit\in\ell_2(\mathbb{Z})$ by
		\begin{align}\label{eq:momentLS}
			\Pi_\delta(\rinit) &=\underset{m\in \Mld{\delta}}{\arg\min}\|\rinit-m\|^2 \\
			&=\underset{m\in \Mld{\delta}}{\arg\min}\sum_{k\in\mathbb{Z}}\{\rinit(k)-m(k)\}^2.\nonumber
		\end{align}
		Note that $\Pi_\delta(\rinit)$ is the closest \textit{moment} sequence with respect to some measure supported on $\ralpha$ to the input autocovariance sequence $\rinit$, with respect to the $\ell_2$ norm $\|\cdot\|$ on $\ell_2(\mathbb{Z})$. 
		This optimization problem can be formulated as a convex quadratic problem, which we discuss further in Section \ref{sec:computation_mLSE}.

		The optimization problem \eqref{eq:momentLS} has one hyperparameter $\delta$, which needs be chosen sufficiently small so that the true autocovariance sequence $\gamma$ is a feasible solution, in the sense that $\gamma\in \Mld{\delta}$, of the optimization problem \eqref{eq:momentLS}. Note that any value of $\delta$ such that $0 \le \delta\le \delta_\gamma$ makes $\gamma$ feasible for $\delta_\gamma = 1-\sup \{|x|; x \in \textrm{Supp}(F)\}$ where $F$ is the representing measure for $\gamma$.
		Empirically, choosing $\delta$ as large as possible subject to $\delta \le \delta_\gamma$ leads to better performance because, roughly speaking, larger $\delta$ corresponds to more shape regularization. However, the method appears to work for a wide range of $\delta$ as long as $\delta$ is chosen to be positive (see Section \ref{sec:emp} for details). We also propose a practical choice of $\delta$ in Section \ref{sec:emp}. Theoretically, we showed the consistency of the proposed estimator $\Pi_\delta(\rinit)$ for any $0 < \delta \le \delta_\gamma$.
		
		For the choice of the initial autocovariance sequence estimator, any estimator $\rinit$ from a Markov chain sample $X_0,X_1,...,X_{M-1}$ of size $M$ satisfying 
		\begin{enumerate}[label=(R.\arabic*)]
			\item \label{cond:R1}(a.s. elementwise convergence) $\rinit(k) \underset{M\to\infty}{\to} \gamma(k)$ for each $k\in \Z$, $P_x$-almost surely, for any initial condition $x \in \mathsf{X}$,
			\item \label{cond:R2}(finite support) $\rinit(k) =0 $ for $k\ge n(M)$ for some $n(M)<\infty$, and
			\item \label{cond:R3}(even function with a peak at 0) $\rinit(k) = \rinit(-k)$ and $\rinit(0) \ge |\rinit (k)|$ for each $k\in \Z$,
		\end{enumerate}
		is allowed. As we demonstrate in Proposition \ref{prop:initial_estimators}, the empirical autocovariance sequence $\tilde{r}_M$ satisfies assumptions~\ref{cond:R1}--\ref{cond:R3}. In addition,~\ref{cond:R1}--\ref{cond:R3} are satisfied by any windowed empirical autocovariance sequence $\check{r}_M$ of the form $\check{r}_M(k) = \tilde{r}_M(k) w_M(|k|)$, where $w_M(k)$ is any window function which meets the following conditions \ref{cond:W1} - \ref{cond:W3}:
		\begin{enumerate}[label=(W.\arabic*)]
			\item \label{cond:W1} $w_M(0)= 1$ for all $M \in \N$, 
			\item \label{cond:W2}$|w_M(k)| \le 1$ for all $k \in \N$ and $M \in \N$, 
			\item \label{cond:W3} $w_M(k) \to 1$ for any fixed $k$ as $M\to\infty$
		\end{enumerate}
		In particular, conditions \ref{cond:W1} - \ref{cond:W3} are satisfied for some widely used window functions such as the simple truncation window $w_M(k) = I(k < b_M)$ and the Parzen window function $w_M(k) = [1 - k^q / b_M^q] I(k < b_M)$ for $q \in \{1,2,3,\dots\}$, which is the modified Bartlett window when $q=1$.
		
		In the following subsection, we provide some results relating to moment sequences, and provide an alternative characterization of moment sequences in relation to complete monotonicity.
		
		\subsection{Characterization of $[a,b]$-moment sequences}
		
		While $\gamma$ is not completely monotone when the support of the representing measure for $\gamma$ is not contained in $[0,1]$, it still exhibits infinitely many constraints. Previous studies have provided characterizations of $[a,b]$-moment sequences \citep{nla.cat-vn2337563, chandler1988moment}. Specifically, an $[a,b]$-moment sequence $m$ can be characterized equivalently by the non-negativity of a specific family of quadratic forms derived from $m$, $a$, and $b$ (e.g., Theorem 3.13 in \citealp{schmudgen2017moment}). 
		
		In Proposition \ref{prop:momentSeq}, we present an alternative characterization for an $[a,b]$-moment sequence $m$ in terms of the complete monotonicity of a transformed sequence $T(m;[a,b])$. It is important to note that while Proposition \ref{prop:momentSeq} gives insights on which (infinite number of) constraints are imposed on an estimator at the sequence level by requiring the estimator to be in the $[a,b]$-moment space $\mathscr{M}_{\infty}([a,b])$, the actual enforcement of these constraints is achieved through a mixture representation as in \eqref{eq:momentRep}. It is also technically convenient to have this alternative characterization for $[a,b]$-moment sequences because we can, e.g., verify that a sequence is an $[a,b]$-moment sequence by checking whether $T(m;[a,b])$ is completely monotone, and guarantee the uniqueness of the representing measures of $[a,b]$-moment sequences based on Theorem \ref{thm:hausdorff_moment}.
		
		For a sequence $m=\{m(k)\}_{k=0}^{\infty}$ and constants $a<b$, we define $T:\R^{\N} \to \R^{\N}$ as follows: 
		\begin{align}\label{eq:def_T}
			T(m;a,b)(k) = (b-a)^{-k}\sum_{i=0}^{k}\binom{k}{i}m(i)(-a)^{k-i}, \quad k=0,1,2,\dots.
		\end{align}
		Note $T(m;a,b)(0) = m(0)$, and when $a=0, b=1$, we have $T(m;0,1) = m$.
		
		\begin{prop}[\mbox{$[a,b]$}-moment sequences]\label{prop:momentSeq}
			For a sequence $m=\{m(k)\}_{k=0}^{\infty}$ and $a<b$, there exists a representing measure $\mu$ for $m$ supported on $[a,b]$ if and only if the sequence $T(m; a,b)$ is completely monotone. Additionally, if $T(m; a,b)$ is completely monotone, then the representing measure for $m$ is unique. 
		\end{prop}
		The proof of Proposition \ref{prop:momentSeq} is deferred to Supplementary Material~S4.1~\citep{berg2023efficientsuppl}. 
		Since throughout this paper we will consider sequences $m=\{m(k)\}_{k=-\infty}^{\infty}$ satisfying the symmetry relation $m(k)=m(-k)$ for each $k\in\mathbb{Z}$, we state the following corollary.
		
		\begin{cor}\label{cor:moment_thm_Z}
			Consider a sequence $m=\{m(k)\}_{k\in\Z}$ which is symmetric around $0$, i.e., $m(k) = m(-k)$ for $k\in\Z$. Additionally, consider $a,b\in\mathbb{R}$ with $a<b$. Then there exists a measure $\mu$ supported on $[a,b]$ such that $m(k)=\int x^{|k|} \mu(dx)$ for all $k\in \mathbb{Z}$ if and only if the sequence $ T( \{m(k)\}_{k\in\N} ; a,b)$ is completely monotone. Additionally, if  $ T( \{m(k)\}_{k\in\N} ; a,b)$  is completely monotone, then the measure corresponding to $m$ is unique. 
		\end{cor}

		\subsection{Properties of the moment least squares estimator\label{sec:projection}}
		
		The moment least squares estimator (moment LSE) $\Pi_\delta(\rinit)$ from an initial autocovariance sequence estimator $\rinit$ involves a projection from $\ell_2(\mathbb{Z})$ to $\Mld{[-1+\delta,1-\delta]}$. In this section, we show the existence and uniqueness of projections from $\ell_2(\mathbb{Z})$ to a moment sequence space 
		$\Mld{C}$ where $C \subseteq [-1,1]$ is a closed set. For an $r \in \ell_2(\Z)$, define $\Pi(r;C)$ be the projection of  $r$ onto $\Mld{C}$. We present a variational characterization of the projection $\Pi(r;C)$ . Finally, we obtain results on the properties of the representing measure of $\Pi(r;C)$ . Namely, we show that for fixed sample size $M$, if $r(k)=0$ for $k\ge n(M)$ for some $n(M)<\infty$, the representing measure $\hat{\mu}_{C}$ corresponding to $\Pi(r;C)$ is discrete, with finite support set ${\rm Supp}(\hat{\mu}_{C})$ having cardinality $|{\rm Supp}(\hat{\mu}_{C})|\leq n_0$, where $n_0$ is the smallest even number with $n_0>(n(M)-1)$.  Similar discreteness and finite support set results appear in the setting of nonparametric maximum likelihood estimation for mixture models, as in~\cite{lindsay1983geometry}, as well as in the least-squares estimation of a $k$-monotone or completely monotone pmf as in~\citet{Giguelay2017-sy} and \citet{Balabdaoui2020-nn}.

		First of all, for any closed $C \subset [-1,1]$, we show that $\Mld{C}$ is a closed and convex subset of $\sqs$. Then, since $\sqs$ is a Hilbert space equipped with the inner product $\braket{v,u} = \sum_{k\in\Z} v(k)u(k)$, we obtain by the Hilbert space projection theorem the existence and uniqueness of projections from $\ell_2(\mathbb{Z})$ to $\Mld{C}$.
		\begin{prop}\label{prop:existence_rhat}
			For any closed $C\subseteq [-1,1]$, the set $\Mld{C}$ is a closed, convex subset of $\sqs$. In particular, for any given vector $r\in\sqs$, $\Pi(r;C)$ exists and is unique in $\Mld{C}$.	
		\end{prop}
		Note that for any $M$, an initial input autocovariance sequence $\rinit$ satisfying \ref{cond:R1} -\ref{cond:R3} is in $\ell_2(\Z)$ since $\rinit(k)=0$ for $|k|\ge n(M)$, and therefore, the moment LSE $\Pi_\delta(\rinit) = \Pi(\rinit;[-1+\delta,1-\delta])$ is well defined. In addition, the optimization problem \eqref{eq:momentLS} is convex. The proof of Proposition \ref{prop:existence_rhat} uses the alternative characterization in Corollary \ref{cor:moment_thm_Z}  of an $[a,b]$-moment sequence and is deferred to Supplementary Material S4.2~\citep{berg2023efficientsuppl}. 
		
		Next, we present a few results regarding the projection $\Pi(r;C)$ of $r \in \sqs$ onto $\Mld{C}$. Proposition~\ref{prop:rhat_variational} provides a variational characterization of the projection $\Pi(r;C)$.
		\begin{prop} \label{prop:rhat_variational}
			Let $C$ be a closed subset of $[-1,1]$, and suppose $r\in\sqs$. 
			Then for $f\in\Mld{C}$, we have $f=\Pi(r;C)$ if and only if 
			\begin{enumerate}
				\item for all $\alpha\in C\cap (-1,1)$, $\braket{f,x_\alpha} \ge \braket{r,x_\alpha}$, i.e.,
				\begin{align}\label{eq:prop:rhat_variational}
					\sum_{k=-\infty}^{\infty}f(k)\alpha^{|k|}\geq \sum_{k=-\infty}^{\infty}r(k)\alpha^{|k|},
				\end{align}
				\item  $\braket{f,f} = \braket{f,r}$, i.e., $\sum_{k=-\infty}^{\infty}f(k)^2=\sum_{k=-\infty}^{\infty}f(k)r(k)$.
			\end{enumerate}
		\end{prop}
		A similar characterization of $\Pi(r;C)$ was also presented in \citet{Balabdaoui2020-nn}. We omit the proof as the result can be obtained by a minor modification of Proposition 2.2 in \citet{Balabdaoui2020-nn}.
		
		Proposition~\ref{prop:inner_product_r} below shows that \eqref{eq:prop:rhat_variational} holds with equality for $\alpha$ in the support of the representing measure for $\Pi(r;C)$ with $|\alpha|<1$.
		
		\begin{prop}\label{prop:inner_product_r}
			Let $C$ be a closed subset of $[-1,1]$, and suppose $r\in\sqs$. 
			Let $\hat{\mu}_{C}$ denote the representing measure for $\Pi(r;C)$. Then for each $\alpha\in {\rm Supp}(\hat{\mu}_C)\cap (-1,1)$, we have\begin{align*}
				\braket{\Pi(r;C), x_{\alpha}}=\braket{r, x_{\alpha}}.
			\end{align*}
		\end{prop}
		The proof for Proposition \ref{prop:inner_product_r} essentially follows from Proposition \ref{prop:rhat_variational}, as we have $\braket{\Pi(r;C) - r, x_\alpha} \ge 0$ for all $\alpha\in C\cap (-1,1)$ and from the second condition in Proposition \ref{prop:rhat_variational}
		\begin{align*}
			\int \braket{\Pi(r;C) - r, x_\alpha} \m_C(d\alpha)  =\braket{\Pi(r;C), \Pi(r;C)} -  \braket{r, \Pi(r;C)} =0,
		\end{align*}
		which implies $\braket{\Pi(r;C) - r, x_\alpha}=0$, for $\m_C$-almost every $\alpha$. We show that this implies that $\braket{\Pi(r;C) - r, x_\alpha}=0$ for all $\alpha\in {\rm Supp}(\hat{\mu}_C)\cap (-1,1)$. The details are deferred to Supplementary Material S4.3~\citep{berg2023efficientsuppl}.
		
		Finally, we show that for an input sequence $r$ with finite support, i.e., $r(k)=0$ for $|k| \ge M$ for some $M$, then the representing measure for the projection $\Pi(r;C)$ is discrete, and the support of the representing measure contains at most a finite number of points. More concretely, we have the following result:
		
		\begin{prop}\label{prop:finiteSupport}
			Let $C$ be a closed subset of $[-1,1]$, and suppose $r\in\sqs$ satisfies $r(k)=0$ for all $k$ with $|k|>M-1$ for $M<\infty$. Let $\Pi(r;C)$ denote the projection of $r$ onto $\Mld{C}$. Let $\m_C$ denote the representing measure for $\Pi(r;C)$. Then ${\rm Supp}(\hat{\mu}_C)$ contains at most $n$ points, where $n$ is the smallest even number such that $n>(M-1)$. Additionally, the support of $\m_C$ is contained in $(-1,1)$, that is, ${\rm Supp}(\m_C) \cap \{-1,1\} = \emptyset$.
		\end{prop}
		The proof follows similar lines as in~\citet{Balabdaoui2020-nn}, but requires nontrivial modification to deal with the possible support of $\hat{\mu}_{C}$ in $[-1,0)$. We defer the proof to Supplementary Material S4.4~\citep{berg2023efficientsuppl}. In particular, a moment LSE $\Pi_\delta(\rinit)$ for any initial estimator $\rinit$ satisfying condition \ref{cond:R2} has a representing measure which is discrete and has support containing at most $n_0$ points, where $n_0$ is the smallest even number such that $n_0>\{n(M)-1\}$. The representing measure for an arbitrary element of $f \in \Ml$ is in general neither finitely supported nor discrete. Thus Proposition~\ref{prop:finiteSupport} provides a considerable simplification of the form of the representing measure of $\prinit$.
		
		\subsection{Computation of the moment least squares estimator}\label{sec:computation_mLSE}
		
		Recall that $\Pi_\delta(\rinit)$ is the minimizer $m$ of $\sum_{k\in\mathbb{Z}}\{\rinit(k)-m(k)\}^2$ such that $m\in \Mld{\delta}$. By Proposition \ref{prop:finiteSupport}, since $\hat{\mu}_\delta$ is a discrete measure, we have
		\begin{align*}
			\Pi_\delta(\rinit)(k) = \int \alpha^{|k|} \hat{\mu}_\delta (d \alpha) = \sum_{\alpha \in {\rm Supp}(\hat{\mu}_\delta)} \alpha^{|k|} \hat{\mu}_\delta(\{\alpha\}).
		\end{align*}
		For a closed set $\Theta \subseteq [-1+\delta,1-\delta]$, recall $\Pi(r; \Theta)$ is the projection of $r$ to the set of $\ell_2(\Z)$ moment sequences with representing measure supported on $\Theta$. Note we have $\Pi_\delta(r) = \Pi(r; [-1+\delta,1-\delta]) = \Pi(r; \Theta_0)$ for any $\Theta_0$ such that ${\rm Supp}(\hat{\mu}_\delta) \subseteq \Theta_0\subseteq [-1+\delta,1-\delta]$. 
		
		For a finite $\Theta = \{\alpha_1,\dots,\alpha_s\}\subset (-1,1)$, $\Pi(r; \Theta)$ can be computed by solving a simple convex quadratic program. For $m\in\mathscr{M}_{\infty}(\Theta)\cap\ell_2(\mathbb{Z})$, the least squares objective in \eqref{eq:momentLS} becomes
		\begin{align}
			\sum_{k\in\Z} (\rinit(k)-m(k))^2 
			&= \sum_{k\in\Z} (\rinit(k)-\sum_{i=1}^s \alpha_i^{|k|} w_i)^2,
			\label{eq:opt_problem_eq1}
		\end{align}
		where we define $w_i = \mu_m(\{\alpha_i\})$ for $i=1,\dots,s$ and $s=|\Theta|$ where $\mu_m$ denotes the representing measure for $m$. Define $\mathbf{w}=[w_1,\dots,w_s]\in\R^s$. Define $\mathbf{a}=[a_1,\dots,a_s] \in \R^s$ such that
		$a_i = \sum_{k\in\mathbb{Z}} \alpha_i^{|k|}\rinit(k) = \sum_{k; \rinit(k)\ne 0} \alpha_i^{|k|}\rinit(k)$ and $\mathbf{B} \in \R^{s\times s}$ such that $\mathbf{B}_{ij} = \frac{1+\alpha_i \alpha_j}{1-\alpha_i\alpha_j}$. Note that $\mathbf{B}$ can be computed easily based on $\Theta$ and $\mathbf{a}$ can be computed easily based on $\Theta$ and $\rinit$ when $\rinit$ satisfies~\ref{cond:R2}.
		Then with some algebra, we can show that $\sum_{k\in\Z} (\rinit(k)-m(k))^2 = {\rinit}^\top \rinit -2 \mathbf{a}^\top \mathbf{w} + \mathbf{w}^\top \mathbf{B} \mathbf{w}$ (see Supplementary Material S1 of~\citealp{berg2023efficientsuppl}). Therefore the optimization problem becomes
		\begin{equation}
			\begin{split}
				&\min_{\mathbf{w}} \quad  {\rinit}^\top \rinit -2 \mathbf{a}^\top \mathbf{w} + \mathbf{w}^\top \mathbf{B} \mathbf{w}\\
				&\textrm{subject to} \quad \mathbf{w}\geq 0 
			\end{split}
			\label{eq:opt_problem_eq2}    
		\end{equation}
		which is a quadratic convex problem because $\mathbf{B}$ can be shown to be a positive definite matrix (Supplementary Material S1 in~\citealp{berg2023efficientsuppl}). Note that this objective is identical to the quadratic programming formulation of the non-negative least squares problem.
		
		For computing $\Pi_{\delta}(\rinit)$, in practice, we approximate the interval $[-1+\delta,1-\delta]$ with a finely spaced finite grid of $s$ points $\Theta=\{\alpha_1,...,\alpha_s\}\subseteq[-1+\delta,1-\delta]$. We then approximate the solution $\Pi_{\delta}(\rinit) = \Pi(\rinit; [-1+\delta,1-\delta])$ by $\Pi(\rinit; \Theta)$. Of course, if $\Theta$ contains the support of $\hat{\mu}_\delta$, then $\Pi_{\delta}(\rinit) =\Pi(\rinit; \Theta)$. We used a grid of $s=1001$ $\alpha$ values in $[-1+\delta, 1-\delta]$, where we first created an equally spaced grid $\mathcal{G}$ in a log-scale from $[0,1-\delta]$ and used $\mathcal{S} = -\mathcal{G}\cup \mathcal{G}$. We used the support reduction algorithm by \citet{groeneboom2008support} (ref. page 388) to solve \eqref{eq:opt_problem_eq2} with this choice of $\Theta$. In terms of run-time of our implementation, it took about $.056$ seconds on average to obtain $\Pi_\delta(\rinit)$ for $\rinit$ from a length $M=10000$ AR1 chain and the choice of grid above, on an author's typical personal laptop operating Mac OS with a 3.2 GHz processor. The implementation is available in \texttt{https://github.com/hsong1/momentLS}.

		\section{Statistical guarantee of the moment LS estimator\label{sec:guarantee}}
		
		In this section, we analyze the statistical performance of the moment LS estimator. Specifically, we show that the moment least squares estimator $\prinit$ obtained from any eligible initial autocovariance sequence estimator $\rinit$ satisfying \ref{cond:R1}-\ref{cond:R3} is $\ell_2$-strongly consistent for the true autocovariance sequence, and the asymptotic variance estimate based on $\prinit$ is strongly consistent for the true asymptotic variance $\sigma^2$ in \eqref{eq:clt}. 
		
		First, the following Proposition shows that a wide range of estimators are allowed for the choice of the initial autocovariance sequence estimator $\rinit$, including the empirical autocovariance estimator as well as windowed autocovariance estimators.
		
		\begin{prop}\label{prop:initial_estimators}
			Assume that a Markov chain $X=\{X_0,X_1,\dots,\}$ with transition kernel $Q$ satisfies conditions~\ref{cond:harris_ergodicity}-\ref{cond:geometric_ergodicity}, and the function of interest $g$ is in $L^2(\pi)$.
			The empirical autocovariance sequence $\tilde{r}_M$, defined as in~\eqref{eq:emp1_example}, satisfies conditions \ref{cond:R1}-\ref{cond:R3} where $Y_M = M^{-1}\sum_{t=0}^{M-1}g(X_t)$. In addition, any windowed autocovariance sequence estimator $\check{r}_M$ such that $\check{r}_M (k) = \tilde{r}_M(k) w_M(|k|)$ for any window function $w_M$ satisfying \ref{cond:W1}-\ref{cond:W3} satisfies \ref{cond:R1}-\ref{cond:R3}.
		\end{prop}
		The proof is deferred to S5.1 in the Supplementary Material~\citep{berg2023efficientsuppl}.
		
		\subsection{L2 consistency of the moment LSE}
		
		We now show the strong consistency (with respect to the $\ell_2$ metric) of the moment LSE $\prinit$ for the true autocovariance sequence, that is, we show $\|\prinit-\gamma\|\overset{a.s.}{\to}0$, for any $\delta>0$ satisfying ${\rm Supp}(F) \subseteq \ralpha$.
		
		First of all, we present the following key lemma, which bounds the $\ell_2$ distance between the projection $\Pi_\delta(r)$ of $r\in\ell_2(\mathbb{Z})$, and an element $\gamma$ in $\Ml$, with a mixture of geometrically weighted differences between the input $r$ and $\gamma$. This lemma plays a crucial role in our convergence analysis. In our setting, the standard bound derived from the property of the projection
		\begin{align*}
			\|\Pi_\delta(\rinit)-\gamma\|^2\leq \|\rinit-\gamma\|^2
		\end{align*}
		for $\gamma\in \Ml$ is not helpful because we do not assume the consistency, with respect to the $\ell_2$ metric, of $\rinit$ for the true autocovariance $\gamma$. In fact, the empirical autocovariance sequence seems not to converge to $\gamma$ in the $\ell_2$ sense.  Even so, we can still show that a geometrically weighted difference between $\rinit$ and $\gamma$ converges to $0$, which leads to the convergence of $\Pi(\rinit)$ to $\gamma$ in the $\ell_2$ sense.
		
		\begin{lem}\label{lem:l2diff_bound}
			Suppose $\delta\in[0,1]$, and let $f\in\Ml$. Additionally, suppose $r\in\sqs$. Then
			\begin{align}
				0\leq \|\Pi_\delta(r)-f\|^2\leq -\int \braket{x_{\alpha},r-f}\mu_f(d\alpha) +\int \braket{x_{\alpha},r-f}\m_\delta(d\alpha).
			\end{align}
			where $\m_\delta$ is the representing measure for $\Pi_\delta(r)$ and $\mu_f$ is the representing measure for $f$.
		\end{lem}
		\begin{proof}
			Clearly $0\leq \|\Pi_\delta(r)-f\|^2$. We have,
			\begin{align}\label{eq:lem2-eq1}
				\|f-\Pi_\delta(r)\|^2=\braket{f,f}-2\braket{\Pi_\delta(r),f}+\braket{\Pi_\delta(r),\Pi_\delta(r)}.
			\end{align}
			First, for the third term in \eqref{eq:lem2-eq1}, by Proposition \ref{prop:inner_product_r} and Lemma 4 in the Supplementary Material~\citep{berg2023efficientsuppl}, we have
			\begin{align*}
				\braket{\Pi_\delta(r),\Pi_\delta(r)}&=\int\braket{x_{\alpha},\Pi_\delta(r)}\m_\delta(d\alpha)\\
				&=\int\braket{x_{\alpha},r}\m_\delta(d\alpha)\\
				&=\int\{\braket{x_{\alpha},r-f}+\braket{x_{\alpha},f}\}\m_\delta(d\alpha)\\
				&=\int\braket{x_{\alpha},r-f}\m_\delta(d\alpha)+\braket{\Pi_\delta(r),f},
			\end{align*} where the second equality follows from $\braket{x_{\alpha},\Pi_\delta(r)}=\braket{x_{\alpha},r}$ for all $\alpha\in {\rm Supp}(\m)$. Thus, \eqref{eq:lem2-eq1} becomes,
			\begin{align*}
				\|f-\Pi_\delta(r)\|^2=\braket{f,f}-\braket{\Pi_\delta(r),f}+\int\braket{x_{\alpha},r-f}\m_\delta(d\alpha).
			\end{align*}
			Now, for the second term in~\eqref{eq:lem2-eq1},
			\begin{align*}
				\braket{\Pi_\delta(r),f}&=\int\braket{x_{\alpha},\Pi_\delta(r)}\mu_f(d\alpha)\\
				&\geq \int\braket{x_{\alpha},r}\mu_f(d\alpha)\\
				&=\int\{\braket{x_{\alpha},r-f}+\braket{x_{\alpha},f}\}\mu_f(d\alpha)\\
				&=\int\braket{x_{\alpha},r-f}\mu_f(d\alpha)+\braket{f,f}.
			\end{align*} 
			where for the second inequality we use Proposition \ref{prop:inner_product_r} which states $\braket{x_{\alpha},\Pi_\delta(r)}\ge\braket{x_{\alpha},r}$ for all $\alpha \in [-1+\delta,1-\delta]\cap(-1,1)$, as well as Lemma 2 in the Supplementary Material~\citep{berg2023efficientsuppl}.
			Therefore, we obtain,
			\begin{align*}
				\|f-\Pi_{\delta}(r)\|^2 \leq -\int \braket{x_{\alpha},r-f}\mu_f(d\alpha)+\int \braket{x_{\alpha},r-f} \m_\delta(d\alpha).
			\end{align*}
		\end{proof}
		
		The next two propositions, Proposition \ref{prop:xalpha_conv} and \ref{prop:finite_muhat}, serve as the basis for proving the moment LS estimator's $\ell_2$ consistency by proving the uniform convergence of the geometrically weighted difference between $\rinit$ and $\gamma$ and the finiteness of the representing measure of $\Pi(\rinit)$. 
		\begin{prop}\label{prop:xalpha_conv}
			Let $\rinit$ denote an initial autocovariance sequence estimator satisfying \ref{cond:R1}-\ref{cond:R3}. Let $\mathcal{K}$ denote a nonempty compact set with $\mathcal{K}\subseteq (-1,1)$. Then we have
			\begin{align}
				\sup_{\alpha \in \mathcal{K}} |\braket{\rinit-\gamma, x_\alpha}| \to 0 \quad P_x\mbox{-almost surely,}
			\end{align}
			as $M\to\infty$, for each initial condition $x\in\mathsf{X}$.
		\end{prop}

		\begin{prop}\label{prop:finite_muhat}
			For a given $\delta>0$ and an initial autocovariance sequence estimator $\rinit$ satisfying \ref{cond:R1}-\ref{cond:R3}, let $\m_{\delta,M}$ denote the representing measure for $\prinit$. Then there exists a constant $C_{\delta,\gamma}<\infty$ with $C_{\delta,\gamma}$ depending only on $\gamma$ and $\delta$ such that \begin{align*}
				\underset{M\to\infty}{\lim\sup}\;\hat{\mu}_{\delta,M}([-1+\delta,1-\delta])\leq C_{\delta,\gamma}
			\end{align*} $P_x$-almost surely for any $x\in\mathsf{X}$. In particular, $\hat{\mu}_{\delta,M}([-1+\delta,1-\delta])$ remains bounded almost surely.
		\end{prop}
		
		The proofs for Propositions \ref{prop:xalpha_conv} and \ref{prop:finite_muhat} are in Supplementary Material S5.2 and S5.3~\citep{berg2023efficientsuppl}.
		Finally, we present the main result of this section in Theorem~\ref{thm:as_l2conv} below, which shows that the moment LSE is $\ell_2$ consistent for the true autocovariance sequence $\gamma$. This result is the consequence of the key inequality in Lemma \ref{lem:l2diff_bound} as well as the uniform convergence of $\braket{x_\alpha, \rinit-\gamma}$ and finiteness of the representing measure of $\Pi_\delta(\rinit)$ in Proposition \ref{prop:xalpha_conv} and \ref{prop:finite_muhat}.
		
		\begin{thm}[$\ell_2$-consistency of Moment LSEs]\label{thm:as_l2conv} Suppose $X_{0},X_1,...,$ is a Markov chain with transition kernel $Q$ satisfying~~\ref{cond:harris_ergodicity}-\ref{cond:geometric_ergodicity}, and suppose $g:\mathsf{X}\to\mathbb{R}$ satisfies~\ref{cond:integrability}. Let $\gamma$ denote the autocovariance sequence as defined in Proposition~\ref{prop:gamma_moment_seq}, and let $F$ denote the representing measure for $\gamma$. Suppose $\delta>0$ is chosen so that $F$ is supported on $\ralpha$. Let $\rinit$ be an initial autocovariance sequence estimator satisfying conditions \ref{cond:R1} - \ref{cond:R3}. 
			Then \begin{align*}
				\|\gamma-\prinit\|^2 \underset{M\to\infty}{\to} 0, \,\,\,P_x\mbox{-a.s.}
			\end{align*}
			for each initial condition $x\in \mathsf{X}$.
		\end{thm}
		
		\begin{proof}
			From Proposition \ref{prop:gamma_moment_seq} and by the choice of $\delta$, we have $\gamma \in \mathscr{M}_\infty(\delta)$ for $\delta>0$. Then Lemma 3 in the Supplementary Material gives that $\gamma \in \ell_1(\Z)$, and therefore $\gamma \in \ell_2(\Z)$. Thus, we have $\gamma\in \Mld{\delta}$. Additionally, $\rinit\in \ell_{2}(\mathbb{Z})$ since $\rinit$ satisfies~\ref{cond:R2}. Therefore, we can apply the result of Lemma \ref{lem:l2diff_bound}, and we have the following inequality
			\begin{align*}
				\|\gamma-\prinit\|^2 
				&\leq -\int_{[-1,1]} \braket{x_{\alpha},\rinit-\gamma}F(d\alpha)+\int_{[-1,1]} \braket{x_{\alpha},\rinit-\gamma}\m_{\delta,M}(d\alpha).
			\end{align*}
			where $\m_{\delta,M}$ is the representing measure for $\prinit$. Note ${\rm Supp}(F) \subseteq \ralpha$ by the assumption on $\delta$. Additionally, ${\rm Supp}(\m_{\delta,M}) \subseteq \ralpha$ since $\prinit\in \Ml$. Therefore, we have for any $M$
			\begin{align*}
				\|\gamma-\prinit\|^2 
				&\leq \left(\sup_{\alpha \in [-1+\delta, 1-\delta]} |\braket{x_{\alpha},\rinit-\gamma}| \right)\{F([-1,1]) + \m_{\delta,M}([-1,1])\}
			\end{align*}
			and thus
			\begin{align*}
				\limsup_{M\to\infty} \|\gamma-\prinit\|^2 &\leq \limsup_{M\to\infty} \left(\sup_{\alpha \in [-1+\delta, 1-\delta]} |\braket{x_{\alpha},\rinit-\gamma}| \right)F([-1,1]) \\
				&+ \limsup_{M\to\infty} \left(\sup_{\alpha \in [-1+\delta, 1-\delta]} |\braket{x_{\alpha},\rinit-\gamma}| \right) \limsup_{M\to\infty} \m_{\delta,M}([-1,1]).
			\end{align*}
			Let the initial condition for the chain $x\in \mathsf{X}$ be given.
			From Proposition \ref{prop:xalpha_conv}, we know that $\left(\sup_{\alpha \in [-1+\delta, 1-\delta]} |\braket{x_{\alpha},r-\gamma}| \right)\to 0$ $P_x$-a.s. Also we have $F([-1,1]) =\gamma(0) <\infty$ and $ \limsup_{M\to\infty} \m_{\delta,M}([-1,1]) \leq C_{\delta,\gamma}<\infty$ $P_x$-a.s. from Proposition \ref{prop:finite_muhat}. Therefore, we have $\limsup_{M\to\infty} \|\gamma-\prinit\|^2  =0$ $P_x$-almost surely. Thus, $\|\gamma -\prinit\|^2 \to 0$ $P_x$-almost surely as $M\to \infty$, as desired.
		\end{proof}
		
		An important consequence of Proposition \ref{prop:xalpha_conv}, \ref{prop:finite_muhat}, and Theorem \ref{thm:as_l2conv} is the measure convergence of $\hat{\mu}_{\delta,M}$ to the true representing measure $F$. Recall that for a sequence of measures $\{\nu_n\}_{n\in \N}$ and $\nu$ on $\R$, $\nu_n$ \textit{converges vaguely} to $\nu$ if and only if $\int f d\nu_n \to \int f d\nu$ for all  $f \in C_0(\R)$ [e.g., \citealp{folland1999real}], where $C_0(\R)$ is the space of continuous functions that vanish at infinity, i.e. $f \in C_0(\R)$ iff $f$ is continuous and the set $\{x; |f(x)| \ge \epsilon\}$ is compact for every $\epsilon>0$.
		
		\begin{prop}[vague convergence of $\m_{\delta,M}$]\label{prop:vague_convergence_muhat}
			Assume the same conditions as in Theorem \ref{thm:as_l2conv}. For each initial condition $x\in \mathsf{X}$, we have $P_x(\m_{\delta,M} \to F \mbox{ vaguely, as }M\to\infty)=1$, where $\m_{\delta,M}$ and $F$ are the representing measures for $\prinit$ and $\gamma$, respectively.
		\end{prop}
		This proposition is a direct consequence of the a.s. $\ell_2$ convergence of $\Pi_\delta(\rinit)$ to $\gamma$ and Lemma 7 in the Supplementary Material S2~\citep{berg2023efficientsuppl}.

		\subsection{Strong consistency of the asymptotic variance estimator based on the moment LSE}
		
		In this subsection, we present the strong consistency result for the asymptotic variance estimator based on the moment least squares estimators. It is well known that for a stationary, $\psi$-irreducible, geometrically ergodic, and reversible Markov chain and for a square integrable $g$, the central limit theorem holds [e.g., see Corollary 6 in \citealp{Haggstrom2007-sy} ], i.e.,
		\begin{align}
			\sqrt{M} (Y_M - \mu) \overset{d}{\to}  N(0,\sigma^2(\gamma)), 
		\end{align} with
		\begin{align}\label{eq:sigma_equiv}
			\sigma^2(\gamma) = \lim_{M\to\infty} M {\rm Var} (Y_M)  = \sum_{k \in \Z} \gamma(k) = \int \frac{1+\alpha}{1-\alpha} F(d\alpha) <\infty, 
		\end{align} 
		where $F$ denotes the representing measure associated with $\gamma$.
		
		The main theorem for this subsection is Theorem \ref{thm:as_conv_avar}, which shows that an asymptotic variance estimate based on the moment least squares estimator $\sigma^2(\prinit) = \sum_{k \in \Z} \prinit(k)$ is strongly consistent for $\sigma^2(\gamma)$ for any $\rinit$ which satisfies conditions \ref{cond:R1} - \ref{cond:R3}.

		\begin{thm}[strong consistency of asymptotic variance estimators based on Moment LSEs]\label{thm:as_conv_avar}
			Assume the same conditions as in Theorem \ref{thm:as_l2conv}.
			Let $\sigma^2(\gamma) = \ \sum_{k\in \Z} \gamma(k)$ be the asymptotic variance based on the true autocovariance sequence $\gamma$. We let  $\sigma^2(\prinit) =  \sum_{k\in \Z} \prinit(k)$ be an estimate of $\sigma^2(\gamma)$ based on the moment least squares estimator $\prinit$. We have $\sigma^2(\prinit) \to \sigma^2(\gamma)$ $P_x$-a.s., for each initial condition $x \in \mathsf{X}$, as $M\to\infty$.
		\end{thm}

		\begin{proof}
			Let $\hat{\sigma}^2_M = \sigma^2(\prinit)$ and $\sigma^2 = \sigma^2(\gamma)$ for notational simplicity. Lemma 3 and Lemma~5 in the Supplementary Material give that $\hat{\sigma}^2_M = \int_\ralpha \frac{1+\alpha}{1-\alpha} \m_{\delta,M}(d\alpha)$, and we have $\sigma^2 =  \int_\ralpha \frac{1+\alpha}{1-\alpha} F(d\alpha)$ from \eqref{eq:sigma_equiv}. Thus, we have
			\begin{align*}
				|\hat{\sigma}^2_M - \sigma^2|= \left\lvert \int_\ralpha \frac{1+\alpha}{1-\alpha} \m_{\delta,M}(d\alpha) -  \int_\ralpha \frac{1+\alpha}{1-\alpha} F(d\alpha) \right\rvert .
			\end{align*}
			We can obtain $f(\alpha) \in C_0(\R)$ such that $f(\alpha) = \frac{1+\alpha}{1-\alpha}$ on $\ralpha$ by extending the two endpoints of $(1+\alpha)/(1-\alpha)$ at $\alpha \in \{ -1+\delta, 1-\alpha\}$ to $0$ linearly so that $f(\alpha)=0$ for $|\alpha|\ge 1$. More concretely, define $f:\mathbb{R}\to\mathbb{R}$ by
			\begin{align*}
				f(\alpha) = \begin{cases}
					\frac{1+\alpha}{1-\alpha}&\alpha \in \ralpha\\
					\frac{\alpha+1}{2-\delta} & -1 \le \alpha \le -1+\delta\\
					\frac{-(2-\delta)\alpha + 2-\delta}{\delta^2 } & 1-\delta \le \alpha \le 1 \\
					0 & \alpha <-1 \,\, {\rm or }\,\, \alpha >1.
				\end{cases}
			\end{align*} Then $f\in C_{0}(\R)$ and $f(\alpha)=\frac{1+\alpha}{1-\alpha}$ for $\alpha\in \ralpha$. Then, since ${\rm Supp}(\m_{\delta,M}), {\rm Supp}(F) \subseteq \ralpha$, we have \begin{align*}
				|\hat{\sigma}^2_M - \sigma^2|= \left\lvert \int f(\alpha) \m_{\delta,M}(d\alpha) -  \int  f(\alpha) F(d\alpha) \right\rvert \to 0,\;\;P_x\mbox{-a.s}.,
			\end{align*}
			for any $x \in \mathsf{X}$ by the almost sure vague convergence of $\m_{\delta,M}$ to $F$ in Proposition \ref{prop:vague_convergence_muhat}.
		\end{proof}

		\section{Empirical studies\label{sec:emp}}
		The goal of this section is two fold: first, we empirically illustrate some of the theoretical aspects discussed in the previous section, in particular, the $\ell_2$ sequence consistency and asymptotic variance consistency of Moment LSEs. Second, we compare the performance of our method to the performance of other current state-of-the-art methods for autocovariance sequence estimation and asymptotic variance estimation. In Section~\ref{subsec:hyperparameters}, we propose a method for tuning the hyperparameter $\delta$ for moment LSEs. In Section \ref{sec:emp_sim}, we use two simulation settings: one from a Metropolis-Hastings algorithm (\citet{metropolis1953equation} and~\citet{hastings1970monte}) with a discrete state space, and the other from a stationary AR(1) chain. In Section \ref{sec:emp_bayes}, we use a Bayesian probit regression.
		
		\subsection{Settings}
		\subsubsection{Settings for simulated chains}
		\paragraph*{Metropolis-Hastings chain}
		We consider a Metropolis-Hastings chain on the discrete state space $(\mathsf{X},2^{\mathsf{X}})$ where $\mathsf{X}=\{1,2,...,d\}$, so that $d=100$ states are possible. The stationary distribution for the simulation was constructed by normalizing a length $d$ random vector $U=[U_1,U_2,...,U_{d}]^T$ with $U_i\overset{iid}{\sim} {\rm Uniform}(0,1)$, so that $\pi(\{i\})=U_{i}/(\sum_{i'=1}^{d}U_{i'})$. Each row of the proposal distribution $P$ was constructed in the same manner, with $P(i,\{j\})=V_{ij}/\sum_{j=1}^{d}V_{ij}$ for random variables $V_{ij}\overset{iid}{\sim}{\rm Uniform}(0,1)$. The Metropolis-Hastings algorithm was used to construct a transition kernel $Q$ corresponding to the proposal distribution $P$. Finally, a function $g:\mathsf{X}\to\mathbb{R}$ was constructed via $g = [g_1,\dots,g_{d}]$ with $g_j \overset{iid}{\sim} N(0,1)$.  We generated a Markov chain $X_0,X_1,...$ with stationary distribution $\pi$ according to $Q$. 
		
		Since in this example, the transition kernel $Q$ is on a small discrete state space, it is possible to compute the eigenvalues $\lambda_i$ and eigenvectors $\phi_i$ corresponding to $\lambda_i$ for $i=1,\dots,d$ numerically. Therefore, the true autocovariance sequence $\gamma$, the representing measure $F$ for $\gamma$, and the asymptotic variance $\sigma^2(\gamma)$ can be all computed explicitly. More concretely, let $\lambda_1\geq \lambda_2\geq\cdots \geq\lambda_d$ denote the eigenvalues of $Q$. Suppose the eigenvectors $\phi_i$ are normalized so that $\braket{\phi_i,\phi_j}_\pi = 1[i=j]$. Note we have $\lambda_1=1$ and $\phi_1 = \mathbf{1}_{d}$ since $Q\mathbf{1}_{d} = \mathbf{1}_{d}$. We can write $g$ and $\bar{g} = g-E_\pi[g(X_0)]\mathbf{1}_{d}$ as
			$g(k) = \sum_{i=1}^d \braket{g,\phi_i}_\pi \phi_i(k)$ and
			$\bar{g}(k) = \sum_{i=2}^d \braket{g,\phi_i}_\pi \phi_i(k)$
		since $\braket{g,\phi_1}_\pi = E_\pi[g(X_0)]$. Then, 
			$\gamma(k)=\braket{Q_0^{|k|} g, g}_\pi = \braket{Q^{|k|}\bar{g}, \bar{g}}_\pi = \sum_{i=2}^{d}\braket{g,\phi_i}_\pi^2\lambda_i^{|k|}$,
		and thus the representing measure for $\gamma$ is 
		\begin{align}
			F=\sum_{i=2}^{d} \braket{g,\phi_i}_\pi^2\delta_{\lambda_i}\label{eq:F_discrete}
		\end{align} where $\delta_{a}$ denotes a unit point mass at $a$. Finally, we have
			$\sigma^2(\gamma) = \sum_{i=2}^{d} \braket{g,\phi_i}_\pi^2 \frac{1+\lambda_i}{1-\lambda_i}$.

			

		\paragraph*{Autoregressive chain} We also consider the autoregressive chain with the identity function $g(x)=x$ as in Example~\ref{exmp:ar1exmp}. We let $\tau^2=1$, and consider both positively and negatively correlated cases by setting $\rho=0.9$ and $\rho=-0.9$ in each case, respectively.

		\subsubsection{Descriptions of estimators}
		We investigated the following autocovariance sequence estimators:
		
		\begin{enumerate}
			\setlength\itemsep{0em}
			\item {\bfseries (Empirical)} the empirical autocovariance sequence $\{\tilde{r}_M(k)\}_{k\in\Z}$,
			\item {\bfseries (Bartlett)} the windowed empirical autocovariance sequence $\check{r}_M(k) = w_M(|k|) \tilde{r}_M(k)$ with $w_M(k) = (1-k/b^{\rm (Bart)}_{M}) I(k<b^{\rm (Bart)}_{M})$ with threshold $b^{\rm (Bart)}_{M}$, and 
			\item {\bfseries (MomentLS(Emp) and MomentLS(Bartlett))} our moment least squares estimators with the empirical autocovariance sequence $\Pi_\delta(\tilde{r}_M)$ and the windowed empirical autocovariance sequence $\Pi_\delta(\check{r}_M)$ as initial input sequences.
		\end{enumerate}
		For all three sequence estimators (Empirical, Bartlett, and MomentLS), asymptotic variance estimates were obtained by summing up the sequence estimators over all $k \in \Z$. In the case of Empirical and Bartlett estimators, this amounts to summing up the non-zero terms in the estimated autocovariance sequences $\tilde{r}_M$ or $\check{r}_M$. 
		For MomentLS estimators, for each input sequence $\rinit \in \{\tilde{r}_M, \check{r}_M\}$ and given $\delta>0$, the sequence estimates were computed following steps outlined in Section \ref{sec:computation_mLSE}. To elaborate further, we start by creating a grid $\Theta=\{\alpha_1,\dots,\alpha_s\} \subseteq [-1+\delta,1-\delta]$. We then solve the optimization problem \eqref{eq:opt_problem_eq2} to obtain $\hat{\mathbf{w}}=[\hat{\mu}_\delta(\{\alpha_1\}),\dots,\hat{\mu}_\delta(\{\alpha_s\})]^\top$. The momentLS sequence estimate for $\gamma(k)$ is $\Pi_{\delta}(\rinit;\Theta)(k) = \sum_{\alpha; \hat{\mu}_\delta(\{\alpha\})>0} \alpha^{|k|} \hat{\mu}_\delta(\{\alpha\})$. The asymptotic variance estimate is \begin{align*}
			\sigma^2(\Pi_\delta(\rinit;\Theta)) = \sum_{k\in\Z} \Pi_\delta(\rinit;\Theta) (k)= \sum_{\alpha; \hat{\mu}_\delta(\{\alpha\})>0} \frac{1+\alpha}{1-\alpha} \hat{\mu}_{\delta} (\{\alpha\}).
		\end{align*}
		The choice of $\delta$ is described in the next subsection \ref{subsec:hyperparameters}.
		
		For the comparison of asymptotic variance estimation performance, in addition to asymptotic variance estimates from the aforementioned estimators, we considered batch means, overlapping batch means, and initial sequence estimators. Let $Y_M=M^{-1}\sum_{t=0}^{M-1}g(X_t)$. For $i\leq M-b$, define the batch mean starting at $i$ with batch length $b$ by $Y_{b}(i)=b^{-1}\sum_{k=0}^{b-1}g(X_{i+k})$. Then the batch means, overlapping batch means, and initial sequence estimators are defined as
		
		\begin{enumerate}
			\setlength\itemsep{0em}
			\setcounter{enumi}{3}
			\item {\bfseries (BM)} the batch mean estimator $\hat{\sigma}^2_{BM}$ with batch size $b^{\rm (BM)}_{M}$, 
			\begin{align*}
				\hat{\sigma}^2_{BM}= \frac{b_M}{a_M-1} \sum_{k=0}^{a_M-1} \{Y_{b_M^{(BM)}}(kb_M^{(BM)}) - Y_M\}^2  ,
			\end{align*}
			where $a_M = \lfloor M/b^{\rm (BM)}_{M} \rfloor$ is the number of batches,
			\item {\bfseries (OLBM)} the overlapping batch mean estimator $\hat{\sigma}^2_{OLBM}$ with batch size $b_M^{\rm (OLBM)}$, 
			\begin{align*}
				\hat{\sigma}^2_{OLBM}= \frac{M b_M^{\rm (OLBM)}}{(M-b_M^{\rm (OLBM)})(M-b_M^{\rm (OLBM)}+1)} \sum_{j=0}^{M-b_M^{\rm (OLBM)}+1} \{Y_{b_M^{\rm (OLBM)}}(j) - Y_M\}^2  ,
			\end{align*}
			\item {\bfseries (Init)} the initial (positive,monotone,convex) sequence estimator $\hat{\sigma}^{2}_{\rm init,{\rm type}}$ computed as
			\begin{align*}
				\hat{\sigma}^2_{\rm init,{\rm type}}= -\tilde{r}_{M}(0)+2\sum_{k=0}^{T-1} \hat{\Gamma}_{M}^{\rm (type)}(k)
			\end{align*}
			for $\rm type \in \{pos,mono,conv\}$, where $\hat{\Gamma}_M(k) = \tilde{r}_M(2k)+\tilde{r}_M(2k+1)$, $T:=\min\{k \in \N; \hat{\Gamma}_M(j) <0,\}$ is the first time point where $\hat{\Gamma}_M(k)$ becomes negative, and $\hat{\Gamma}_{M}^{\rm (pos)}(k)$, $\hat{\Gamma}_{M}^{\rm (mono)}(k)$, and $\hat{\Gamma}_{M}^{\rm (conv)}(k)$ are defined for $k<T$ as
			\begin{itemize}
				\item $\hat{\Gamma}_{M}^{\rm (pos)}(k)=\hat{\Gamma}_M(k)$,
				\item $\hat{\Gamma}_{M}^{\rm (mono)}(k) = \min_{j\le k} \hat{\Gamma}^{\rm (pos)}_M(j)$, and
				\item $\hat{\Gamma}_{M}^{\rm (conv)}(k)$ is the $k$th element of the greatest convex minorant of $\hat{\Gamma}_M(0),\dots,\hat{\Gamma}_M(T-1)$
			\end{itemize}
			for $k=0,\dots,T-1$.
			
		\end{enumerate}
		The asymptotic variance estimator from the empirical autocovariance sequence is always $0$, i.e., $\sigma^2(\tilde{r}_M) = \sum_{k \in \Z} \tilde{r}_M(k)=0$ [e.g., \citealp{brockwell2009time}], and therefore is inconsistent for $\sigma^2(\gamma)$ whenever $\sigma^2(\gamma)>0$. The asymptotic variance estimator from a windowed empirical autocovariance sequence is also sometimes called a spectral variance estimator since it corresponds to an estimated spectral density function at frequency $0$.

		\subsubsection{Choice of hyperparameters\label{subsec:hyperparameters}}
		
		Hyperparameters are required for the Bartlett windowed estimators, BM, OLBM, and Moment LSEs.  A batch size $b_M$ needs to be specified a priori for the Bartlett windowed sequence estimate, BM, and OLBM, and $\delta$ determining the set $\Ml$ onto which the initial autocovariance sequence $\rinit$ is projected must be specified for the MomentLS estimators. 
		
		For BM and OLBM, we used oracle hyperparameter settings when possible. From~\citet{flegal2010batch}, for the BM and OLBM methods, the mean-squared-error optimal batch sizes for estimating $\sigma^2(\gamma)$ are
		\begin{align}\label{eq:opt_batchsizes}
			b^{\rm (BM)}_{M} = \left(\frac{\Gamma^2 M}{\sigma^2(\gamma)}\right)^{1/3} \quad \mbox{and} \quad b_{M}^{\rm (OLBM)} = \left(\frac{8\Gamma^2 M}{3\sigma^2(\gamma)}\right)^{1/3}
		\end{align}
		respectively, where $\Gamma = -2\sum_{s=1}^\infty s\gamma(s)$. Since the spectral variance estimator based on the Bartlett window is asymptotically equivalent to the overlapping batch mean estimator \citep{Damerdji1991-oj}, we let $b^{\rm (Bart)}_M = b_{M}^{\rm (OLBM)}$. If oracle hyperparameters cannot be obtained because $\gamma$ is unknown, we used the batch size tuning method implemented in the R package {\bfseries mcmcse}~\citep{liu2021batch}.
		
		For MomentLS estimators, we consider an oracle and data-driven choice of $\delta$. 
		An oracle choice of $\delta$ for MomentLS would be $\delta_\gamma = 1-\sup\{|x|; x\in {\rm Supp}(F)\}$ for the representing measure $F$ for the autocovariance sequence $\gamma$. For the data-driven choice of $\delta$, we tune $\delta$ based on a modification of an adaptive bandwidth selection method proposed by \citet{politis2003adaptive}. 
		
		\citet{politis2003adaptive} proposed an empirical rule of picking a lag $\hat{m}$ at which to truncate the autocovariance sequence. Under the assumption of uniform convergence of the empirical autocorrelations $\hat{\rho}_M(k) = \tilde{r}_M(k)/\tilde{r}_M(0)$ such that
		\begin{align}\label{eq:rho_bound}
			\max_{k=0,\dots,M-1} |\hat{\rho}_M(k) - \rho(k)|= O_{P}(\sqrt{\log M/M})   
		\end{align}
		(ref. eq (10) in \cite{politis2003adaptive}), \citet{politis2003adaptive} proposed the use of an estimator $\hat{m}$ satisfying $\hat{m} /\log (M) \to -1/(2\log|\alpha|)$ in probability, for stationary discrete-time process $X_1,\dots,X_M$ with an exponentially-decaying autocovariance sequence satisfying $\gamma(k) = C\alpha^{|k|}, |k|>k_0$ for some $k_0<\infty$ and $|\alpha|<1$.
		
		In our setting, $\gamma(k) = \int \alpha^{|k|} F(d\alpha)$ is a mixture of $\alpha^{|k|}$. Recall that any fixed choice of $\delta>0$ such that ${\rm Supp}(F) \subseteq [-1+\delta,1-\delta]$ is valid to guarantee the a.s. sequence and asymptotic variance estimator convergences in Theorems \ref{thm:as_l2conv} and \ref{thm:as_conv_avar}. In particular, any fixed $\delta\leq \delta_\gamma$ is a valid choice for a moment LS estimator. Note that $\delta_{\gamma}$ can be larger than the spectral gap of the transition kernel $Q$.
		With a modification of the empirical rule in \cite{politis2003adaptive}, we propose to use a data-driven $\tilde{\delta}_M$ such that $\tilde{\delta}_M <\delta_\gamma $ with high probability under the condition of \eqref{eq:rho_bound}. 
		
		Compared to the empirical rule by \citet{politis2003adaptive}, our proposed rule focuses only on even lags of empirical autocorrelations. More concretely, we first choose $\hat{m}$ such that 
		\begin{align}\label{eq:mhat}
			\hat{m} = \min\{t\in 2\mathbb{N}; \hat{\rho}_M(t+2) \leq c_M\sqrt{\log M/M}\}
		\end{align}
		for some $c_M \ge 0$. 
		This change is motivated by the fact that for reversible chains, $\gamma(k)$ is always nonnegative for even $k$, and the magnitude of $\gamma(k)$ can be arbitrarily small for odd $k$ due to the potential cancellations of $\alpha^k$ terms from positive and negative $\alpha$ values. To illustrate this point, consider a simple example with $\gamma(k) = (-0.9)^{k}+0.9^{k}$ for $k=0,1,2,\dots$; it is clear that $\gamma(k)=0$ for any odd $k$. 
		
		Once we have determined $\hat{m}$, we let
		\begin{align}\label{eq:delta_hat}
			\hat{\delta}_M= \max\{1-\exp\{-\log(M)/(2\hat{m})\} ,1/M\}.
		\end{align}
		Under the condition \eqref{eq:rho_bound}, we show that $1-\exp\{-\log(M)/(2\hat{m})\}$ is not asymptotically larger than $\delta_\gamma$ for any choice of $c_M \ge 0$, and converges to $\delta_\gamma$ in probability as $M\to \infty$ if we choose $c_M$ so that $c_M \to \infty$ such that $c_M= O(\log(M))$ (see Supplementary Material S6 of~\citealp{berg2023efficientsuppl}). Therefore for any positive $c_\delta<1$, $\tilde{\delta}_M = c_{\delta}\hat{\delta}_M$ should serve as an asymptotically conservative choice for $\delta_\gamma$. We choose $c_\delta<1$ in $\tilde{\delta}_M = c_{\delta}\hat{\delta}_M$, since for $c_{\delta}=1$, the probability of $\tilde{\delta}_M > \delta_\gamma$ may not go to $0$, even in the case that $\hat{\delta}_M$ converges to $\delta_\gamma$ in probability. We note that whereas the~\citet{politis2003adaptive} procedure allows for nonincreasing $c_M=c$, we were unable to verify $\hat{m}/\log(M)\overset{p}{\to}-1/\{2\log(1-\delta_{\gamma})\}$ without the condition $c_M\to\infty$. 
		
		Additionally, since $\hat{\delta}_M$ is random, the finite sample performance of momentLS estimators is influenced by the variability of $\hat{\delta}_{M}$.  We use an averaging procedure in order to reduce the variability of $\hat{\delta}_M$, in which $\hat{\delta}$ estimates from separate segments of the observed chain $\{g(X_t)\}_{t=0}^{M-1}$ are averaged. Specifically, we partition the observed series $\{g(X_t)\}_{t=0}^{M-1}$ into $L$ equal length splits, and compute the empirical autocovariances for each split in the following way. Let $B = \lfloor M/L \rfloor$. The $k$th autocovariance from the $l$th split, for $l=1,\dots,L$, is computed as
		\begin{align*}
			\tilde{r}_{M/L}^{(l)} (k) = 
			\begin{cases}
				\frac{1}{B} \sum_{t=0}^{B-1-k} \tilde{g}(X_t)\tilde{g}(X_{t+k}) & l=1\\
				\frac{1}{B} \sum_{t=(l-1)B-k}^{lB-1-k} \tilde{g}(X_t)\tilde{g}(X_{t+k}) & l>1\\
			\end{cases}
		\end{align*}
		where we recall $\tilde{g}(X_t) = g(X_t) - M^{-1}\sum_{t=0}^{M-1}g(X_t)$.  Then we computed $\hat{\delta}_{M/L}^{(l)}$ using $\{\tilde{r}_{M/L}^{(l)} (k)\}_{k=0}^{B-1}$ for $l=1,\dots,L$. Finally, we used 
		\begin{align*}
			\tilde{\delta}_M = 0.8 \frac{1}{L}\sum_{l=1}^L \hat{\delta}_{M/L}^{(l)}
		\end{align*}
		with the choice $L=5$ as the input for the Moment LS estimators in the experiments.

		It is worth mentioning that in Theorems \ref{thm:as_l2conv} and \ref{thm:as_conv_avar}, the provided almost sure convergence guarantees are applicable to a Moment LS estimator $\Pi_\delta(\rinit)$ with a valid, non-random $\delta$. Also, while uniform convergence of empirical autocovariance sequences has been studied and the uniform bound \eqref{eq:rho_bound} has been established for certain stationary time series whose examples include IID chains and the AR(1) chain of Example~\ref{exmp:ar1exmp} with $g(x)=x$, see e.g., \citet{hong1982autocorrelation, kavalieris2008uniform}, it is still an open question to establish similar results for a general geometrically ergodic Markov chain with arbitrary initial condition. We leave it as a future work to provide a full justification for moment LS estimators with this tuned choice of $\delta$.

		\subsection{Empirical illustration of the convergence properties of Moment LSEs}
		We recall that the convergence guarantees in Theorems~\ref{thm:as_l2conv} and \ref{thm:as_conv_avar} apply for Moment LS estimates with $\delta$ chosen such that $\delta>0$ and ${\rm Supp}(F) \subseteq [-1+\delta,1-\delta]$, where $F$ is the representing measure for the autocovariance sequence. Here, we empirically explore convergence of both the autocovariance sequence and the asymptotic variance estimators at varying $\delta$ levels, including cases in which the support of $F$ is \textit{not} contained in $[-1+\delta,1-\delta]$. This latter setting is not covered by our Theorems~\ref{thm:as_l2conv} and~\ref{thm:as_conv_avar}, and in this case we expect the projection to $\Ml$ to lead to bias in the corresponding Moment LSE.
		
		For $\delta$ chosen such that ${\rm Supp}(F)\subseteq [-1+\delta,1-\delta]$, Figures \ref{fig:discreteDeltaComp} and \ref{fig:ar(1)deltaComp} show that Moment LSEs lead to consistent estimates for both the autocovariance sequence (with respect to the $\ell_2$ distance) and the asymptotic variance $\sigma^2(\gamma)$. Larger values of $\delta$ (subject to ${\rm Supp}(F)\subseteq[-1+\delta,1-\delta]$) lead to relatively better performance in the estimation of both the autocovariance sequence and the asymptotic variance, although the rates of convergence at different values of $\delta$ appear to be similar. 
		
		When $\delta=0$, the moment LS estimator appears to be consistent for the true autocovariance sequence with respect to the $\ell_2$ norm distance, but inconsistent with respect to the asymptotic variance (Figure~\ref{fig:discreteDeltaComp}). On $(-1,1)$, the function $\alpha \to \frac{1+\alpha}{1-\alpha}$ is unbounded and can no longer be uniformly approximated by polynomials of finite degree. Thus the $\ell_2$ sequence convergence property at $\delta=0$ does not transfer, as in Theorem~\ref{thm:as_conv_avar} with $\delta>0$, to convergence of the estimated asymptotic variance.
		
		In the setting where $\delta>0$ is chosen so large that ${\rm Supp}(F)$ is not contained in $[-1+\delta,1-\delta]$, we observe an apparent bias variance trade-off. Our results in this setting suggest that an optimal choice of $\delta$ will strike a balance between the increase in variability expected in projecting to larger sets $\Ml$ for small $\delta$, and the increase in approximation error expected when $\gamma\notin \Ml$ for large $\delta$. In the discrete state space Metropolis-Hastings example, $\delta\leq 0.355$ is required for ${\rm Supp}(F)\subseteq[-1+\delta,1-\delta]$, yet for the smaller sample sizes in our study $\delta=0.5$ leads to the best performance out of all values of $\delta$ considered for estimating both the autocovariance sequence  and the asymptotic variance (Figure~\ref{fig:discreteDeltaComp}). We suspect that the improved performance for $\delta=0.5$ results from decreased variance, and that the bias introduced by restricting the support of $\hat{\mu}_{\delta,M}$ to $[-0.5,0.5]$ is not too large since the representing measure $F$ in this example has a substantial amount of mass between $[-0.5,0.5]$. On the other hand, in the AR(1) example with $\rho=0.9$, the representing measure $F$ has no support within $[-0.8,0.8]$, and the setting $\delta=0.2$ leads to poor performance, suggesting that the bias introduced at this value of $\delta$ overcomes any gains in performance due to variance reduction.

		\begin{figure}[ht]
			\centering
			\includegraphics[width=.8\linewidth]{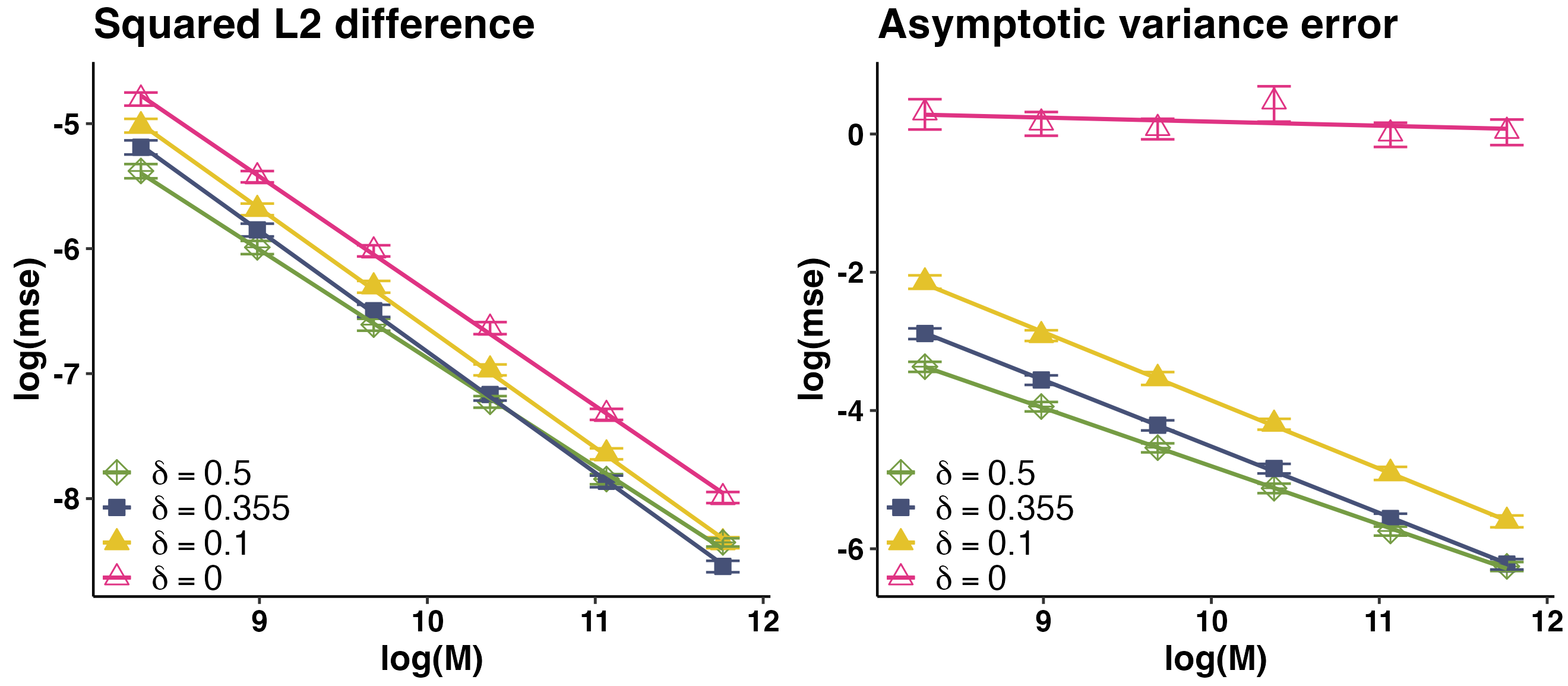}
			\caption{\textit{Metropolis-Hastings example. The support of the representing measure for $\gamma$ is contained in $[-.645,.645]$, i.e., the valid $\delta$ range is $0<\delta\le.355$.}\label{fig:discreteDeltaComp}}
		\end{figure}
		\begin{figure}[ht]
				\centering
				\includegraphics[width=.8\linewidth]{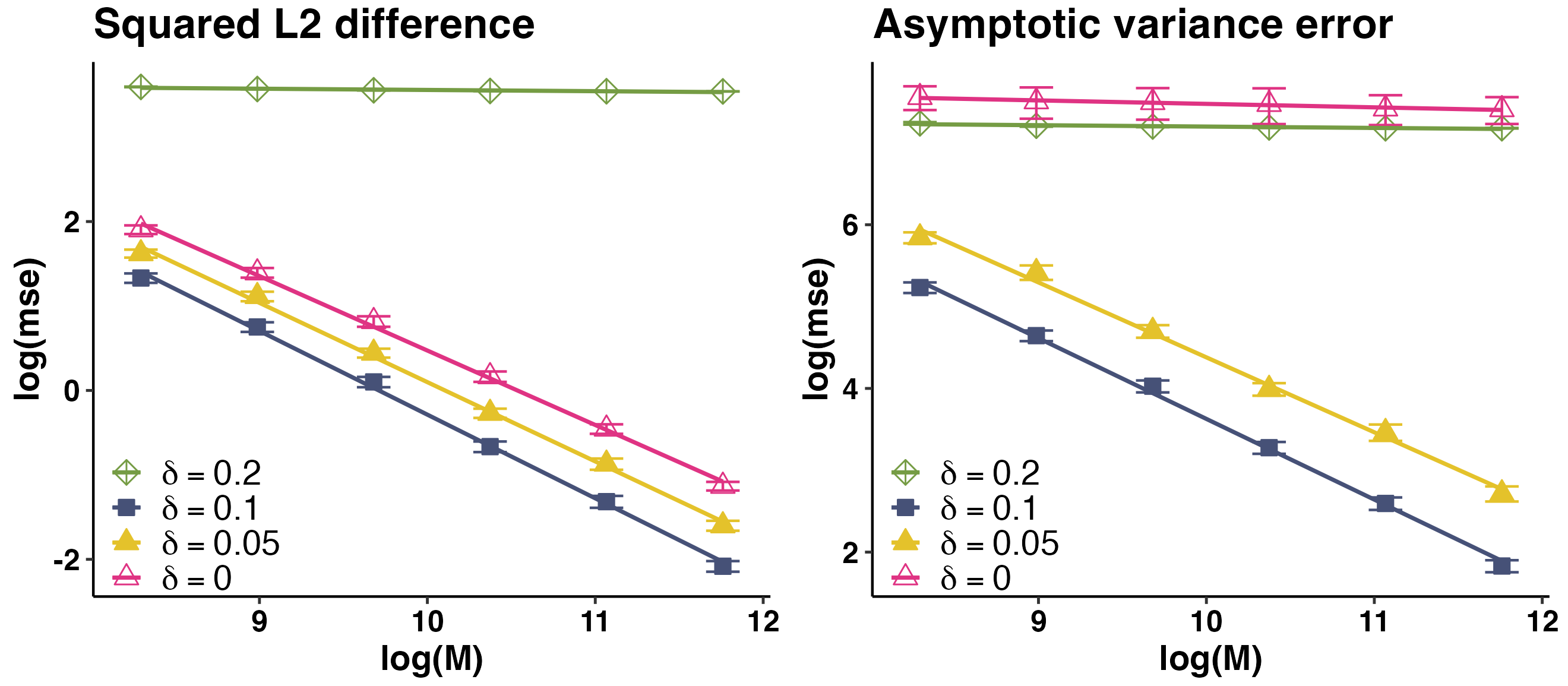}
				\caption{\textit{AR(1) example with a positive correlation ($\rho=0.9$). The representing measure has a single support point at $.9$. The valid $\delta$ range is $0<\delta\le.1$.}\label{fig:ar(1)deltaComp}}
			
		\end{figure} 
		
		\subsection{Comparison with other state-of-the-art estimators for simulated chains}\label{sec:emp_sim}
		This subsection compares the performance of our method to the performance of other current state-of-the-art methods for autocovariance sequence estimation and asymptotic variance estimation using two simulated chains. 
		
		We computed the squared $\ell_2$ autocovariance sequence error $\|\hat{r} - \gamma\|_2^2$ (when eligible) and the squared asymptotic variance error $(\hat{\sigma}^2 - \sigma^2(\gamma))^2$ for $B=400$ simulations from each method with varying chain lengths $M\in \{4000,8000,16000,32000,64000,128000\}$. All simulations were performed using R software \citep{rlang}. We used the \textbf{mcmcse} package \citep{mcmcse_R} for computing BM and OLBM estimators and the \textbf{mcmc} package \citep{mcmc_R} for computing initial positive, monotone, and convex sequence estimators. 
		
		The average squared $\ell_2$ autocovariance sequence error and average squared asymptotic variance estimation error are reported in Figures \ref{fig:discretel2avarComp}-\ref{fig:ar(1)l2AVarComp} and in tables in Supplementary Material S7~\citep{berg2023efficientsuppl}. In these results,
		\begin{itemize}
			\setlength\itemsep{0em}
			\item MomentLS(Tune,Emp) and MomentLS(Tune-Incr,Emp) refer to the moment LS estimators with the empirical autocovariance used for $\rinit$ and with $\delta$ chosen using the tuning procedure in Section~\ref{subsec:hyperparameters}, with the choices $c_M=0$ and $c_M=0.01\sqrt{\log M}$ in \eqref{eq:mhat}, and
			\item MomentLS(Orcl,Emp), MomentLS(Orcl,Brtl) refer to the moment LS estimates with oracle hyperparameter $\delta=\delta_{\gamma}$ and the empirical and Bartlett windowed autocovariances as inputs respectively.
		\end{itemize}
		
		We excluded the initial positive and monotone sequence estimators from the plots, since these generally performed similarly to or worse than the initial convex sequence estimator. To avoid overcrowding the plots, we also excluded the empirical estimator for the squared $\ell_2$ error and the empirical, Bartlett, and MomentLS(Orcl,Brtl) estimators for the asymptotic variance error from Figures \ref{fig:discretel2avarComp} - \ref{fig:ar(1)l2AVarComp}. 
		We also reported only the MomentLS(Tune,Emp) results and excluded the MomentLS(Tune-Incr,Emp) results in Figures \ref{fig:discretel2avarComp} - \ref{fig:ar(1)l2AVarComp} because both sets of results were very similar. 
		Tables that include these results can be found in Supplementary Material Section S7~\citep{berg2023efficientsuppl}.
		
		\paragraph*{Metropolis-Hastings chain}
		
		The first plot in Figure~\ref{fig:discretel2avarComp} displays the squared $\ell_2$ error $\|\hat{r}-\gamma\|^2$ for several estimators $\hat{r}$. Notably, the moment LSEs using the empirical autocovariance sequence as $\rinit$ perform best out of the estimators considered for all sample sizes, with both the data driven and oracle tuning of $\delta$.
		The moment LSE with the Bartlett windowed sequence as the input sequence (MomentLS(Orcl,Brtl)) has reduced $\ell_2$ sequence error relative to the original Bartlett windowed autocovariance sequence (Bartlett). MomentLS(Orcl,Brtl) appear to converge slower than for the Moment LSEs with the empirical autocovariance sequence as input. This decrease in convergence rate may be due to information loss from the thresholding of higher lag autocovariances in the Bartlett window sequences, which prevents information at higher lags from being used at all, in contrast to the empirical autocovariance sequence, where information from all lags can be used.
		
		The second plot in Figure \ref{fig:discretel2avarComp} compares mean squared errors for the asymptotic variance estimation. The MomentLSEs using the empirical autocovariance sequence as input again perform best out of the considered estimators. While the performance of the moment LSE with the data-driven selection of $\delta$ and that of the initial convex sequence estimator appear to be quite similar, the former shows a slightly superior performance, especially for larger values of $M$.
		The batch means estimator (BM) appears to perform slightly worse than the overlapping batch means estimator (OLBM).
		
		Figure~\ref{fig:discreteCovarianceComp} shows a plot of the true, empirical, and moment LS estimated covariances for lags $k=0,...,100$ based on a single simulation with sample size $M=8000$. The empirical and moment LS estimated covariances are similar for very small $k$, but for larger $k$ the empirical autocovariances clearly have large fluctuations about the true covariances relative to the moment LS covariances. These fluctuations apparently account for the large squared $\ell_2$ error $\|\gamma-\tilde{r}_{M}\|^2$ of the empirical estimator.

		\begin{figure}[ht]
			\centering
			\includegraphics[width=.8\linewidth]{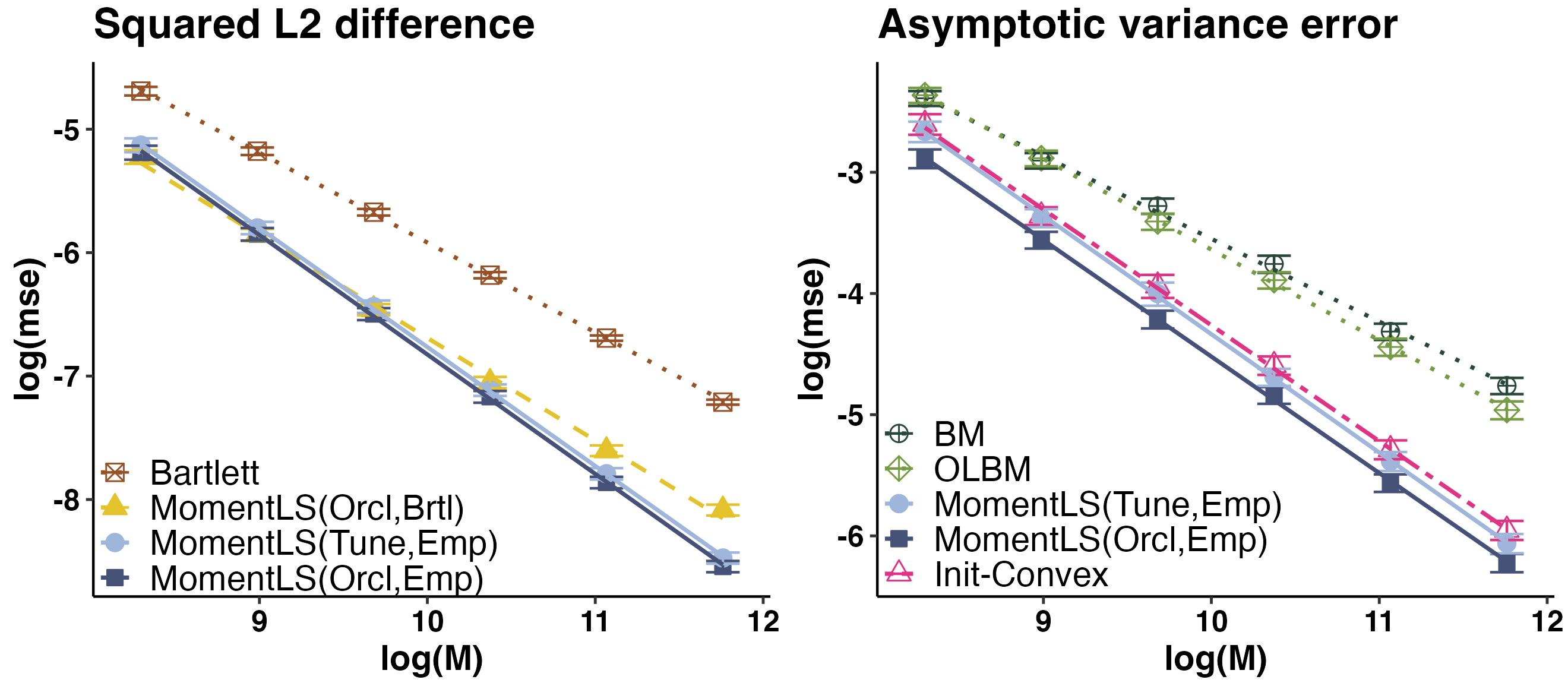}
			\caption{\textit{Plots for the discrete state space Metropolis-Hastings example. The first plot shows squared $\ell_2$ error $\|\gamma-\hat{r}\|^2$ and the second plot shows mean squared error for the asymptotic variance estimation. The error bars represent $1$ standard error from $B=400$ simulations.}} \label{fig:discretel2avarComp}
		\end{figure}

		\begin{figure}[ht]
			\centering
			\includegraphics[width=.45\linewidth]{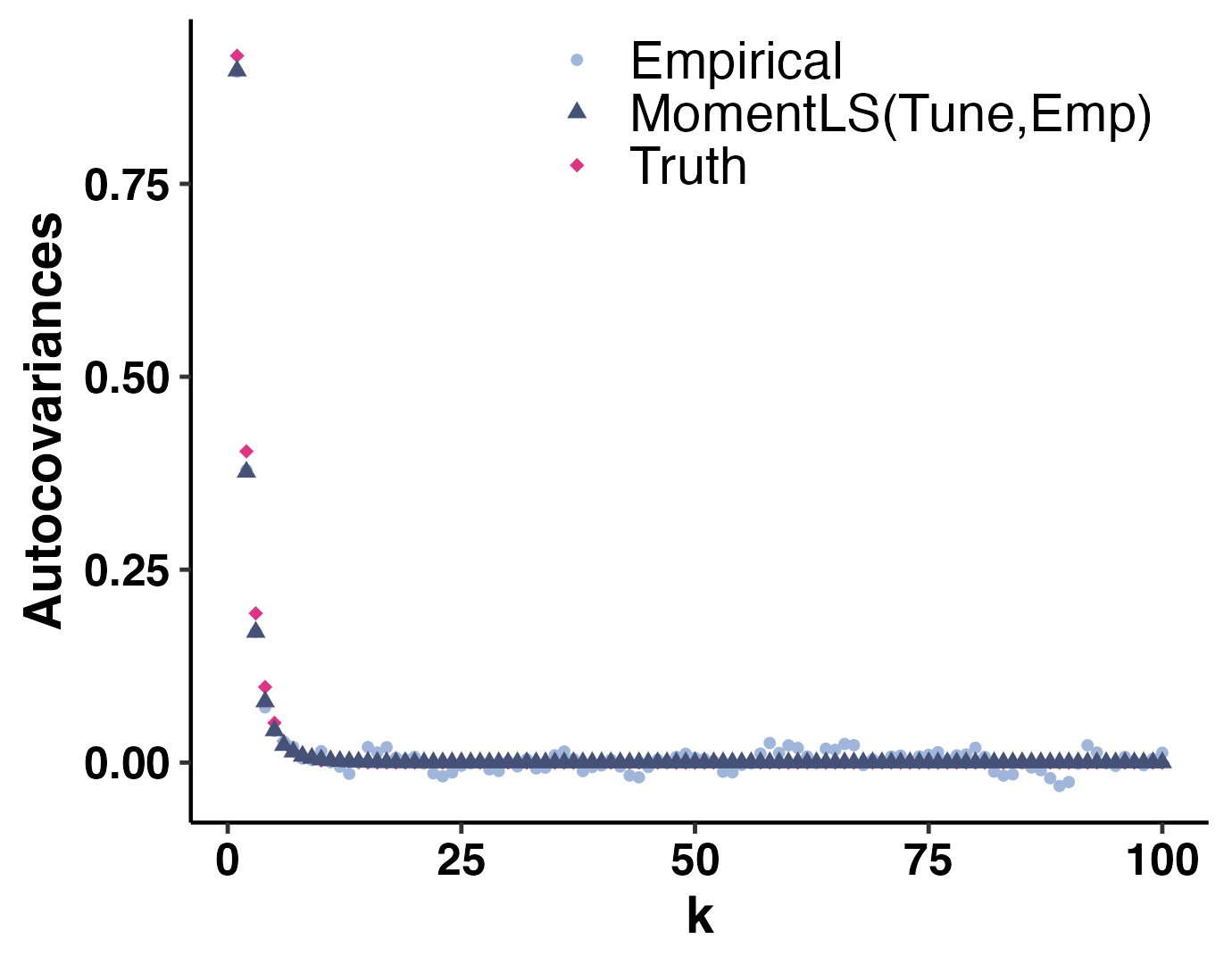}
			\caption{\textit{For the discrete state space Metropolis-Hastings example, a comparison of true, empirical, and moment LS estimated autocovariances from a single simulation with $M=8000$.}\label{fig:discreteCovarianceComp}}
		\end{figure}
		
		\paragraph*{Autoregressive chain}
		In Figure \ref{fig:ar(1)l2AVarComp}, we see generally comparable patterns in both of the AR(1) chain settings as in the discrete Metropolis-Hastings scenario. The estimated autocovariance sequences from the MomentLSEs with empirical autocovariances as the input sequences generally perform the best of the considered estimators in terms of squared $\ell_2$ error and mean squared error for estimation of the asymptotic variance. In the $\rho=-0.9$ setting, the performance of the initial convex sequence estimator appears to be quite poor relative to the other estimators. Similarly to Figure~\ref{fig:discreteCovarianceComp} for the Metropolis-Hastings example, Figure~\ref{fig:ar(1)covarianceComp} clearly shows the benefit of imposing shape constraints on the autocovariance sequence estimation, as the moment LS estimates $\Pi_{\delta}(\tilde{r}_{M})(k)$ are much closer to the true autocovariance sequence than the empirical autocovariances $\tilde{r}_M(k)$, especially for large lags $k$.
		
		\begin{figure}[ht]
			\centering
			\begin{subfigure}[t]{.8\textwidth}
				\centering
				\includegraphics[width=.96\linewidth]{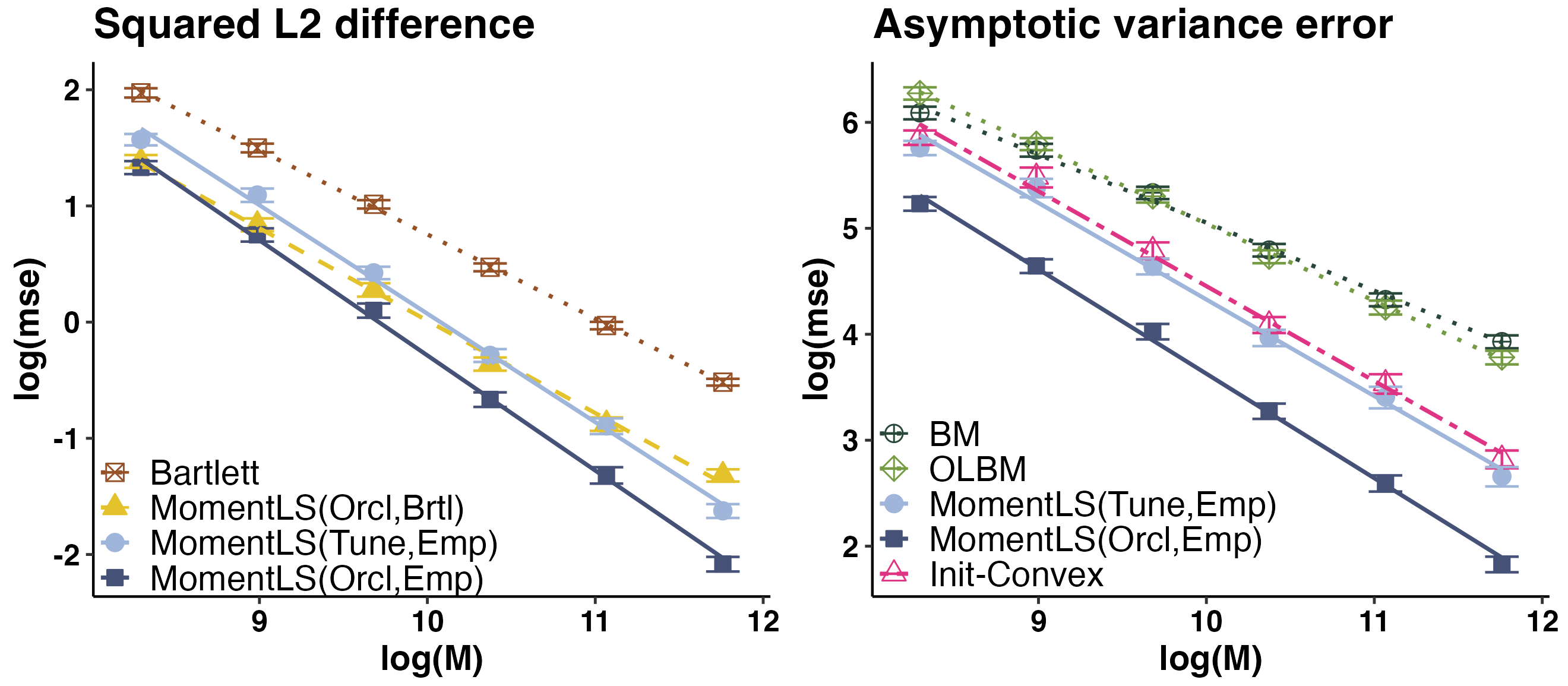}
				\caption{$\rho=0.9$\label{fig:l20.9}}
			\end{subfigure}%
			\\
			\begin{subfigure}[t]{0.8\textwidth}
				\centering
				\includegraphics[width=.96\linewidth]{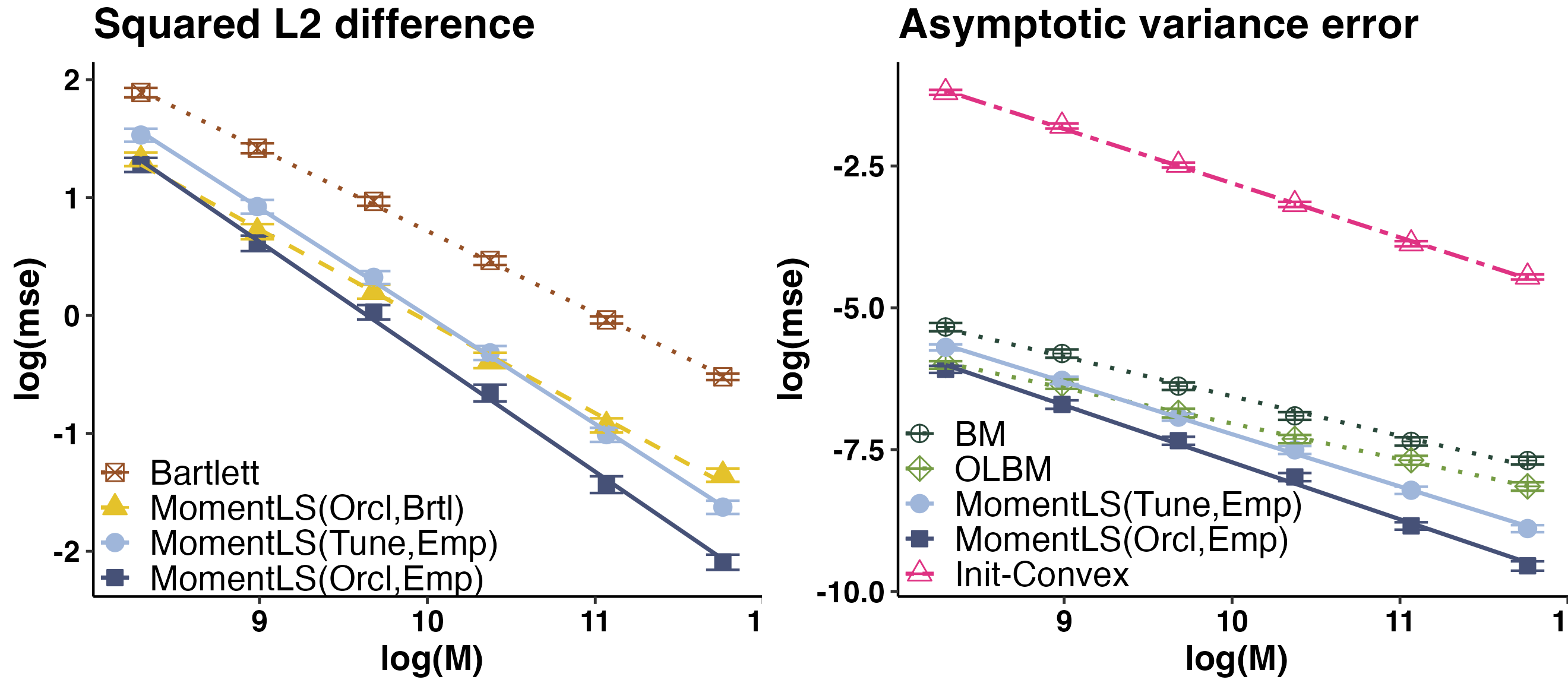}
				\caption{$\rho=-0.9$\label{fig:l2-0.9}}
			\end{subfigure}%
			\caption{\textit{Plots for the autoregressive example. Plots in the left column show squared $\ell_2$ error $\|\gamma-\hat{r}\|^2$ at $\rho=0.9$ and $\rho=-0.9$. Plots in the right column show mean squared error for the asymptotic variance estimation at $\rho=0.9$ and $\rho=-0.9$. The error bars represent $1$ standard error from $B=400$ simulations.}\label{fig:ar(1)l2AVarComp}}
		\end{figure}

		\begin{figure}[ht]
			\centering
			\begin{subfigure}[t]{0.45\textwidth}
				\centering
				\includegraphics[width=.96\linewidth]{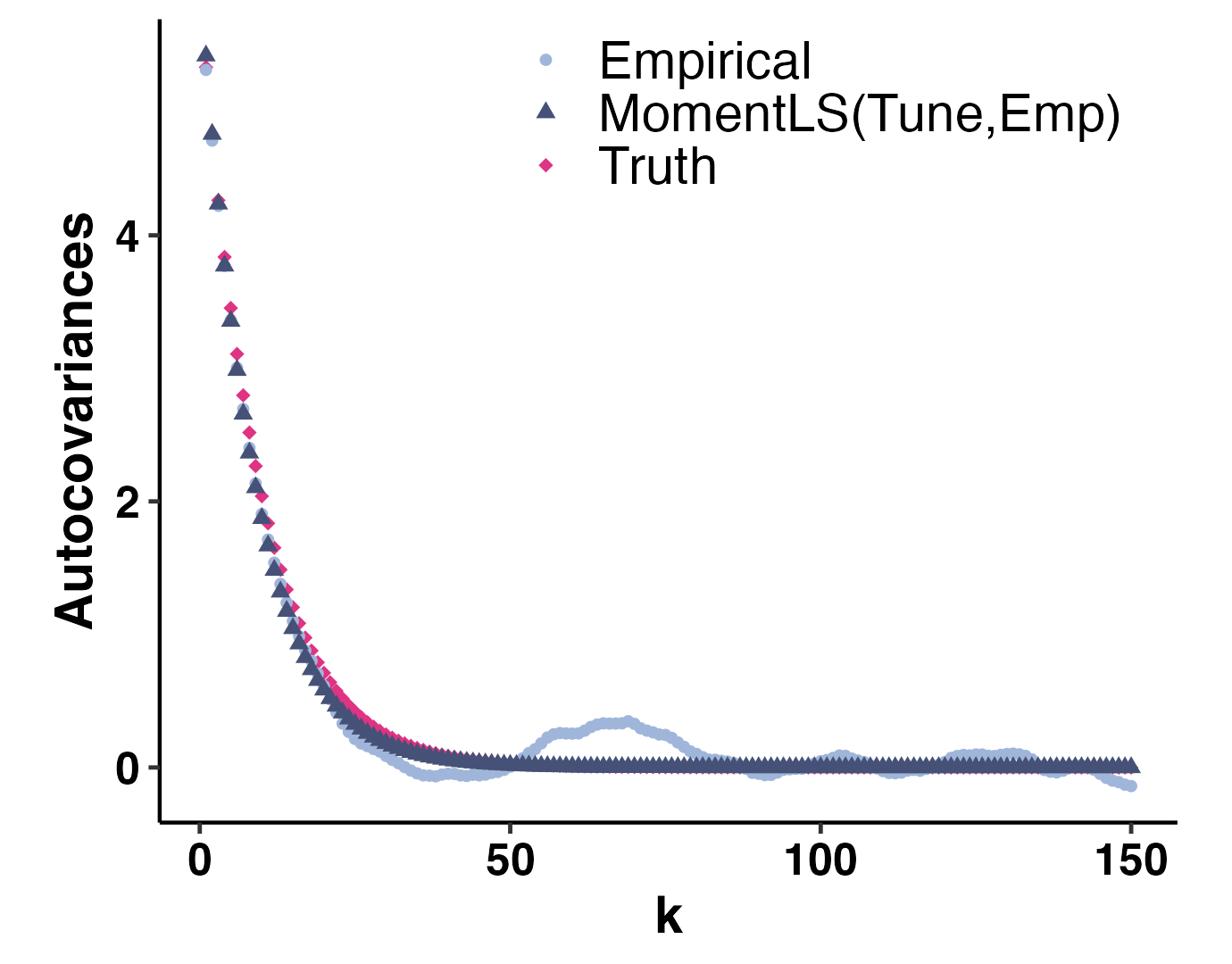}
				\caption{$\rho=0.9$ \label{fig:covariances0.9}}
			\end{subfigure}%
			~ 
			\begin{subfigure}[t]{0.45\textwidth}
				\centering
				\includegraphics[width=.95\linewidth]{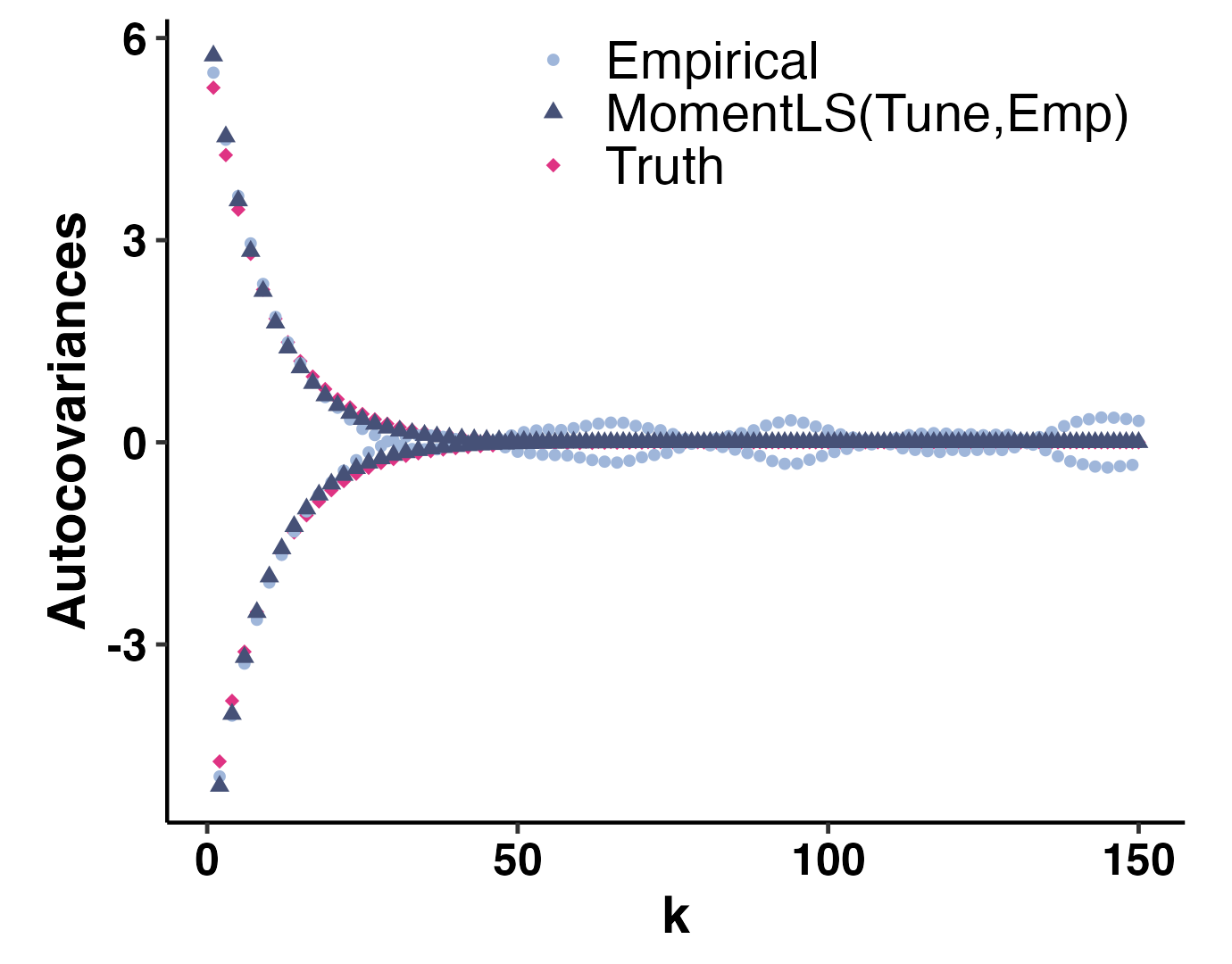}
				\caption{$\rho=-0.9$ \label{fig:covariances-0.9}}
			\end{subfigure}
			\caption{\textit{For the autoregressive example with (\subref{fig:covariances0.9}) $\rho=0.9$ and (\subref{fig:covariances-0.9}) $\rho=-0.9$, a comparison of true, empirical, and moment LS estimated autocovariances from a single simulation with $M=8000$.}\label{fig:ar(1)covarianceComp}}
		\end{figure}
		
		\subsection{Bayesian probit regression}\label{sec:emp_bayes}
		
		In this section, we illustrate the effectiveness of our method in a more realistic Bayesian probit regression model. 
		We first compare the estimated asymptotic variances from the competing methods. In addition to this, as we mentioned in the Introduction, an asymptotic variance estimator is needed to quantify uncertainty in the MCMC estimates and to effectively terminate the chain based on the perceived precision of the MCMC estimates. We conduct two experiments in this regard: first, we construct confidence intervals based on the estimated asymptotic variances of competing methods for a fixed length chain and compare their coverage probabilities; and second, we compare the coverage probabilities of competing methods for a variable length chain, where for each method the chain length is determined by a fixed-width rule.
		
		We consider the Glass identification data from the UCI machine learning repository. The dataset contains $214$ examples of the chemical analysis of 7 different types of glass. We aim to predict the first glass type based on its $9$ chemical properties $\mathbf{x}=(x_1,\dots,x_9) \in \R^{9}$. For the $i$th observation, we let $Y_i=1$ if it is of the first glass type. We suppose
		\begin{align*}
			Pr(Y_i=1) = \Phi ( \beta_0 + \sum_{j=1}^9 \beta_j x_{ij})
		\end{align*}
		and assign independent $N(0,1)$ priors on $\beta = (\beta_0,\dots,\beta_9)$. 
		
		We sample $\{\beta(t)\}_{t=0}^{M-1}$ from the posterior distribution $\beta| \{Y_i\}_{i=1}^{214} \sim \pi(\cdot)$ using the data augmentation Gibbs sampler of \cite{chib1993}. This sampler is displayed in Algorithm \ref{alg:ACsampler}. We let $\mathbf{X} \in \R^{n\times 10}$ be the design matrix where each row of $\mathbf{X}$ is $[1,\mathbf{x}_i]$. The marginal chain $\{\beta(t)\}_{t \ge 0}$, which we consider here, is reversible with respect to the posterior $\pi$~\citep[see, e.g.,][]{wongKongLiu,robert2004monte}. Additionally, the $\{\beta(t)\}_{t \ge 0}$ chain has been shown to be geometrically ergodic \citep{chakraborty2017convergence}. 
		
		\begin{algorithm}
			\caption{\cite{chib1993} sampler}\label{alg:ACsampler}
			\begin{algorithmic}
				\State 1. Draw independent $z_1,\dots,z_n$ with $z_i \sim TN(\mathbf{x}_i^\top \beta, 1, y_i)$, $i=1,\dots,n$.
				\State \quad Let $\mathbf{z} = [z_1,\dots,z_n]$.
				\State 2. Draw $\beta \sim N_p((\mathbf{X}^\top \mathbf{X} + I_n)^{-1}\mathbf{X}^\top \mathbf{z},(\mathbf{X}^\top \mathbf{X} + I_n)^{-1}).$
			\end{algorithmic}
		\end{algorithm}

		To compare estimated asymptotic variances and coverage probabilities from the competing methods, we need accurate reference estimates of posterior mean and asymptotic variance for each coefficient.  
		Since both quantities are unknown, we independently generated a long chain $\{\beta_{\rm long}(t)\}_{t=0}^{M_1-1}$ with $M_1 = 5 \times 10^6$ iterations to estimate posterior mean and also $B=1000$ independent chains $\{\beta_{\rm par}^{(b)}(t)\}_{t=0}^{M_2-1}$ with $M_2=5 \times 10^4$ to estimate asymptotic variance. Specifically, we use $\beta_{\textrm{orcl},j} = M_1^{-1}\sum_{t=0}^{M_1-1}\beta_{\textrm{long},j}(t)$ to estimate the posterior mean of the $j$th coefficient, and use
		$\sigma^2_{\textrm{orcl},j}=M\sum_{b=1}^{1000}(\bar{\beta}_{\textrm{par},j}^{(b)}-\bar{\bar{\beta}}_{\textrm{par},j})^2$ to estimate the asymptotic variance for the $j$th coefficient, where $\bar{\beta}_{\textrm{par},j}^{(b)}$ refers to the sample mean value of $\beta_j$ from the $b$th chain and $\bar{\bar{\beta}}_{\textrm{par},j}=\frac{1}{1000}\sum_{b=1}^{1000}\bar{\beta}_{\textrm{par},j}^{(b)}$ refers to the sample mean of $\bar{\beta}_{\textrm{par},j}^{(b)}$.
		
		Table \ref{tbl:chain_properties} shows some estimated summary properties for the chains from \citet{chib1993} sampler, including the estimated posterior mean $\beta_{\rm orcl}$, asymptotic variance $\sigma^2_{\textrm{ orcl}}$, Monte Carlo standard error (MCSE) for $\beta_{\rm orcl}$, as well as the estimated multiplier for the effective sample size $M_{\rm eff}/M = 1/(1+2\sum_{t\in\Z} \rho(t))$, lag 1 autocorrelation $\rho(1)$, and $\delta_\gamma$, the gap between $1$ and the largest support point (in magnitude) for the representing measure of $\gamma$. Note that a smaller value of $\delta_\gamma$ implies slower mixing, as the spectral gap should be at least as small as $\delta_\gamma$. In the table, MCSE$_j = \sigma_{\textrm{orcl},j}/\sqrt{M_1}$, $M_{\rm eff}/M$ was estimated based on $\sigma^2_{\textrm{orcl}}$ and the lag 0 empirical autocovariances from the parallel chains, $\rho(1)$ was estimated based on the empirical autocovariances at lag 0 and 1 of the long chain $\{\beta_{\textrm{long}}(t)\}_{t=0}^{M_1-1}$, and $\delta_\gamma$ was estimated by $\hat{\delta}_{M_1}$ in \eqref{eq:delta_hat}, also using the long chain $\{\beta_{\textrm{long}}(t)\}_{t=0}^{M_1-1}$. For many of the coefficients, the estimated gap $\delta_{\gamma}$ is relatively small.
		
		\begin{table}[ht]
			\caption{\textit{Some estimated summary properties of the chains from the \cite{chib1993} sampler.  }}\label{tbl:chain_properties}
			\centering
			\setlength\tabcolsep{3pt}
			\begin{tabular}{ccccccc}
				\hline
				Coef & $\beta_{\rm orcl}$ & $\sigma^2_{\rm orcl}$ & MCSE & $M_{\rm eff}/M$ & $\rho(1)$ & $\delta_\gamma$ \\ 
				\hline
				$\beta_0$ & -1.262 & 3.965 & 8.91$\times 10^{-4}$ & 0.013 & 0.912 & 0.025 \\ 
				$\beta_1$ & 0.301 & 0.337 & 2.56$\times 10^{-4}$ & 0.268 & 0.553 & 0.114 \\ 
				$\beta_2$ & -0.198 & 1.187 & 4.87$\times 10^{-4}$ & 0.102 & 0.351 & 0.050 \\ 
				$\beta_3$ & 1.555 & 3.055 & 7.82$\times 10^{-4}$ & 0.111 & 0.257 & 0.039 \\ 
				$\beta_4$ & -0.768 & 1.611 & 5.68$\times 10^{-4}$ & 0.062 & 0.599 & 0.040 \\ 
				$\beta_5$ & 0.451 & 0.772 & 3.93$\times 10^{-4}$ & 0.155 & 0.339 & 0.058 \\ 
				$\beta_6$ & -0.016 & 7.863 & 1.25$\times 10^{-3}$ & 0.025 & 0.708 & 0.042 \\ 
				$\beta_7$ & 0.047 & 0.966 & 4.40$\times 10^{-4}$ & 0.347 & 0.217 & 0.114 \\ 
				$\beta_8$ & 0.080 & 9.235 & 1.36$\times 10^{-3}$ & 0.019 & 0.791 & 0.027 \\ 
				$\beta_9$ & -0.103 & 0.056 & 1.06$\times 10^{-4}$ & 0.216 & 0.567 & 0.088 \\ 
				\hline
			\end{tabular}
		\end{table}

		\paragraph*{Comparison of asymptotic variance estimates}
		We first compare the asymptotic variance estimates $\hat{\sigma}_{jM}$ obtained by BM, OLBM, Init-Convex, and MomentLS, for each coefficient $\beta_j$, $j=0,\dots,9$. 
		
		\begin{table}[ht]
			\setlength\tabcolsep{3pt} 
			\centering
			
			\caption{\textit{Estimated mean squared relative errors (s.e.) for asymptotic variance estimates for the Glass data Bayesian probit regression with $B=400$ simulations. For each simulation, we generated a length $M=16000$ chain for $\beta$. The method with the smallest estimated mean squared errors is highlighted in bold for each coefficient.}}\label{tbl:asymp_var_error}
			\begin{tabular}{ccccc}
				\hline
				Coef & BM & OLBM & MomentLS.Tune.Emp. & Init.Convex \\ 
				\hline
				$\beta_0$ & 0.102 (0.003) & 0.089 (0.003) & \textbf{0.048 (0.004)} & 0.062 (0.006) \\ 
				$\beta_1$ & 0.008 (0.000) & 0.007 (0.000) & \textbf{0.003 (0.000)} & 0.004 (0.000) \\ 
				$\beta_2$ & 0.299 (0.002) & 0.268 (0.002) & 0.036 (0.002) & \textbf{0.033 (0.002)} \\ 
				$\beta_3$ & 0.446 (0.002) & 0.415 (0.002) & 0.069 (0.003) & \textbf{0.046 (0.002)} \\ 
				$\beta_4$ & 0.195 (0.002) & 0.172 (0.002) & \textbf{0.029 (0.002)} & 0.030 (0.002) \\ 
				$\beta_5$ & 0.216 (0.002) & 0.193 (0.002) & 0.039 (0.002) & \textbf{0.031 (0.002)} \\ 
				$\beta_6$ & 0.235 (0.003) & 0.204 (0.003) & \textbf{0.024 (0.002)} & 0.026 (0.002) \\ 
				$\beta_7$ & 0.113 (0.001) & 0.101 (0.001) & 0.054 (0.001) & \textbf{0.035 (0.001)} \\ 
				$\beta_8$ & 0.229 (0.004) & 0.201 (0.004) & \textbf{0.052 (0.005)} & 0.061 (0.005) \\ 
				$\beta_9$ & 0.020 (0.001) & 0.017 (0.001) & 0.011 (0.001) & \textbf{0.011 (0.001)} \\ 
				\hline
			\end{tabular}
		\end{table}
		
		Table \ref{tbl:asymp_var_error} shows the mean squared relative errors $\{(\hat{\sigma}_j^2 - \sigma_{\textrm{ orcl},j}^2)/\sigma_{\textrm{ orcl},j}^2\}^2$ from $B=400$ simulated chains of length $M=16000$.
		Generally, both moment LS and initial convex sequence estimators perform better than the batch means and overlapping batch means estimators. The Moment LS estimator and initial convex sequence estimator perform quite similarly.

		\paragraph*{Comparison of coverage probabilities}
		
		We compare the coverage probabilities of the confidence intervals
		\begin{align}\label{eq:confInt}
			\bar{\beta}_{jM} \pm t_{\alpha/2,M-1}\frac{\hat{\sigma}_{jM}}{\sqrt{M}}
		\end{align}
		for each coefficient $\beta_j$, $j=0,\dots,9$, using $\hat{\sigma}_{jM}$ produced by BM, OLBM, Init-Convex, and MomentLS. 
		For comparison, we also consider Oracle coverage probabilities based on the estimated ``true"  asymptotic variances $\sigma^2_{\textrm{orcl},j}$ as in the previous section.

		Table \ref{tab1:coverage_p} shows the estimated coverage probabilities for 95\% confidence intervals \eqref{eq:confInt} from length $M=16000$ chains based on the asymptotic variances from the four methods (BM, OLBM, Init-Convex, and MomentLS) as well as using the Oracle asymptotic variance estimate. We used $B=1000$ independent simulations. From Table \ref{tab1:coverage_p}, we observe that the coverage percentages for the BM and OLBM methods tend to be lower than the nominal 95\% coverage probability. The moment LS and initial convex sequence estimates show more similar behavior, with the initial convex sequence estimates achieving coverage closest to the nominal $95\%$ more often.
		
		\begin{table}[ht]
			\caption{\textit{Estimated coverage probabilities for the Glass data Bayesian probit regression with $B=1000$ simulations. For each simulation, we generated a length $M=16000$ chain for $\beta$. The method whose coverage probability is closest to 95\% (excluding the Oracle) is highlighted in bold for each coefficient.}}\label{tab1:coverage_p}
			\centering
			\setlength\tabcolsep{3pt} 
			
			\begin{tabular}{lcccccccccc}
				\hline
				Estimator & $\beta_0$ & $\beta_1$ & $\beta_2$ & $\beta_3$ & $\beta_4$ & $\beta_5$ & $\beta_6$ & $\beta_7$ & $\beta_8$ & $\beta_9$ \\ 
				\hline
				BM & 0.88 & \textbf{0.93} & 0.81 & 0.73 & 0.85 & 0.81 & 0.84 & 0.89 & 0.84 & \textbf{0.92} \\ 
				OLBM & 0.89 & \textbf{0.93} & 0.83 & 0.74 & 0.86 & 0.82 & 0.85 & 0.89 & 0.85 & \textbf{0.92} \\ 
				MomentLS(Tune,Emp) & \textbf{0.94} & \textbf{0.93} & \textbf{0.93} & 0.91 & \textbf{0.93} & 0.91 & \textbf{0.94} & 0.92 & \textbf{0.93} & \textbf{0.92} \\ 
				Init-Convex & 0.93 & \textbf{0.93} & \textbf{0.93} & \textbf{0.92} & \textbf{0.93} & \textbf{0.92} & \textbf{0.94} & \textbf{0.93} & \textbf{0.93} & \textbf{0.92} \\ 
				Oracle & 0.94 & 0.93 & 0.94 & 0.94 & 0.95 & 0.95 & 0.95 & 0.95 & 0.95 & 0.93 \\ 
				\hline
			\end{tabular}
		\end{table}
		
		We also compared the coverage probabilities in the context of fixed-width methodology~\citep{jones2006fixed}. The idea of fixed-width rules is to terminate the simulation once a desirable confidence interval half-width $\epsilon$ for an MCMC estimate is achieved. For a specified accuracy $\epsilon$, we terminate the chain the first time the following inequality holds:
		\begin{align}\label{eq:fixed_width}
			\max \left\{t_{\alpha/2,M-1} \frac{\hat{\sigma}_{jM}}{\sqrt{M}}, \text { for } j=0,1,2,\dots,9\right\}+p(M) \leq \epsilon    
		\end{align}
		where $p(M)=\epsilon I\left(M \leq M^*\right)+M^{-1}$ and $M^*$ is a desirable minimum chain length. The role of $p(M)$ is to ensure that the simulation is not terminated too prematurely. 
		\citet{glynn1992asymptotic} established that if a functional central limit theorem holds and if a strongly consistent asymptotic variance estimator is used, the $1-\alpha$ confidence interval whose chain length $M$ is chosen based on the fixed-width rule \eqref{eq:fixed_width} is asymptotically valid as $\epsilon \to 0$.
		
		We simulated $B=1000$ chains using the fixed-width rules based on the BM, OLBM, Init-Convex, and Moment LS asymptotic variance estimates. As before, the Oracle row of the table refers to coverage probability and sample size selection based on the reference asymptotic variance values $\sigma^2_{\textrm{orcl},j}$ for each coefficient. We began each simulation with a minimum chain length of $M^*$, and if the criterion \eqref{eq:fixed_width} is not satisfied, an additional 10\% of the current number of iterations were performed before checking the criterion again. We computed the 95\% confidence intervals based on the simulated chains (with random lengths) and checked whether the constructed confidence intervals included the true posterior mean or not. We used $\epsilon = 0.05$ and the minimum chain length $M^* = 1000$.
		
		Table \ref{tbl2:coverage_p_fixed_width} reports the coverage probabilities. 
		We observe a similar result as in the previous comparison. BM and OLBM tend to produce too liberal intervals. Moment LS and initial sequence estimates seem to achieve coverage probability closest to the nominal level on average, with the initial sequence estimates achieving coverage closer to nominal more often.
		
		\begin{table}[ht]
			\centering
			\caption{\textit{Average chain length at termination and coverage probabilities for the Glass data Bayesian probit regression with $B=1000$ simulations using fixed-width methods. The first column displays the mean (s.e.) chain length at termination. The method whose coverage probability is closest to 95\% (excluding the Oracle) is highlighted in bold for each coefficient.}}\label{tbl2:coverage_p_fixed_width}
			\setlength\tabcolsep{3pt} 
			
			\begin{tabular}{lccccccccccc}
				\hline
				Estimator & M (s.e.) & $\beta_0$ & $\beta_1$ & $\beta_2$ & $\beta_3$ & $\beta_4$ & $\beta_5$ & $\beta_6$ & $\beta_7$ & $\beta_8$ & $\beta_9$ \\ 
				\hline
				BM & 4,227 (40) & 0.82 & 0.94 & 0.75 & 0.63 & 0.80 & 0.79 & 0.74 & 0.88 & 0.77 & 0.91 \\ 
				OLBM & 4,563 (42) & 0.83 & 0.94 & 0.76 & 0.66 & 0.80 & 0.81 & 0.75 & 0.90 & 0.80 & 0.91 \\ 
				MomentLS(Tune,Emp) & 9,850 (70) & 0.93 & \textbf{0.95} & \textbf{0.93} & 0.87 & \textbf{0.93} & 0.89 & \textbf{0.94} & \textbf{0.93} & 0.92 & \textbf{0.94} \\ 
				Init-Convex & 10,022 (76) & \textbf{0.94} & \textbf{0.95} & \textbf{0.93} & \textbf{0.90} & 0.92 & \textbf{0.91} & \textbf{0.94} & \textbf{0.93} & \textbf{0.93} & \textbf{0.94} \\ 
				Oracle & 10,832 (0) & 0.95 & 0.94 & 0.96 & 0.94 & 0.94 & 0.94 & 0.95 & 0.95 & 0.94 & 0.94 \\ 
				\hline
			\end{tabular}
		\end{table}

		We note that in this section we have treated asymptotic variance estimation for the coefficient vector $\beta$ in a component-wise fashion. It can be beneficial to also consider output analysis tools that take cross-covariance between components into consideration~\citep[e.g.,][]{vats2019multivariate}. In this regard, extending the current framework to estimate the asymptotic variance matrix for multivariate functions of the Markov chain state, as in \citet{dai2017multivariate, vats2018strong}, is of interest.

		\section{Conclusion}
		
		In this work, we proposed a novel shape-constrained estimator for the autocovariance sequence from a reversible Markov chain. To the best of our knowledge, this is the first work in which the spectral representation of the autocovariance sequence is exploited to estimate the autocovariance sequence subject to infinitely many shape constraints. We have carried out a thorough analysis of the proposed Moment LS estimator, including its characterization and theoretical guarantees. Especially, we showed the strong consistency of the autocovariance sequence estimate from the Moment LS estimator in terms of an $\ell_2$ error metric, convergence of the representing measure of the Moment LS estimator to the true representing measure, and the strong consistency of an estimate of the Markov chain CLT asymptotic variance based on our autocovariance sequence estimator. Our theoretical results hold for reversible and geometrically ergodic Markov chains. Finally, we empirically validated our theoretical findings and demonstrated the effectiveness of the proposed estimator compared to existing autocovariance estimators in both simulated and real data settings, including batch means, spectral variance estimators, and initial sequence estimators.

		\section{Acknowledgements}
		HS and SB gratefully acknowledge support from NSF DMS-2311141.

		\ifnum\pageoption=1 
		\bibliographystyle{plainnat}
		\bibliography{bib}
		\fi
		
		\ifnum\pageoption=2 
		\putbib 
	\end{bibunit}
	\fi 
	
	\ifnum\pageoption>1
	\newpage
	\ifnum\pageoption=2\begin{bibunit}\fi
		\setcounter{page}{1}
		\renewcommand{\thepage}{S\arabic{page}} 
		\renewcommand{\thesection}{S\arabic{section}}  
		\renewcommand{\thetable}{S\arabic{table}}  
		\renewcommand{\thefigure}{S\arabic{figure}}
		\renewcommand{\theequation}{S-\arabic{equation}}
		\setcounter{equation}{0}
		\setcounter{section}{0}
		\setcounter{table}{0}
		\spacingset{1.5}
		\begin{center}
			{\Large\bf Supplement to ``\Title"}\\
			\vspace{1em}
			{\large Hyebin Song and Stephen Berg}\\
			{\large Department of Statistics, Pennsylvania State University}
		\end{center}
		
		\section{Computation of Moment LS estimators}\label{sec:computation}
		In this section, we provide some details in obtaining the convex optimization problem in \eqref{eq:opt_problem_eq2}. Recall
		\begin{align}
			\sum_{k\in\Z} (\rinit(k)-m(k))^2 
			&= \sum_{k\in\Z} (\rinit(k)-\sum_{i=1}^s \alpha_i^{|k|} w_i)^2\nonumber\\
			&= \sum_{k;\in \Z} \rinit(k)^2 -2\sum_{k\in\Z} \rinit(k)\left(\sum_{i=1}^s \alpha_i^{|k|} w_i\right)  + \sum_{k\in \Z}\left(\sum_{i=1}^s \alpha_i^{|k|} w_i\right)^2. \label{eq:quadProg1} 
		\end{align} and the definitions of $\mathbf{w}$, $\mathbf{a}$, and $\mathbf{B}$.
		
		The first term in~\eqref{eq:quadProg1} is simply ${\rinit}^\top \rinit$ for an input vector $\rinit$.
		For the second term, we have
		\begin{align*}
			\sum_{k\in\Z} \rinit(k)\left(\sum_{i=1}^s \alpha_i^{|k|} w_i\right) 
			&=\sum_{|k|\le T_0 -1} \rinit(k)\left(\sum_{i=1}^s \alpha_i^{|k|} w_i\right) \\
			&=\sum_{i=1}^s\left(\sum_{|k|\le T_0 -1} \rinit(k) \alpha_i^{|k|} \right) w_i = \mathbf{a}^\top \mathbf{w}.
		\end{align*}
		
		For the third term, we have
		\begin{align*}
			\sum_{k\in \Z}\left(\sum_{i=1}^s \alpha_i^{|k|} w_i\right)^2  
			&=\sum_{k\in \Z}\sum_{i=1}^s \sum_{j=1}^s \alpha_i^{|k|} \alpha_j^{|k|} w_i w_j\\
			&=\sum_{i,j=1}^{s} \frac{1+\alpha_i \alpha_j}{1-\alpha_i\alpha_j} w_i w_j = \mathbf{w}^\top \mathbf{B} \mathbf{w}.
		\end{align*}
		Therefore we have
		\begin{align*}
			\sum_{k\in\Z} (\rinit(k)-m(k))^2  = {\rinit}^\top \rinit -2 \mathbf{a}^\top \mathbf{w} + \mathbf{w}^\top \mathbf{B} \mathbf{w}
		\end{align*}
		as desired. Since we minimize over $m\in\mathscr{M}_{\infty}(\Theta)\cap\ell_2(\mathbb{Z})$, we require $$\mathbf{w} = [\mu_m(\{\alpha_1\}),\dots,\mu_{m}(\{\alpha_s\})] \ge 0$$ elementwise. Finally, we note that $\mathbf{B}$ is a positive definite matrix because $\mathbf{w}^\top \mathbf{B} \mathbf{w} =0 $ implies that 
		$\sum_{i=1}^s \alpha_i^{|k|} w_i  = 0$ for all $k \in \Z$. By choosing at least $s$ distinct $|k|$, we obtain $\mathbf{w} = 0$.

		\section{A few technical Lemmas}\label{sec:lemmas}
		
		\begin{lem}\label{lem:meas_boundary}
			Suppose $f \in \Mld{0}$, and let $F$ be the representing measure for $f$, i.e., $f(k) = \int  x^{|k|} F(dx)$. Then $F(\{-1,1\})=0$. That is, the measure $F$ does not have any point mass on $-1$ or $1$.
			
		\end{lem}
		\begin{proof}
			For any $k \in \Z$, 
			\begin{align*}
				F(\{-1,1\}) \le \int_{[-1,1]} x^{2k} F(dx) = f(2k).
			\end{align*}
			Since $f \in \sqs$, we have $f(2k)\to 0$ as $k\to \infty$. Thus, $F(\{-1,1\})  = 0$.
			
		\end{proof}
		
		\begin{lem}\label{lem:l1moments}
			Suppose $f \in \mathscr{M}_{\infty}(\delta)$ for some $\delta>0$. Then $f \in \ell_1(\Z)$.
			
		\end{lem}
		\begin{proof}
			
			Let $F$ denote the representing measure for $f$. We have
			\begin{align*}
				\sum_{k \in \Z} |f(k)| 
				&= \sum_{k \in \Z} \lvert \int  x^{|k|} F(dx) \rvert\\
				& \le \sum_{k \in \Z}  \int  |x|^{|k|} F(dx) \\
				& =   \int  \sum_{k \in \Z}|x|^{|k|} F(dx) 
			\end{align*}
			where the last equality is due to Tonelli's theorem. Also since 
			\begin{align*}
				\sum_{k \in \Z}|x|^{|k|} = 1 + 2\sum_{k\ge 1} |x|^{k} = \frac{1+|x|}{1-|x|},
			\end{align*}
			we have,
			\begin{align*}
				\sum_{k \in \Z} |f(k)|  \le \int \frac{1+|x|}{1-|x|} F(dx) \le \sup_{x \in \ralpha} \left( \frac{1+|x|}{1-|x|}\right)\int 1 F(dx) =\frac{2-\delta}{\delta} f(0) <\infty.
			\end{align*} where we used the fact $0\leq \int 1F(dx)<\infty$ since $f\in \mathscr{M}_{\infty}(\delta)$ implies $F$ is a finite, regular measure. Thus $f\in \ell_1(\mathbb{Z})$. 
			
		\end{proof}
		
		\begin{lem}\label{lem:inner_product}
			Suppose $f\in \Mld{0}$ with $f(k)=\int x^{|k|}F(dx)$ and $g\in\ell_2(\mathbb{Z})$. Then
			$$\braket{f,g}=\int \braket{x_{\alpha},g}F(d\alpha)=\int_{[-1,1]}\braket{x_{\alpha},g}F(d\alpha).$$	
		\end{lem}
		\begin{proof} We have \begin{align*}
				\braket{f,g}&=\sum_{k\in\mathbb{Z}}f(k)g(k)\\
				&=\sum_{k\in\mathbb{Z}}g(k)\int_{[-1,1]}\alpha^{|k|}F(d\alpha)\\
				&=\sum_{k\in\mathbb{Z}}\int_{[-1,1]}g(k)\alpha^{|k|}F(d\alpha).
			\end{align*}

			We will show that $\sum_{k\in\mathbb{Z}}\int_{[-1,1]}|g(k)\alpha^{|k|}|F(d\alpha) <\infty$. Then, the desired result follows from Fubini's theorem, since
			\begin{align*}
				\sum_{k\in\mathbb{Z}}\int_{[-1,1]}g(k)\alpha^{|k|}F(d\alpha) =\int_{[-1,1]} \sum_{k\in\mathbb{Z}}g(k)\alpha^{|k|}F(d\alpha) =\int_{[-1,1]}\braket{x_{\alpha},g}F(d\alpha).
			\end{align*}

			We have
			\begin{align*}
				\sum_{k\in\mathbb{Z}}\int_{[-1,1]}|g(k)\alpha^{|k|}|F(d\alpha)=\sum_{k\in\mathbb{Z}}|g(k)|\int_{[-1,1]}|\alpha|^{|k|}F(d\alpha) = \sum_{k\in\mathbb{Z}}|g(k)|\tilde{f}(k) \le \|g\|\|\tilde{f}\|,
			\end{align*} 
			where we define $\tilde{f}(k)=\int_{[-1,1]}|\alpha|^{|k|}F(d\alpha)$ and we use Cauchy-Schwarz for the last inequality. First, since $g\in\sqs$, $\|g\| <\infty.$ For $\|\tilde{f}\|$, we have $\tilde{f}(k)=\tilde{f}(-k)$, and for $0\leq k_1<k_2$, we have $\tilde{f}(k_1)\geq \tilde{f}(k_2)$. Additionally, for $n=2k$ we have $\tilde{f}(n)=f(n)$. Thus \begin{align*}
				\|\tilde{f}\|_2^2 = \sum_{k\in\mathbb{Z}}\tilde{f}(k)^2&=\tilde{f}(0)^2+2\tilde{f}(1)^2+2\sum_{k=2}^{\infty}\tilde{f}(k)^2\leq 3f(0)^2+4\sum_{k=1}^{\infty}f(2k)^2<\infty
			\end{align*} since $f\in \ell_2(\mathbb{Z})$. 
		\end{proof}
		
		\begin{cor}\label{cor:double_inner_product}
			For $f,g\in\Mld{0}$ with $f(k)=\int_{[-1,1]}x^{|k|}F(dx)$ and $g(k)=\int_{[-1,1]}x^{|k|}dG(x)$, \begin{align*}
				\braket{f,g}&=\int_{[-1,1]}\int_{ [-1,1]}\braket{x_{\alpha_1},x_{\alpha_2}}F(d\alpha_1)G(d\alpha_2)\\
				&=\int_{[-1,1]}\int_{[-1,1]}\frac{1+\alpha_1\alpha_2}{1-\alpha_1\alpha_2}F(d\alpha_1)G(d\alpha_2).
			\end{align*} Additionally, the order of integration in both expressions can be interchanged.
		\end{cor}
		
		\begin{proof} 
			Note than since both $f,g \in \Mld{0}$, by Lemma \ref{lem:meas_boundary}, $F(\{-1,1\}), G(\{-1,1\})=0$. We have \begin{align*}
				\braket{f,g}&=\int_{[-1,1]}\braket{x_{\alpha_1},g}F(d\alpha_1)\\
				&=\int_{(-1,1)}\braket{x_{\alpha_1},g}F(d\alpha_1)\\
				&=\int_{(-1,1)}\int_{[-1,1]}\braket{x_{\alpha_1},x_{\alpha_2}}G(d\alpha_2)F(d\alpha_1)\\
				&=\int_{[-1,1]}\int_{[-1,1]}\braket{x_{\alpha_1},x_{\alpha_2}}G(d\alpha_2)F(d\alpha_1)\\
				&=\int_{(-1,1)}\int_{(-1,1)}\braket{x_{\alpha_1},x_{\alpha_2}}G(d\alpha_2)F(d\alpha_1)\\
				&=\int_{(-1,1)}\int_{(-1,1)}\frac{1+\alpha_1\alpha_2}{1-\alpha_1\alpha_2}G(d\alpha_2)F(d\alpha_1)\\
				&=\int_{[-1,1]}\int_{[-1,1]}\frac{1+\alpha_1\alpha_2}{1-\alpha_1\alpha_2}G(d\alpha_2)F(d\alpha_1)<\infty.
			\end{align*} Since $\braket{x_{\alpha_1},x_{\alpha_2}}=\frac{1+\alpha_1\alpha_2}{1-\alpha_1\alpha_2}\geq 0$ for all $\alpha_1,\alpha_2\in [-1,1]$, we can interchange the order of integration.
		\end{proof}
		A by-product of the Corollary is $\braket{f,g}\geq 0$ for all $f,g\in\Mld{0}$.
		
		\begin{lem} \label{lem:spectralRepAvar}
			Suppose $f \in \mathscr{M}_{\infty}(0)\cap \ell_1(\mathbb{Z})$ with $f(k) = \int  \alpha^{|k|} F(d\alpha)$. Define $$\sigma^2(f)=\sum_{k=-\infty}^{\infty}f(k).$$ Then
			\begin{align*}
				\sigma^2(f) = \int \frac{1+\alpha}{1-\alpha} F(d\alpha).
			\end{align*}
		\end{lem}
		\begin{proof}
			
			We have \begin{align*}
				&\sum_{k=-\infty}^{\infty}\int |\alpha|^{|k|}F(d\alpha)
				=f(0)+2\sum_{k=1}^{\infty}\int |\alpha|^{k}F(d\alpha)
				\leq 3f(0)+4\sum_{j=1}^{\infty}f(2j)<\infty
			\end{align*} where the first inequality follows from \begin{align*}
				\int |\alpha|F(d\alpha)\leq \int 1F(d\alpha)
			\end{align*} and \begin{align*}
				\int |\alpha|^{2k+1}F(d\alpha)\leq \int \alpha^{2k}F(d\alpha), \quad k\geq 1,
			\end{align*} and the second inequality follows from $f\in \ell_1(\mathbb{Z})$. Thus, from Fubini's theorem, we have \begin{align*}
				\sigma^2(f)&=\sum_{k=-\infty}^{\infty}\int \alpha^{|k|}F(d\alpha)\\
				&=\int\sum_{k=-\infty}^{\infty} \alpha^{|k|}F(d\alpha)\\
				&=\int\frac{1+\alpha}{1-\alpha}F(d\alpha).
			\end{align*}
			
		\end{proof}
		\begin{lem}\label{lem:closed}
			Let $I$ be a closed interval in $[-1,1]$. The space $\Mld{I}$ is closed.
		\end{lem}
		\begin{proof}
			Let $a = \inf I$ and $b = \sup I$.  Consider a sequence of vectors $\{m_n\} \subseteq \Mld{[a,b]}$ where $\|m_n-f\|\to0$ for some $f=\{f(k)\}_{k=-\infty}^{\infty}$. We show that $f \in \Mld{[a,b]}$.

			First, $f\in \sqs$ since 
			\[\|f\|\leq \|m_n\|+\|f-m_n\|<\infty\] for large enough $n$. 
			
			Next, we show that $f \in \mathscr{M}_\infty([a,b])$. Note for any $j \in \Z$, $|m_n(j)-f(j)|\leq \|m_n-f\|\to 0$. 
			We consider two cases where Case I: $b-a=0$ and Case II: $b-a>0$,
			
			{\bfseries Case I:} we have $f(0) = \lim_n m_n(0)$. Additionally, $m_n(k) = m_n(0)a^{|k|}$ for $k\ne 0$, and thus, $f(k) = \lim_{n}m_n(k)=f(0)a^{|k|}$ for $k\ne 0$.  Then for the measure $\mu$ with point mass at $a$ with mass $f(0)$, i.e., $\mu = f(0) \delta_a$, we have $f(k) = \int x^{|k|} \mu(dx)$ for $k\in \Z$, and $\mu$ is supported on $\{a\}$. Thus $f\in\mathscr{M}_{\infty}([a,b])\cap\ell_2(\mathbb{Z})$.
			
			{\bfseries Case II:} we define $\tilde{f}$ on $\N$ as $\tilde{f}(k)=T(f; a,b)(k)$. We have for any $r,k\in \N$,
			\begin{align*}
				(-1)^{r}\Delta^{r}\tilde{f}(k)
				&=(-1)^r \Delta^r \underset{n\to\infty}{\lim}T(m_n;a,b)(k)\\
				&=\underset{n\to\infty}{\lim} (-1)^r \Delta^r T(m_n;a,b)(k)\\
				&\ge 0.
			\end{align*} where the last inequality holds since $m_n$ are $[a,b]$-moment sequences. Thus $\tilde{f}$ is completely monotone, so by Corollary \ref{cor:moment_thm_Z}, $f$ is an $[a,b]$-moment sequence. Since in addition $f\in \sqs$, we have $f\in \Mld{[a,b]}$. Thus $\Mld{[a,b]}$ is closed.
		\end{proof}
		
		\begin{lem}\label{lem:vague_convergence}
			Let $I$ be a closed interval in $[-1,1]$. Consider a sequence of moment sequences $\{f_n\}_{n\in\N}$ and $f$ such that $\{f_n\} \subseteq \Mld{I}$ and $\|f_n - f\| \to 0$. Then $f\in \Mld{I}$. Let $\mu_n$ and $\mu$ be the representing measures for $f_n$ and $f$ respectively. Then, we have $\mu_n \to \mu$ vaguely. 
		\end{lem}
		\begin{proof}
			First of all, $f \in \Mld{I}$ follows from Lemma \ref{lem:closed}. Let $\epsilon>0$ and $h\in C_0(\mathbb{R})$ given. We want to show that $|\int h(\alpha)\mu_n(d\alpha) - \int h(\alpha)\mu(d\alpha)|\le \epsilon$ for a sufficiently large $n$. Also,  since $\|f_n-f\|\to 0$, we have $f_n(0) \to f(0)$. In other words, $\mu_n(I) \to \mu(I)$. 
			
			Now we approximate $h$ on $I$. Since $h$ is continuous, there exists a sequence of polynomials $p_N(\alpha)=\sum_{k=0}^{N}c_k\alpha^k$ which uniformly approximates $h(\alpha)$ on $I$, by the Weierstrass approximation theorem.  Let $B = \mu(I) + \sup_n \mu_n(I)$.  Suppose $B = 0$. That is, $f(0)=0$ and $f_n(0)=0$ for all $n$. Therefore, both $\mu_n$ and $\mu$ are null measures, and the conclusion trivially holds. Now suppose $B > 0$. We choose  $N<\infty$ so that 
			\begin{align}
				\sup_{\alpha \in I} |h(\alpha) - p_N(\alpha)| \le \frac{\epsilon}{2B}.
			\end{align}
			We have,
			\begin{align*}
				&\lvert\int h(\alpha) \mu_n(d\alpha) -\int h(\alpha) \mu(d\alpha)\rvert\\
				&=\lvert \int_{I} h(\alpha)  \mu_n(d\alpha)- \int_{I} h(\alpha)  \mu(d\alpha)\rvert\\
				&= \left\lvert \int_{I} \{h(\alpha) -p_N(\alpha)\} \mu_n(d\alpha)+  \int_{I} p_N(\alpha) \mu_n(d\alpha)  \right.\nonumber\\
				&\qquad \qquad \left. -  \int_{I} \{h(\alpha)-p_N(\alpha)\} \mu(d\alpha) -  \int_{I} p_N(\alpha)  \mu(d\alpha) \right\rvert .\nonumber\\
				&\leq\underbrace{ \int_{I} \lvert h(\alpha) -p_N(\alpha) \rvert  \mu_n(d\alpha) +\int_{I} \lvert h(\alpha) -p_N(\alpha) \rvert \mu(d\alpha)}_{Term 1} \nonumber\\
				&\qquad \qquad + \underbrace{\left\lvert \int_{I} p_N(\alpha) \mu_n(d\alpha) - \int_{I} p_N(\alpha) \mu(d\alpha) \right\rvert}_{Term 2}.\nonumber
			\end{align*}
			By the choice of $N$, 
			\begin{align}
				\textrm{Term 1} \le \epsilon \frac{\mu_n(I)}{2B} + \epsilon\frac{\mu(I)}{2B} \le \frac{\epsilon}{2},
			\end{align}
			since $\mu(I)+\mu_n(I)\le B$ for any $n$.
			
			For term II, define $v_N=\{v_N(k)\}_{k\in\mathbb{Z}}$ such that
			\begin{align*}
				v_N(k) = \begin{cases}
					c_k & 0\le k \le N, \\ 0 & otherwise.
				\end{cases}
			\end{align*}
			for $\{c_k\}_{k=0}^N$ from the coefficients of the approximating polynomial $p_N(\alpha)= \sum_{k=0}^N c_k\alpha^k$.
			Note that $\|v_N\|^2 = \sum_{k=0}^N c_k^2 <\infty$ since $N<\infty$. In particular, $v_N \in \sqs$, and thus by Lemma \ref{lem:inner_product},
			\begin{align*}
				\int_I p_N(\alpha) \mu_n(d\alpha) 
				&=\int_I \sum_{k=0}^N c_k \alpha^{k} \mu_n(d\alpha) \\
				&=\int_I \sum_{k\in\Z} v_N(k) \alpha^{|k|} \mu_n(d\alpha) \\
				&=\int_I \braket{v_N, x_\alpha} \mu_n(d\alpha)  = \braket{v_N,f_n}.
			\end{align*}
			Similarly, we can show $\int_I p_N(\alpha) \mu(d\alpha)  = \braket{v_N,f}.$
			Therefore,
			\begin{align*}
				\mbox{Term 2} = |\braket{v_N,f_n} - \braket{v_N, f}|\le \|v_N\| \|f_n-f\|.
			\end{align*}
			We can find $L<\infty$ such that for $n\ge L$, $\|v_N\|\|f_n-f\|\le \epsilon/2$.
			
			Combining these results for term I and term II, we have for $n \ge L$,
			\begin{align*}
				\textrm{Term 1+ Term 2} \le \epsilon
			\end{align*} But $\epsilon>0$ was arbitrary. This proves the result.
		\end{proof}

		\section{Proofs for results in Section \ref{sec:mcmcIntro}}
		\subsection{Proof of Proposition \ref{prop:gamma_moment_seq}}\label{pf:prop1}
		\begin{proof}
			The representation \eqref{eq:gamma_k} is a consequence of the spectral theorem [e.g., \citealp{Rudin1973-ry}] since $Q_0$ is a self-adjoint bounded linear operator on $L^2(\pi)$. The spectrum $\sigma(Q_0)$ lies on the real axis due to \ref{cond:piReversible}. Since the spectral radius $\rho(Q_0) = \sup\{|\rho|; \rho \in \sigma(Q_0)\}$ is equal to $\|Q_0\|_{L^2(\pi)}$ since $Q_0$ is self-adjoint, and $\|Q_0\|_{L^2(\pi)} \le 1$, we have $\sigma(Q_0) \subseteq [-1,1]$.
			
			For $\Gamma$, we have
			\begin{align*}
				\Gamma(k) = \braket{Q_0^{2k} g, g}_\pi +  \braket{Q_0^{2k+1} g, g}_\pi =\braket{Q_0^{2k} (I+Q_0)g, g}_\pi 
			\end{align*} for $k\in\mathbb{N}$. Let $(I+Q_0)^{1/2}$ denote the square root of $I+Q_0$, which is well defined because $I+Q_0$ is positive and self-adjoint, where the positivity of $I+Q_0$ is due to the fact that $\|Q_0\|_{L^2(\pi)} \le 1$.
			Also, we have that $(I+Q_0)^{1/2}$ is positive, self-adjoint, and commutes with $Q_0$ [e.g., Theorem in \citealp{Riesz2012-tp}, p265]. Therefore,
			\begin{align*}
				\Gamma(k) &= \braket{ (I+Q_0)^{1/2} Q_0^{2k} (I+Q_0)^{1/2}g, g}_\pi \\
				&= \braket{  Q_0^{2k} (I+Q_0)^{1/2}g, (I+Q_0)^{1/2} g}_\pi\\
				& = \braket{  (Q_0^2)^{k} h, h}_\pi
			\end{align*}
			where $h = (I+Q_0)^{1/2}g$. Then by the spectral theorem, there exists a regular measure $H$ supported on $\sigma(Q_0^2)$ such that $\Gamma(k)=\int x^{k}H(dx)$, $k\in\mathbb{N}$. Since $Q_0^2$ is positive and $\|Q_0^2\|_{L^2(\pi)} =\|Q_0\|^2_{L^2(\pi)}$, we have $\sigma(Q_0^2) \subseteq [0, 1]$.
			
			Finally, with the additional assumption of \ref{cond:harris_ergodicity} and \ref{cond:geometric_ergodicity}, the spectral gap $1-\rho(Q_0) >0$ \citep{Roberts1997-qk, kontoyiannis2012geometric}. We can find $\delta_0>0$ so that $\rho(Q_0) = \|Q_0\|_{L^2(\pi)} = 1-\delta_0$. Therefore $\sigma(Q_0) \subseteq [-1+\delta_0,1-\delta_0] $. Since $\|Q_0^2\|_{L^2(\pi)} =\|Q_0\|^2_{L^2(\pi)}=(1-\delta_0)^2$, we have $\sigma(Q_0^2) \subseteq [0, (1-\delta_0)^2]$.
		\end{proof}
		\section{Proofs for results in Section \ref{sec:momentSequences}}
		\subsection{Proof of Proposition \ref{prop:momentSeq}}\label{pf:momentSeq}
		\begin{proof}
			First, suppose there is a measure $\mu$ on $[a,b]$ such that $m(k)=\int x^k\mu(dx)$ for all $k\in \mathbb{N}$. Define $g_k(x)=x^k$ and $f(x)=(x-a)/(b-a)$. Also, we define $\tilde{\mu}(A)=\mu(h(A))$ where $h(x) = f^{-1}(x) = (b-a)x+a$ and $h(A):=\{h(x); x\in A\}$. We show that $\tilde{\mu}$ is a representing measure for $T(m;a,b)$ and $\tilde{\mu}$ is supported on $[0,1]$. 
			First of all, by the change of variable formula, for $k\in \N$,
			\begin{align*}
				T(m;a,b)(k) 
				& = (b-a)^{-k}\sum_{i=0}^{k}\binom{k}{i}m(i)(-a)^{k-i}\\
				& = (b-a)^{-k}\int_{[a,b]}\sum_{i=0}^{k}\binom{k}{i} x^k (-a)^{k-i}  \mu(dx) \\
				& = (b-a)^{-k}\int_{[a,b]}(x-a)^k  \mu(dx) \\
				& =\int I\{f(x)\in [0,1]\}  f(x)^k  \mu(dx) \\
				& = \int_{[0,1]}y^{k}\tilde{\mu}(dy).
			\end{align*}
			Also, $\tilde{\mu}$ is supported on $[0,1]$ since
			\begin{align*}
				\tilde{\mu}(\R\setminus [0,1]) = \mu(\{h(x); x<0 \mbox{ or } x>1\}) = \mu(\R\setminus [a,b])=0.
			\end{align*}
			Then by Theorem \ref{thm:hausdorff_moment}, $\tilde{\mu}$ is the unique representing measure for $T(m;a,b)$, and $T(m;a,b)$ is completely monotone.
			
			Now suppose $T(m;a,b)$ is completely monotone. Then there exists a measure $\tilde{\mu}$ supported on $[0,1]$ such that $T(m;a,b)(k)=\int y^{k}\tilde{\mu}(dy)$. Define the measure $\mu$ by $\mu(A)=\tilde{\mu}(f(A))$ where $f(A):=\{f(x); x\in A\}$. First, note that $\mu$ is supported on $[a,b]$. 
			From the definition of $T(m;a,b)$, we have $m(0)=T(m;a,b)(0)$ and, recursively,
			\begin{align}\label{eq:lem_moment_recursive}
				m(k)=(b-a)^kT(m;a,b)(k)-\sum_{i=0}^{k-1}\binom{k}{i}m(i)(-a)^{k-i}. 
			\end{align}
			We now show that $m(k) = \int x^k \mu(dx)$ for $k \in \N$.
			Recall $h(x) = f^{-1}(x) = (b-a)x+a$.
			From the definition of $T(m;a,b)$ and change of variable formula, we have for any $k\in\N$,
			\begin{align*}
				T(m;a,b)(k) &= \int_{[0,1]} y^{k}\tilde{\mu}(dy)\\
				&=\int I\{ f(h(y)) \in [0,1]\} f(h(y))^{k}\tilde{\mu}(dy)\\
				&=\int I\{f(x) \in [0,1]\} f(x)^{k}(\tilde{\mu}\circ f)(dx)\\
				&=\int_{[a,b]} \left(\frac{x-a}{b-a}\right)^{k}\mu(dx)\\
				&=\int  \left(\frac{x-a}{b-a}\right)^{k}\mu(dx).
			\end{align*}
			When $k=0$, $ m(0) = T(m;a,b)(0) = \int 1 \mu(dx).$
			Suppose $m(i)=\int x^{i}\mu(dx)$ for $i=0,...,k$. We show $m(k+1)=\int x^{k+1}\mu(dx)$. By \eqref{eq:lem_moment_recursive},
			\begin{align*}
				m(k+1)&=(b-a)^{k+1}T(m;a,b)(k+1)-\sum_{i=0}^{k}\binom{k+1}{i}m(i)(-a)^{k+1-i}\\
				&=\int (x-a)^{k+1}\mu(dx)-\sum_{i=0}^{k}\binom{k+1}{i}m(i)(-a)^{k+1-i}\\
				&=\int \sum_{i=0}^{k+1}\binom{k+1}{i}x^{k}(-a)^{k+1-i}\mu(dx)-\int \sum_{i=0}^{k}\binom{k+1}{i}x^{k}(-a)^{k+1-i} \mu(dx)\\
				&=\int x^{k+1}\mu(dx).
			\end{align*} Thus, by induction, $m(k)=\int x^{k}\mu(dx)$ for $k=0,1,...$, for $\mu$ defined by $\mu(A)=\tilde{\mu}(f(A))$.

			Finally, for uniqueness, let $\mu_1$ and $\mu_2$ be two representing measures for $m$. From the first part of the proof, we see that both $\tilde{\mu}_1(A):=\mu_1(h(A))$ and $\tilde{\mu}_2(A):=\mu_2(h(A))$ are representing measures supported on $[0,1]$ for $T(m;a,b)$. Then $\tilde{\mu}_1 = \tilde{\mu}_2 = \tilde{\mu}$ from Theorem \ref{thm:hausdorff_moment}. Then for any measurable set $E$,
			\begin{align*}
				\mu_1(E) = \tilde{\mu}(h^{-1}(E)) = \mu_2(E).
			\end{align*}
			Thus, the measure $\mu$ corresponding to $m$ is unique.
		\end{proof}

		\subsection{Proof of Proposition \ref{prop:existence_rhat}}\label{pf:existence_rhat}

		\begin{proof} 
			We show $\Mld{C}$ is a convex and closed subset of $\sqs$. Convexity holds since for $p,q\in\Mld{C}$ where $p(k)=\int_\ralpha x^{|k|}F_1(dx)$ and $q(k)=\int_\ralpha x^{|k|}F_2(dx)$, we have $u=\alpha p+(1-\alpha)q\in \sqs$ and $u(k)=\int_\ralpha x^{|k|}(\alpha F_1+(1-\alpha)F_2)(dx)$, i.e., $u\in\Mld{C}$. 
			
			Now, we show $\Mld{C}$ is closed. In the case $C$ is a closed interval, then from Lemma~\ref{lem:closed}, $\Mld{C}$ is closed. Otherwise for a general closed set $C$, consider a sequence of vectors $\{m_n\} \subseteq \Mld{C}$ where $\|m_n-f\|\to0$ for some $f=\{f(k)\}_{k=-\infty}^{\infty}$. Now, let $a=\inf\,C$ and $b=\sup\,C$. Then $m_n\in\Mld{[a,b]}$ so from Lemma~\ref{lem:closed}, $f\in\Mld{[a,b]}$. In particular, $f$ is an $[a,b]$-moment sequence. Let $\mu_n$ denote the representing measure for $\mu_n$ and let $\mu$ denote the representing measure for $f$. We now show $f\in\Mld{C}$. 
			
			Suppose $x\in [a,b]$ and $x\notin C$. We show $x\notin \textrm{Supp}(\mu)$. We show that there exists $\epsilon>0$ such that $\mu(N_\epsilon(x)) =0$ where $N_\epsilon(x) = \{y; |y-x|<\epsilon\}$. Since $C$ is closed we can find an $\epsilon'>0$ such that $N_{\epsilon'}(x)\cap C=\emptyset$. Take $\psi:\mathbb{R}\to[0,1]$ to be the continuous function with \begin{align*}
				\psi(y)=\begin{cases}
					0 & |y-x|>3\epsilon'/4\\
					1 & |y-x|< \epsilon'/2\\
					1-(4/\epsilon')\{y-(x+\epsilon'/2)\}& x+\epsilon'/2\leq y\leq x+3\epsilon'/4\\
					(4/\epsilon')\{y-(x-3\epsilon'/4)\}& x-3\epsilon'/4\leq y\leq x-\epsilon'/2.
				\end{cases}
			\end{align*} From Lemma~\ref{lem:vague_convergence}, $0\le \mu(N_{\epsilon'/2}(x))\le \int \psi(y)\, \mu(dy)=\underset{n\to\infty}{\lim}\int \psi(y)\,\mu_n(dy)=\underset{n\to\infty}{\lim}0=0$.  Taking $\epsilon = \epsilon'/2$, we obtain $x\notin \textrm{Supp}(\mu)$. Since $x$ was arbitrary, $\textrm{Supp}(\mu)\subseteq C$ and $f\in\Mld{C}$.
			
			Since $\Ml$ is a closed, convex subset of the Hilbert space $\sqs$, the existence and uniqueness of $\Pi(r;C)$ follows from the Hilbert space projection theorem.

		\end{proof}

		\subsection{Proof of Proposition \ref{prop:inner_product_r}}\label{pf:inner_product_r}
		
		\begin{proof}
			In the case ${\rm Supp}(\hat{\mu}_{C})\cap(-1,1)=\emptyset$, the statement in the Proposition is trivially true. Otherwise, suppose ${\rm Supp}(\hat{\mu}_C)\cap (-1,1)$ is nonempty. Let $\tilde{\alpha}\in {\rm Supp}(\hat{\mu}_C)\cap (-1,1)$ given. We show $\braket{\Pi(r;C),x_{\tilde{\alpha}}}=\braket{r,x_{\tilde{\alpha}}}$.
			
			First, we show that $\braket{x_{\alpha},\Pi(r;C)-r} = 0$ for $\hat{\mu}_{C}$-almost every $\alpha$. Let $E=C\cap(-1,1)$. From Lemma~\ref{lem:meas_boundary}, we have $\hat{\mu}_{C}(\{-1,1\})=0$, and from the definition of $\hat{\mu}_{C}$, we have ${\rm Supp}(\hat{\mu}_{C})\subset C$, so $\hat{\mu}_{C}(E^c)=\hat{\mu}_{C}(C^c \cup \{-1,1\}) \le \hat{\mu}_C(C^c) +\hat{\mu}_{C}(\{-1,1\})= 0$. From Proposition~\ref{prop:rhat_variational} and Lemma~\ref{lem:inner_product}, we have \begin{align}
				0=\braket{\Pi(r;C),\Pi(r;C)-r}&=\int \braket{x_{\alpha},\Pi(r;C)-r}\hat{\mu}_{C}(d\alpha)\nonumber\\
				&=\int_E\braket{x_{\alpha},\Pi(r;C)-r}\hat{\mu}_{C}(d\alpha)\label{eq:nonnegative}
			\end{align} From Proposition~\ref{prop:rhat_variational}, $\braket{x_{\alpha},\Pi(r;C)-r}\geq 0$ for all $\alpha\in E$. Thus, from~\eqref{eq:nonnegative} and the fact $\hat{\mu}_{C}(E^c)=0$, we have
			\begin{align}\label{eq:a.e.equality}
				\braket{x_{\alpha},\Pi(r;C)-r}=0, \,\,\mbox{  for $\hat{\mu}_{C}$-a.e. $\alpha$.}   
			\end{align}
			
			Now, we complete the proof that $\braket{\Pi(r;C),x_{\tilde{\alpha}}}=\braket{r,x_{\tilde{\alpha}}}$. Let $\epsilon>0$ given. Choose $R>0$ such that $|\braket{x_{\tilde{\alpha}}-x_{\alpha},\Pi(r;C)-r}|\leq \epsilon$ for all $\alpha\in N_{R}(\tilde{\alpha})$, where $N_{R}(\tilde{\alpha})=\{\alpha:|\alpha-\tilde{\alpha}|<R\}$. Since $$|\braket{x_{\tilde{\alpha}}-x_{\alpha},\Pi(r;C)-r}|\leq \|x_{\tilde{\alpha}}-x_{\alpha}\|\|\Pi(r;C)-r\|$$ and $$\underset{\alpha\to\tilde{\alpha}}{\lim}\;\|x_{\tilde{\alpha}}-x_{\alpha}\|=0,$$ such a choice of $R>0$ exists. Now, since $\tilde{\alpha}\in {\rm Supp}(\hat{\mu}_{C})$ and $N_R(\tilde{\alpha})$ is open, we have $\hat{\mu}_{C}(N_{R}(\tilde{\alpha}))>0$, so \begin{align*}
				\braket{x_{\tilde{\alpha}},\Pi(r;C)-r}&=\{\hat{\mu}_{C}(N_{R}(\tilde{\alpha}))\}^{-1}\int_{N_{R}(\tilde{\alpha})} \braket{x_{\tilde{\alpha}},\Pi(r;C)-r}\hat{\mu}_{C}(d\alpha)\\
				&=\{\hat{\mu}_{C}(N_{R}(\tilde{\alpha}))\}^{-1}\int_{N_{R}(\tilde{\alpha})} \braket{x_{\alpha}+x_{\tilde{\alpha}}-x_\alpha,\Pi(r;C)-r}\hat{\mu}_{C}(d\alpha)\\
				&=\{\hat{\mu}_{C}(N_{R}(\tilde{\alpha}))\}^{-1}\int_{N_{R}(\tilde{\alpha})} \braket{x_{\alpha},\Pi(r;C)-r}\hat{\mu}_{C}(d\alpha)\\
				&+\{\hat{\mu}_{C}(N_{R}(\tilde{\alpha}))\}^{-1}\int_{N_{R}(\tilde{\alpha})} \braket{x_{\tilde{\alpha}}-x_\alpha,\Pi(r;C)-r}\hat{\mu}_{C}(d\alpha)\\
				&=\{\hat{\mu}_{C}(N_{R}(\tilde{\alpha}))\}^{-1}\int_{N_{R}(\tilde{\alpha})} \braket{x_{\tilde{\alpha}}-x_\alpha,\Pi(r;C)-r}\hat{\mu}_{C}(d\alpha)
			\end{align*} where we used \eqref{eq:a.e.equality}. From the choice of $R$, we have $|\braket{x_{\tilde{\alpha}}-x_{\alpha},\Pi(r;C)-r}|\leq \epsilon$ for all $\alpha\in N_{R}(\tilde{\alpha})$, and thus \begin{align*}
				-\epsilon\leq \braket{x_{\tilde{\alpha}},\Pi(r;C)-r}\leq \epsilon.
			\end{align*} Since $\epsilon>0$ was arbitrary, we have $\braket{x_{\tilde{\alpha}},\Pi(r;C)-r}=0$ and thus $\braket{\Pi(r;C),x_{\tilde{\alpha}}}=\braket{r,x_{\tilde{\alpha}}}$. This proves the result.
		\end{proof}

		\subsection{Proof of Proposition \ref{prop:finiteSupport}}\label{pf:finiteSupport}
		\begin{proof} 
			Define $g(\alpha)=\braket{x_\alpha,r-\Pi(r;C)}=\sum_{k\in \mathbb{Z}}\alpha^{|k|}\{r(k)-\Pi(r;C)(k)\}$.
			We consider derivatives of $g$, i.e., for $n\ge 1$, 
			$$g^{(n)}(\alpha) = \frac{d^n}{d\alpha^n} g(\alpha) = \frac{d^n}{d\alpha^n} \sum_{k\in \mathbb{Z}}\alpha^{|k|}\{r(k)-\Pi(r;C)(k)\}$$
			We first show that the term-by-term differentiation of $g(\alpha)$ is justified, so that \begin{align}\label{eq:prop6:g1}
				g^{(1)}(\alpha)=\frac{d}{d\alpha}g(\alpha)=\sum_{k\in\mathbb{Z}:|k|\geq 1}|k|\alpha^{|k|-1}\{r(k)-\Pi(r;C)(k)\}
			\end{align}
			and, similarly, \begin{align*}
				g^{(n)}(\alpha)=\frac{d^n}{d\alpha^n}g(\alpha)=\sum_{k\in\mathbb{Z}:|k|\geq n}\frac{|k|!}{(|k|-n)!}\alpha^{|k|-n}\{r(k)-\Pi(r;C)(k)\}.
			\end{align*}
			
			We first consider the case when $n=1$. Let $\alpha_0\in (-1,1)$ be arbitrary. Let $\tilde{\alpha}=(|\alpha_0|+1)/2$. Then $|\alpha_0|<\tilde{\alpha}<1$.	Take $\beta=\tilde{\alpha}-|\alpha_0|=1-|\alpha_0|/2$ and define $N_{\beta}(\alpha_0)=\{\alpha':|\alpha'-\alpha_0|\leq \beta\}$. We will show that the term by term differentiation of $g(\alpha)$ at $\alpha=\alpha_0$ is justifiable, by showing that each summand in \eqref{eq:prop6:g1} for $\alpha \in N_\beta(\alpha_0)$ is dominated by some absolutely summable $\tilde{g}_1(k)$. 
			
			For $k\in \mathbb{Z}$ and $\alpha\in (-1,1)$, define     \begin{align*}
				g_1(k,\alpha)=\frac{d[\alpha^{|k|}\{r(k)-\Pi(r;C)(k)\}]}{d\alpha}=\begin{cases}
					|k|\alpha^{|k|-1}\{r(k)-\Pi(r;C)(k)\} & k\geq 1\\
					0 & k=0,
				\end{cases}
			\end{align*} and for $k\in \mathbb{Z}$, define $\tilde{g}_{1}(k)=|g_1(k,\tilde{\alpha})|$. Then $|g_1(k,\alpha)|\leq \tilde{g}_1(k)$ for all $\alpha\in N_{\beta}(\alpha_0)$.

			Define the sequence $\Delta=\{\Delta(k)\}_{k=-\infty}^{\infty}$ by
			\begin{align*}
				\Delta(k) = |r(k) - \Pi(r;C)(k)|, \, k\in\Z.
			\end{align*} Since $r(k)=0$ for $|k|>M-1$, we have $r\in \sqs$, and so $\Pi(r;C)\in \sqs$ also. Thus $\Delta\in \sqs$, since 
			\begin{align}\label{eq:prop6:deltabound}
				\|\Delta\|^2=\sum_{k\in\mathbb{Z}}\Delta(k)^2=\sum_{|k|\leq M-1}\{r(k)-\Pi(r;C)(k)\}^2+\sum_{|k|> M-1}\Pi(r;C)(k)^2<\infty.
			\end{align}
			
			For $n\geq 0$ and $\alpha\in (-1,1)$, define the sequence $\tilde{x}_{\alpha,n}=\{\tilde{x}_{\alpha,n}(k)\}_{k=-\infty}^{\infty}$ by
			\begin{align*}
				\tilde{x}_{\alpha,n}(k)=\frac{d^n\{x_{\alpha}(k)\}}{d\alpha^n}=\begin{cases}
					\frac{|k|!}{(|k|-n)!}\alpha^{|k|-n} &|k|\geq n\\
					0 & |k|<n.
				\end{cases}
			\end{align*} We have 
			\begin{align}\label{eq:prop6:xtilde_alphabound}
				\sum_{k\in\mathbb{Z}}\tilde{x}_{\alpha,n}^2 
				= 2\sum_{k\ge n} \frac{k!}{(k-n)!}|\alpha|^{k-n} = 2\sum_{l \ge 0} \frac{(l+n)!}{l!}|\alpha|^{l}<\infty,
			\end{align} for each $\alpha\in (-1,1)$, $n\geq 0$, so $\tilde{x}_{\alpha,n}\in \sqs$ for any $\alpha\in (-1,1)$, $n\geq 0$. Therefore by \eqref{eq:prop6:deltabound} and \eqref{eq:prop6:xtilde_alphabound}, we can conclude that $\tilde{g}_1$ is absolutely summable,
			\begin{align*}
				\sum_{k\in\mathbb{Z}}\tilde{g}_1(k)=\braket{\tilde{x}_{\tilde{\alpha},1},\Delta}\le \|\tilde{x}_{\tilde{\alpha},1}\|\|\Delta\|<\infty.
			\end{align*} 
			Then, by the Lebesgue differentiation theorem, we have \begin{align*}
				g^{(1)}(\alpha_0)=\sum_{k\in\mathbb{Z}}g_1(k,\alpha_0)=\sum_{k\in\mathbb{Z}:|k|\geq 1}|k|\alpha_0^{|k|-1}\{r(k)-\Pi(r;C)(k)\}
			\end{align*} 	\citep[see, e.g., Theorem 2.27 in][]{folland1999real}. Since $\alpha_0\in(-1,1)$ was arbitrary, we have $g^{(1)}(\alpha)=\sum_{k\in\mathbb{Z}:|k|\geq 1}|k|\alpha^{|k|-1}\{r(k)-\Pi(r;C)(k)\}$ for each $\alpha\in (-1,1)$.

			Proceeding similarly, we obtain \begin{align*}
				g^{(n)}(\alpha)=\frac{d^n}{d\alpha^n}g(\alpha)=\sum_{k\in\mathbb{Z}:|k|\geq n}\frac{|k|!}{(|k|-n)!}\alpha^{|k|-n}\{r(k)-\Pi(r;C)(k)\}.
			\end{align*} for $\alpha\in (-1,1)$, $n\geq 1$.
			
			We now show that $\hat{\mu}_{C}$ has finite support. Recall for $|k|>M-1$, $r(k)=0$, so for $n>(M-1)$ we have 
			\begin{align*}
				g^{(n)}(\alpha)&=-2\sum_{k=n}^{\infty}\frac{k!}{(k-n)!}\alpha^{k-n}\Pi(r;C)(k)\\
				&=-2\sum_{k=n}^{\infty}\frac{k!}{(k-n)!}\alpha^{k-n}\int_{[-1,1]}x^{k}\hat{\mu}_{C}(dx).
			\end{align*} 
			For $|\rho|<1$, we have 
			\begin{align}\label{eq:prop6:eq1}
				\sum_{k=n}^{\infty}\frac{k!}{(k-n)!}\rho^{k-n}=\frac{n!}{(1-\rho)^{(n+1)}}.
			\end{align}
			
			Recall $n>(M-1)$. From Lemma \ref{lem:meas_boundary}, we have $\hat{\mu}_{C}(\{-1,1\})=0$ since $\Pi(r;C)\in \Mld{C}$. 
			Thus, by Fubini's theorem, 
			\begin{align*}
				&\sum_{k=n}^{\infty}\int_{[-1,1]}\frac{k!}{(k-n)!}|\alpha|^{k-n}|x|^{k}\hat{\mu}_{C}(dx)\\
				&=\sum_{k=n}^{\infty}\int_{(-1,1)}\frac{k!}{(k-n)!}|\alpha|^{k-n}|x|^{k}\hat{\mu}_{C}(dx)\\
				&=\int_{(-1,1)}\sum_{k=n}^{\infty}\frac{k!}{(k-n)!}|\alpha|^{k-n}|x|^{k}\hat{\mu}_{C}(dx)\\
				&=\int_{(-1,1)}|x|^{n}\sum_{k=n}^{\infty}\frac{k!}{(k-n)!}|\alpha x|^{k-n}\hat{\mu}_{C}(dx)\\
				&=\int_{(-1,1)}\frac{n!|x|^n}{(1-|\alpha x|)^{n+1}}\hat{\mu}_{C}(dx)<\infty
			\end{align*} for $|\alpha|<1$ where the last equality is due to \eqref{eq:prop6:eq1}. 
			Thus, the integral and summation in $g^{(n)}(\alpha)$ may be interchanged, so that 
			\begin{align}\label{eq:prop6:gn_rep}
				g^{(n)}(\alpha)=-2\int_{(-1,1)}\frac{n!x^n}{(1-\alpha x)^{n+1}}\hat{\mu}_{C}(dx).
			\end{align}
			
			Now, we consider two subcases. In the first subcase, ${\rm Supp}(\hat{\mu}_{C})\subset \{0\}$ (that is, $\hat{\mu}_{C}$ is the null measure, or $\hat{\mu}_{C}$ puts point mass on 0 only.) Then ${\rm Supp}(\hat{\mu}_{C})\cap(-1,1)$ contains at most a single point. Otherwise, take $n$ to be the smallest even number such that $n>(M-1)$. Then since the support of $\hat{\mu}_{C}$ contains points away from 0, $g^{(n)}(\alpha)<0$ for all $\alpha\in (-1,1)$, since for even $n$ the integrand in $g^{(n)}(\alpha)$ in \eqref{eq:prop6:gn_rep} is positive for $x\neq 0$.
			Since $g^{(n)}(\alpha)\neq 0$ for all $\alpha\in (-1,1)$, there exist at most $n$ points $-1<\alpha_1<\alpha_2<...<\alpha_{n}<1$ such that $g(\alpha_i)=0$. Thus ${\rm Supp}(\hat{\mu}_{C})\cap(-1,1)$ contains at most $n$ points, where $n$ is the smallest even number with $n>M-1$. Since $\hat{\mu}_{C}$ has a finite number of support points in $(-1,1)$, we have $\sup \{|x|:x\in (-1,1)\cap {\rm Supp}(\hat{\mu}_{C})\}< 1-\epsilon$ for some $\epsilon>0$ and thus $\{-1,1\}\cap {\rm Supp}(\hat{\mu}_{C})=\emptyset$.
		\end{proof}
		
		\section{Proofs for results in Section \ref{sec:guarantee}}
		
		\subsection{Proof of Proposition \ref{prop:initial_estimators}}\label{pf:initial_estimators}
		First we show that \ref{cond:R1}-\ref{cond:R3} hold for the empirical autocovariance sequence $\tilde{r}_M$.
		The convergence in~\ref{cond:R1} is shown in Lemma~\ref{lem:as_convergence_fixedK} which is presented at the end of this proof. Assumption~\ref{cond:R2} holds from the definition of $\tilde{r}_{M}$ in~\eqref{eq:emp1_example}, with the choice $n(M)=M$, and the symmetry in~\ref{cond:R3} also follows from the definition of $\tilde{r}_{M}(k)$ in~\eqref{eq:emp1_example} as \eqref{eq:emp1_example} depends on $k$ only through $|k|$. 
		Finally, we show that
		\begin{align}\label{eq:emp_autocov_decr}
			|\tilde{r}_M(k)|\leq \tilde{r}_M(0)    
		\end{align}
		By symmetry, it is sufficient to prove the result for $k \in \N$. For notational simplicity, let $h(x) = g(x) - Y_M$. First, we consider $0\leq k\leq M-1$. We define $M$ length $M$ vectors $v_j\in\mathbb{R}^{M}$, $j=0,...,M-1$ such that
		\begin{align*}
			v_0 = [h(X_0), h(X_{1}), \dots, h(X_{M-1})],
		\end{align*}
		and for $k=1,\dots,M-1$,
		\begin{align*}
			v_k = [h(X_k), h(X_{k+1}), \dots, h(X_{M-1}), 0,0,\dots,0].
		\end{align*}
		Then, for $0\leq k\leq M-1$, $\|v_k\| \le \|v_0\|$ by the definition of the $v_k$'s. Also, note that $\tilde{r}_M(k) =M^{-1}\braket{v_0,v_k}$ by the definition of empirical autocovariance.
		We have \begin{align*}
			|\tilde{r}_M(k)| = |\frac{1}{M}\braket{v_0,v_k}| \le \frac{1}{M}\|v_0\|\|v_k\| \le \frac{1}{M}\|v_0\|^2 = \tilde{r}_M(0).
		\end{align*} Additionally, $\tilde{r}_M(k)=0\leq \|v_0\|^2=\tilde{r}_M(0)$ for $k$ with $|k|>(M-1)$.

		Now, we argue that a windowed autocovariance $\check{r}_M (k) = w_M(|k|) \tilde{r}_M(k)$ satisfying conditions \ref{cond:W1}-\ref{cond:W3} satisfies \ref{cond:R1}-\ref{cond:R3}. First of all, the symmetry in~\ref{cond:R3} holds since $\check{r}_M (-k) = w_M(|k|) \tilde{r}_M(-k) = w_M(|k|) \tilde{r}_M(k) = \check{r}_M(k), \forall k\in \Z$. Also, $\check{r}_M(0) = \tilde{r}_M(0) \ge |\tilde{r}_M(k)| \ge |\tilde{r}_M(k)||w_M(|k|)| = |\check{r}_M(k)|, \forall k\in \Z$ by \eqref{eq:emp_autocov_decr}, conditions \ref{cond:W1} and $\ref{cond:W2}$. Assumption~\ref{cond:R2} holds with the choice $n(M)=\min\{b_M,M\}$ by conditions \ref{cond:W3}, and assumption~\ref{cond:R1} follows from Lemma \ref{lem:as_convergence_fixedK}, \ref{cond:W3}, and Slutsky's theorem.

		\begin{lem}\label{lem:as_convergence_fixedK} Assume \ref{cond:harris_ergodicity}, \ref{cond:piReversible} and \ref{cond:integrability}.
			Let $k\in \Z$ given. Then \begin{align*}
				\lim_{M\to\infty}\tilde{r}_{M}(k) = \gamma(k), \,\, P_x\mbox{-a.s.}
			\end{align*}  for each $x\in\mathsf{X}$.
		\end{lem}
		
		\begin{proof}
			First, we show that $\bar{r}_{M}(k)\overset{a.s.}{\to}\gamma(k)$ as $M\to\infty$, where $\bar{r}_M(k) = M^{-1}\sum_{t=0}^{M-1-|k|} \bar{g}(X_t) \bar{g}(X_{t+|k|})$ where we define $\bar{g}(x) = g(x) - E_\pi[g(X_0)]$.
			Without loss of generality, assume $k>0$. We define $h_k:\Omega \to \R$ such that $h_k(\omega)=\bar{g}(X_0(\omega))\bar{g}(X_k(\omega))$ for $k\geq 0$. Note $h_k\in L_1(\Omega, \mathcal{F}, P_\pi)$ by \ref{cond:integrability}. Let $\theta: \Omega \to \Omega$ be the shift operator. We first want to show that 
			\begin{align*}
				\bar{r}_M(k) = \frac{1}{M}\sum_{t=0}^{M-1-k} \bar{g}(X_t) \bar{g}(X_{t+k}) = \frac{1}{M}\sum_{t=0}^{M-1-k}\theta^t h_k \to \gamma(k)
			\end{align*}
			$P_x$-almost surely, for any initial condition $x \in \mathsf{X}$.
			
			Using the fact that the set of $P_\pi$-invariant events is trivial due to $X$ being Harris recurrent and from Theorem 17.1.2 in \cite{meyn2009markov}, we have a set $F_h$ of full $\pi$-measure such that for any initial condition in $x \in F_h$,
			\begin{align}\label{eq:lem6_eq1}
				\lim_{M\to\infty} \frac{1}{M}\sum_{t=0}^{M-1-k}\theta^t h_k = \lim_{M\to\infty} \frac{M-k}{M}\frac{1}{M-k}\sum_{t=0}^{M-1-k}\theta^t h_k = E_\pi[h_k] = \gamma(k), \,\,\mbox{a.s. } P_x.
			\end{align}
			Now via a modification of Proposition 17.1.6 in \cite{meyn2009markov}, we show \eqref{eq:lem6_eq1} holds for all $x \in \mathsf{X}$. Define $h_\infty(x) = P_x(\lim_{M\to\infty} M^{-1}\sum_{t=0}^{M-1-k}\bar{g}(X_t) \bar{g}(X_{t+k}) = \gamma(k))$. We know that $h_\infty(x) = 1$ for $x \in F_h$. If we show $h_\infty(x) = 1$ for all $x \in \mathsf{X}$, we have the desirable result. We show that $h_\infty(x)$ is harmonic.
			\begin{align*}
				&Qh_\infty(x) \\
				&= E_x [ P_{X_1}\{\lim_{M\to\infty} M^{-1}\sum_{t=0}^{M-1-k}\bar{g}(X_t) \bar{g}(X_{t+k}) = \gamma(k)\}]\\
				&= E_x [ P_x \{\lim_{M\to\infty} M^{-1}\sum_{t=0}^{M-1-k}\bar{g}(X_{t+1}) \bar{g}(X_{t+1+k})= \gamma(k)| \mathcal{F}_1\}]\\
				&= E_x [ P_x \{\lim_{M\to\infty}\left[ \frac{M+1}{M}\frac{1}{M+1}\sum_{t=0}^{M-k}\bar{g}(X_{t}) \bar{g}(X_{t+k}) -\frac{1}{M}\bar{g}(X_{0}) \bar{g}(X_{k}) \right] = \gamma(k)\}]\\
				&= h_\infty(x).
			\end{align*}
			Therefore $h_\infty(x)=1$ for any $x \in \mathsf{X}$, and \eqref{eq:lem6_eq1} holds for any initial condition $x \in \mathsf{X}$.
			
			Finally, we show that $\tilde{r}_{M}(k)-\bar{r}_M(k)\to0$ as $M\to\infty$, $P_x$-a.s., for all $x \in \mathsf{X}$, where $\tilde{r}_M(k) = M^{-1}\sum_{t=0}^{M-1-k} (g(X_t) - Y_M)(g(X_{t+k}) - Y_M)$ for $Y_M = \sum_{t=0}^{M-1} g(X_t)$. 
			First we have $Y_M \to \mu$, $P_x$-almost surely for all $x\in \mathsf{X}$ by SLLN in Theorem 17.1.7 in \cite{meyn2009markov}. For any $k \in \N$, we have,
			\begin{align}
				\tilde{r}_M(k) - \bar{r}_M(k) &= \frac{1}{M}\sum_{t=0}^{M-1-k} \{(g(X_t) - Y_M)(g(X_{t+k}) - Y_M)-(g(X_t) - \mu)(g(X_{t+k}) - \mu)\}\nonumber\\
				& = \frac{1}{M}\sum_{t=0}^{M-1-k} \{(Y_M-\mu)(g(X_{t+k}) +g(X_t)) + \mu^2- Y_M^2\}\nonumber\\
				& = (Y_M-\mu) \frac{1}{M}\sum_{t=0}^{M-1-k}  \{g(X_{t+k}) +g(X_t)\} + \frac{M-k}{M} (\mu^2 - Y_M^2)\label{eq:emp_mean_diff}.
			\end{align}
			
			Since both $\frac{1}{M-k} \sum_{t=0}^{M-1-k} g(X_{t+k})$  and $\frac{1}{M-k} \sum_{t=0}^{M-1-k} g(X_t)$ converge to $\mu$ by SLLN in Theorem 17.1.7 in \cite{meyn2009markov}, $\frac{1}{M}\sum_{t=0}^{M-1-k}  \{g(X_{t+k}) +g(X_t)\} \to 2\mu$, $P_x$-almost surely for any $x\in\mathsf{X}$.
			Therefore, by continuous mapping theorem, \eqref{eq:emp_mean_diff} $\to 0$ as $M\to\infty$, $P_x$-almost surely for any $x\in\mathsf{X}$, which proves the result.
		\end{proof}

		\subsection{Proof of Proposition \ref{prop:xalpha_conv} }\label{pf:xalpha_conv}
		\begin{proof}
			First, we show that 
			\begin{align}\label{eq:xalpha_conv}
				|\braket{x_{\alpha},\rinit-\gamma}| \to 0
			\end{align} $P_x$-a.s. for any $x\in\mathsf{X}$, for any $\alpha \in (-1,1)$. For the ease of notation, if a $P_x$-almost sure convergence holds for any $x\in \mathsf{X}$, we will just say the convergence holds almost surely.
			Let $\epsilon>0$ given. Note for any $B>0$, \begin{align*}
				\braket{x_{\alpha},\rinit-\gamma}&=\sum_{k=-(B-1)}^{(B-1)}\alpha^{|k|}\{\rinit(k)-\gamma(k)\}
				+2\sum_{k=B}^{\infty}\alpha^{|k|}\{\rinit(k)-\gamma(k)\}.
			\end{align*} 
			since $\rinit(k) = \rinit(-k)$ by \ref{cond:R3}.
			Choose $B$ such that \begin{align*}
				\sum_{k=B}^{\infty}|\alpha|^{k}\gamma(0)=\frac{|\alpha|^B}{1-|\alpha|}\gamma(0)\leq \epsilon/4.
			\end{align*} Then\begin{align*}
				&\underset{M\to\infty}{\lim\sup}\;|\sum_{k=B}^{\infty}\alpha^{k}\{\rinit(k)-\gamma(k)\}|\\
				&\leq \underset{M\to\infty}{\lim\sup}\;\sum_{k=B}^{\infty}|\alpha|^{k}\{|\rinit(k)|+|\gamma(k)|\}\\
				&\leq \underset{M\to\infty}{\lim\sup}\;\sum_{k=B}^{\infty}|\alpha|^{k}\{\rinit(0)+\gamma(0)\}\\
				&=\epsilon/2
			\end{align*} 
			where the second inequality uses $|\rinit(k)|\le \rinit(0)$ in \ref{cond:R3} and the equality uses $\rinit(0)\overset{a.s.}{\to}\gamma(0)$ by \ref{cond:R1}.
			
			Furthermore, we have \begin{align*}
				\sum_{k=-(B-1)}^{(B-1)}\alpha^{|k|}\{\rinit(k)-\gamma(k)\}\overset{a.s.}{\to}0
			\end{align*} since $\rinit(k)\overset{a.s.}{\to}\gamma(k)$ for each $k\in\mathbb{Z}$. Thus \begin{align*}
				\underset{M\to\infty}{\lim \sup}\;|\braket{x_{\alpha},\rinit-\gamma}|\leq \epsilon.
			\end{align*} Since $\epsilon>0$ was arbitrary, we have $\braket{x_{\alpha},\rinit-\gamma}\overset{a.s.}{\to}0$ as $M\to\infty$. This proves the a.s. convergence result for each $\alpha \in (-1,1)$.
			
			Now, we show that the convergence is uniform over $\mathcal{K}$. First, let $\delta_0$ denote the minimum distance between $\mathcal{K}$ and $\{-1,1\}$, i.e., $\delta_0 = \inf\{\min(|1-x|,|-1-x|): x\in \mathcal{K}\}$. Since $\mathcal{K} \subset (-1,1)$, the gap $\delta_0>0$. Since for $x\in \mathcal{K}$, $x\in (-1,1)$, we have $\delta_0\leq 1$. If $\delta_0=1$, then $\mathcal{K}=\{0\}$ since $\mathcal{K}$ is nonempty by assumption, and $\underset{\alpha\in\mathcal{K}}{\sup}|\braket{x_{\alpha},\prinit-\gamma}|\overset{a.s.}{\to}0$ from the previously shown convergence for each $\alpha\in(-1,1)$. 
			
			Otherwise, suppose $\delta_0<1$. Let $\epsilon_1>0$ be given, and choose $\beta>0$ such that 
			\begin{align}\label{eq:prop7_beta}
				\beta = \epsilon_1/(4\delta_0^2\gamma(0)).
			\end{align}
			For $\alpha\in(-1,1)$, define $B_{\beta}(\alpha)= \{|x-\alpha|\le \beta\}\cap [-1+\delta_0,1-\delta_0]$. Take $N(\beta)=\lceil 2(1-\delta)/\beta\rceil$ and define $\alpha_j=(-1+\delta_0)+j\beta$, $j=0,...,N(\beta)-1$. Then $\mathcal{K} \subset \cup_{j=0}^{N(\beta)-1} B_\beta(\alpha_j)$ and 
			\begin{align}
				&\sup_{\alpha \in \mathcal{K}} |\braket{x_\alpha, \rinit-\gamma}| \nonumber\\
				&\leq \max_{j =0,\dots,N(\beta)-1}\; \sup_{\alpha \in B_\beta(\alpha_j)} |\braket{x_\alpha, \rinit-\gamma}| \nonumber\\
				&\leq \max_{j=0,\dots,N(\beta)-1}\; \sup_{\alpha \in B_\beta(\alpha_j)}  |\braket{x_\alpha-x_{\alpha_j},\rinit-\gamma}|+|\braket{x_{\alpha_j},\rinit-\gamma}| \nonumber\\
				&\leq \max_{j=0,\dots,N(\beta)-1}\; \sup_{\alpha \in B_\beta(\alpha_j)}|\braket{x_\alpha-x_{\alpha_j},\rinit-\gamma}|+\max_{j=0,\dots,N(\beta)-1}|\braket{x_{\alpha_j},\rinit-\gamma}|\label{eq:twoterms}
			\end{align}
			From convergence \eqref{eq:xalpha_conv}, we have \begin{align}
				\underset{M\to\infty}{\lim\sup}\max_{j=0,\dots,N(\beta)-1}|\braket{x_{\alpha_j},\rinit-\gamma}|=0 \label{eq:limsup1}
			\end{align} almost surely since $N(\beta)$ is finite. 
			
			To control the size of the first term by the distance between $\alpha$ and $\alpha_i$, we have the following Lemma, the proof of which is deferred to the end of this proof.
			
			\begin{lem}\label{lem:r_xalpha_diff}
				For any $r$ such that $r(0) \geq0$, $r(k) = r(-k)$, and $|r(k)|\le r(0)$ for $k\in \Z$, and $\alpha,\beta\in (-1,1)$, we have 
				\begin{align*}
					|\braket{r,x_\alpha-x_\beta}|\le  \frac{2r(0)}{(1-|\alpha|)(1-|\beta|)} |\alpha-\beta|.
				\end{align*}
			\end{lem}
			From Lemma \ref{lem:r_xalpha_diff}, we have, for $\alpha\in B_{\beta}(\alpha_j)$,
			\begin{align*}
				|\braket{x_\alpha-x_{\alpha_j}, \rinit-\gamma}| 
				&= |\braket{x_\alpha-x_{\alpha_j},\rinit} - \braket{x_\alpha-x_{\alpha_j},\gamma}| \\
				&\leq |\braket{x_\alpha-x_{\alpha_j},\rinit}| + |\braket{x_{\alpha}-x_{\alpha_j},\gamma}|\\
				&\le \frac{2|\alpha-{\alpha_j}|}{(1-|\alpha|)(1-|{\alpha_j}|)}(\rinit(0)+\gamma(0))\\
				& \le 2(\beta/\delta_0^2)\{\rinit(0)+\gamma(0)\}.
			\end{align*} 
			Since this bound does not depend on $j$, we have, 
			\begin{align}
				&\underset{M\to\infty}{\lim\sup}\max_{j=0,\dots,N(\beta)-1}\; \sup_{\alpha \in B_\beta(\alpha_j)}|\braket{x_\alpha-x_{\alpha_j},\rinit-\gamma}|\nonumber\\
				&\leq\underset{M\to\infty}{\lim\sup} \;2(\beta/\delta_0^2)\{\rinit(0)+\gamma(0)\}\nonumber\\
				&=4\beta\delta_0^{-2}\gamma(0)\label{eq:limsup2}
			\end{align} since $\rinit(0)\overset{a.s.}{\to}\gamma(0)$. Thus, from~\eqref{eq:twoterms},~\eqref{eq:limsup1}, and~\eqref{eq:limsup2}, we have \begin{align*}
				&\underset{M\to\infty}{\lim \sup}\;\underset{\alpha \in \mathcal{K}}{\sup} |\braket{x_\alpha, \rinit-\gamma}|\\
				&\leq \underset{M\to\infty}{\lim \sup}\max_{j=0,\dots,N(\beta)-1}\; \sup_{\alpha \in B_\beta(\alpha_j)}|\braket{x_\alpha-x_{\alpha_j},\rinit-\gamma}|\nonumber\\
				&+\underset{M\to\infty}{\lim \sup}\max_{j=0,\dots,N(\beta)-1}|\braket{x_{\alpha_j},\rinit-\gamma}|\\
				&\leq 4\beta\delta_0^{-2}\gamma(0) = \epsilon_1.
			\end{align*} 
			where the last equality is due to the choice of $\beta$ in \eqref{eq:prop7_beta}. But $\epsilon_1$ was arbitrary, so $\underset{M\to\infty}{\lim }\;\underset{\alpha \in \mathcal{K}}{\sup}|\braket{x_\alpha, \rinit-\gamma}|=0$ almost surely. This proves the result.
			
		\end{proof}
		
		Now we present the proof for Lemma \ref{lem:r_xalpha_diff}.
		\begin{proof}[Proof of Lemma \ref{lem:r_xalpha_diff}]
			By definition,
			\begin{align*}
				|\braket{\rinit,x_\alpha-x_\beta}| 
				&= |2 \sum_{k=1}^\infty r(k) \{\alpha^k - \beta^k\}|\\
				&\le 2 r(0) \sum_{k=1}^\infty |\alpha^k - \beta^k|
			\end{align*}
			where the second inequality uses the fact that $\max_{k\ge 1} |r(k)| \le r(0)$.
			Using the following equality:
			\begin{align*}
				\alpha^k-\beta^k = (\alpha-\beta) \sum_{j=1}^k \alpha^{k-j} \beta^{j-1}
			\end{align*}
			we have,
			\begin{align*}
				|\braket{r,x_\alpha-x_\beta}|
				&\le 2 r(0) \sum_{k=1}^\infty |(\alpha-\beta) \sum_{j=1}^k \alpha^{k-j} \beta^{j-1}|\\
				&\le 2 r(0) |\alpha-\beta| \sum_{k=1}^\infty  \sum_{j=1}^k |\alpha^{k-j} \beta^{j-1}|\\
				& = 2 r(0) |\alpha-\beta| \sum_{j=1}^\infty  \sum_{k=j}^\infty |\alpha|^{k-j} |\beta|^{j-1}\\
				& = 2 r(0) |\alpha-\beta| \frac{1}{(1-|\alpha|)(1-|\beta|)}
			\end{align*}
			where the last equality follows from
			\begin{align*}
				\sum_{j=1}^\infty  \sum_{k=j}^\infty |\alpha|^{k-j} |\beta|^{j-1} 
				=  \sum_{j=1}^\infty |\beta|^{j-1} \sum_{k=j}^\infty |\alpha|^{k-j}  =\frac{1}{(1-|\alpha|)(1-|\beta|)}.
			\end{align*}

		\end{proof}
		
		\subsection{Proof of Proposition \ref{prop:finite_muhat}}\label{pf:finite_muhat}
		
		\begin{proof}
			By Lemma \ref{lem:inner_product} and Proposition \ref{prop:finiteSupport}, we have,
			\begin{align}
				\braket{\prinit,\prinit}& = \int_{[-1,1]} \braket{\prinit, x_\alpha} \m_{\delta,M}(d\alpha)\nonumber\\
				&= \sum_{\alpha \in {\rm Supp}(\m_{\delta,M})} \braket{\prinit,x_\alpha} \m_{\delta,M}(\{\alpha\}) \nonumber\\
				&=  \sum_{\alpha \in {\rm Supp}(\m_{\delta,M})} \braket{\rinit,x_\alpha} \m_{\delta,M}(\{\alpha\}) \label{eq:muhatSum}
			\end{align}
			where the last equality is due to Proposition \ref{prop:inner_product_r}. On the one hand,
			\begin{align}
				\braket{\prinit,\prinit}  &= \sum_{\alpha \in {\rm Supp}(\m_{\delta,M})} \braket{\prinit,x_\alpha} \m_{\delta,M}(\{\alpha\}) \nonumber\\
				&=\sum_{\alpha,\alpha' \in {\rm Supp}(\m_{\delta,M})} \braket{x_\alpha,x_{\alpha'}} \m_{\delta,M}(\{\alpha\}) \m_{\delta,M}(\{\alpha'\})\nonumber \\
				&\geq  \inf_{\alpha,\alpha' \in [-1+\delta,1-\delta]} \braket{x_\alpha,x_{\alpha'}}  \sum_{\alpha,\alpha' \in {\rm Supp}(\m_{\delta,M})}\m_{\delta,M}(\{\alpha\}) \m_{\delta,M}(\{\alpha'\}) \label{eq:muhatLower}
			\end{align}
			since $ \braket{x_\alpha,x_{\alpha'}}  = \frac{1+\alpha \alpha'}{1-\alpha \alpha'} >0$ for any $\alpha, \alpha' \in (-1,1)$.
			On the other hand, we have from~\eqref{eq:muhatSum} that
			\begin{align}\label{eq:muhatUpper}
				\braket{\prinit,\prinit} \le \sup_{\alpha \in [-1+\delta,1-\delta]} |\braket{\rinit,x_\alpha}| \sum_{\alpha \in {\rm Supp}(\m_{\delta,M})} \m_{\delta,M}(\{\alpha\}).
			\end{align}
			Thus, from~\eqref{eq:muhatLower} and \eqref{eq:muhatUpper}, we have
			\begin{align*}
				&\inf_{\alpha,\alpha' \in [-1+\delta,1-\delta]} \braket{x_\alpha,x_{\alpha'}}  \sum_{\alpha,\alpha' \in {\rm Supp}(\m_{\delta,M})}\m_{\delta,M}(\{\alpha\}) \m_{\delta,M}(\{\alpha'\}) \\
				&\qquad \le \sup_{\alpha \in [-1+\delta,1-\delta]} |\braket{\rinit,x_\alpha}| \sum_{\alpha \in {\rm Supp}(\m_{\delta,M})} \m_{\delta,M}(\{\alpha\}).
			\end{align*}
			That is,
			\begin{align}\label{eq:prop1-eq3}
				\sum_{\alpha \in {\rm Supp}(\m_{\delta,M})} \m_{\delta,M}(\{\alpha\})  \le \frac{\sup_{\alpha \in [-1+\delta,1-\delta]} |\braket{\rinit,x_\alpha}|}{\inf_{\alpha,\alpha' \in [-1+\delta,1-\delta]} \braket{x_\alpha,x_{\alpha'}} }.
			\end{align}
			The denominator is deterministic, and we let $C_0:= \inf_{\alpha,\alpha' \in [-1+\delta,1-\delta]} \braket{x_\alpha,x_{\alpha'}} $.
			Now we show that the numerator is bounded almost surely. We have,
			\begin{align*} 
				\sup_{\alpha \in [-1+\delta,1-\delta]} |\braket{\rinit,x_\alpha}| \leq 
				\sup_{\alpha \in [-1+\delta,1-\delta]} |\braket{\rinit-\gamma,x_\alpha}| +\sup_{\alpha \in [-1+\delta,1-\delta]} |\braket{\gamma,x_\alpha}|   
			\end{align*}
			The second term $\sup_{\alpha \in [-1+\delta,1-\delta]} |\braket{\gamma,x_\alpha}| $ is deterministic and bounded by $\gamma(0)(2-\delta)/\delta$ from Holder's inequality.  For the first term, from Proposition \ref{prop:xalpha_conv}, we have \begin{align*}
				\underset{M\to\infty}{\lim\sup}\;\underset{\alpha \in [-1+\delta,1-\delta]}{\sup} |\braket{\rinit-\gamma,x_\alpha}|=0
			\end{align*} almost surely. Define $C_{\delta,\gamma}=\frac{\gamma(0)(2-\delta)}{(\delta) (C_0)}$. Then
			\begin{align*}
				&\underset{M\to\infty}{\lim \sup}\;\m_{\delta,M}([-1+\delta,1-\delta])\\
				&=\underset{M\to\infty}{\lim \sup}\;\sum_{\alpha \in {\rm Supp}(\m_{\delta,M})} \m_{\delta,M}(\{\alpha\}) \\
				&\le \frac{\gamma(0)(2-\delta)}{(\delta) (C_0)}=C_{\delta,\gamma}
			\end{align*} almost surely by \eqref{eq:prop1-eq3}.
		\end{proof}
		
		\section{Proofs for results in Section \ref{sec:emp}}\label{pf:delta_hat}
		
		\begin{prop}
			Suppose $X_{0},X_1,...,$ is a Markov chain with transition kernel $Q$ satisfying~~\ref{cond:harris_ergodicity}-\ref{cond:geometric_ergodicity}, and suppose $g:\mathsf{X}\to\mathbb{R}$ satisfies~\ref{cond:integrability}. Let $\gamma$ denote the autocovariance sequence as defined in Proposition~\ref{prop:gamma_moment_seq}, and let $F$ denote the representing measure for $\gamma$. Assume that $\gamma(0)= {\rm Var}_\pi(g(X_0))>0$. Let $\rho(k)=\gamma(k)/\gamma(0)$, $k\in\mathbb{Z}$ denote the autocorrelation sequence and $\hat{\rho}_M(k)=\tilde{r}_{M}(k)/\tilde{r}_{M}(0)$ denote the empirical autocorrelation sequence. Define $\delta_\gamma$ such that $\delta_\gamma = 1-\sup\{|x|; x \in {\rm Supp}(F)\}$ where $F$ is the representing measure for $\gamma$. Let 
			\begin{align}
				\hat{\delta}_{M} = 1-\exp\{-\log M/(2\hat{m})\},
			\end{align}
			with  $\hat{\delta}_{M}:=1$ in the case $\hat{m}=0$.
			Suppose in addition to~~\ref{cond:harris_ergodicity}-\ref{cond:geometric_ergodicity} and~\ref{cond:integrability} that
			\begin{align}
				\sup_{k=0,\dots,M-1} |\hat{\rho}_M(k) - \rho(k) | = O_{P_x}(\sqrt{\frac{\log M}{M}})
			\end{align} with respect to the Markov chain law $P_{x}$ for each $x\in \textsf{X}$. 
			Choose $\hat{m}$ such that
			\begin{align}\label{eq:def_trunc}
				\hat{m} = \min\{t\in 2\mathbb{N}; \hat{\rho}_M(t+2) \le c_M\sqrt{\frac{\log M}{M}}\}
			\end{align}
			for some $c_M \ge 0$. Then we have $\hat{\delta}_{M}$ is asymptotically not larger than $\delta_\gamma$, i.e., for any $\epsilon>0$, we have $\lim_{M\to\infty} P_x(\hat{\delta}_{M} \geq \delta_\gamma+\epsilon) =0$.
			
			Furthermore, under the assumption of $c_M \to \infty$ such that $c_M=O(\log M)$, $\hat{\delta}_{M}$ converges in $P_x$-probability to $\delta_\gamma$, i.e., for any $\epsilon>0$, we have $\lim_{M\to\infty} P_x(|\hat{\delta}_{M} -\delta_\gamma|\geq \epsilon) =0$.
		\end{prop}
		
		\begin{proof}
			Define $\tilde{F}$ such that 
			\begin{align*}
				\tilde{F}((-\infty,t]) = 
				\begin{cases}
					0 & t<0\\
					\gamma(0)^{-1} F([-t,t]) & t\ge 0\\
				\end{cases}
			\end{align*}
			Let $H_{\tilde{F}}$ be the distribution function of $\tilde{F}$. Note $\delta_{\gamma}$ exists and is finite since $\gamma(0)>0$ implies $\textrm{Supp}(F)$ is nonempty, and $\textrm{Supp}(F)\subset [-1,1]$. Also we note $\tilde{F}((-\infty,\infty))= \gamma(0)^{-1} F((-\infty,\infty)) =1$ since $\gamma(0) = \int \alpha^0 F(d\alpha) = F((-\infty,\infty))$.
			
			We first show that $1-\delta_{\gamma}$ is the smallest value of $b$ such that $[-b,b]$ has full $F$-measure, i.e., $1-\delta_{\gamma} =  \inf \{t; F([-t,t]) \ge \gamma(0)\} = \inf \{t ; H_{\tilde{F}}(t) \geq 1\}$, and for $a>\delta_\gamma$, $F([-(1-a),(1-a)]) < \gamma(0)$.
			
			In the case $\delta_{\gamma}=1$, then $\textrm{Supp}(F)=\{0\}$, and for $a>\delta_{\gamma}$, $[-(1-a),(1-a)]=\emptyset$, which is not full measure since $\gamma(0)>0$. We next consider the case $\delta_{\gamma}<1$. From the definition of $\delta_{\gamma}$, $\textrm{Supp}(F)\subset [-(1-\delta_{\gamma}),1-\delta_{\gamma}]$. Now, consider $a$ such that $\delta_{\gamma}<a\leq 1$. We show $[-(1-a),1-a]$ is not full $F$ measure. Since $\textrm{Supp}(F)$ is closed, we have $\{-(1-\delta_{\gamma}),1-\delta_{\gamma}\}\cap \textrm{Supp}(F)\neq \emptyset$. 
			Let $N_{\theta}(x)=\{y:|y-x|<\theta\}$ denote the open $\theta$-neighborhood of $x$. Define $\theta_0=(a-\delta_{\gamma})/2$. Then $\{-(1-\delta_{\gamma}),1-\delta_{\gamma}\}\cap \textrm{Supp}(F)\neq \emptyset$ implies the open set $A=N_{\theta_0}(1-\delta_{\gamma})\cup N_{\theta_0}(-(1-\delta_{\gamma}))$ has $F(A)>0$, but $A\cap [-(1-a),1-a]=\emptyset$, and so $F([-(1-a),1-a]) < F((-\infty,\infty)) = \gamma(0)$. Thus 
			\begin{align}\label{eq:1-delta_as_infimum}
				1-\delta_{\gamma} =  \inf \{t; F([-t,t]) \ge \gamma(0)\} = \inf \{t ; H_{\tilde{F}}(t) \geq 1\}\end{align}

			From the definition of $\hat{m}$, we have
			\begin{align}
				\hat{\rho}_M(\hat{m}) \geq  c_M\sqrt{\frac{\log M}{M}}
				\qquad \mbox{and} \qquad
				\hat{\rho}_M(\hat{m}+2) \leq  c_M\sqrt{\frac{\log M}{M}}.\label{eq:mHatdef}
			\end{align}
			First, we consider the case $\delta_{\gamma}<1$. 
			
			Let $\Delta_M$ be $\Delta_M = \sup_{k=0,\dots,M-1}|\rho(k) - \hat{\rho}_M(k)|$. Since $\Delta_M = O_{P_x} (\sqrt{\log M/M})$, we have $C_\beta >0$ and a finite $M_1$ such that $\Delta_M \le C_\beta \sqrt{\log M/M}$ with probability at least $1-\beta$ for all $M \ge M_1$. Let $\mathcal{E}_M$ be the event such that this inequality holds.
			
			On the event $\mathcal{E}_M$, the second condition in \eqref{eq:mHatdef} implies
			\begin{align*}
				&\hat{\rho}_M(\hat{m}+2) \leq  c_M\sqrt{\frac{\log M}{M}}\\
				&\Rightarrow \rho(\hat{m}+2) -|\hat{\rho}_M(\hat{m}+2) - \rho(\hat{m}+2)|\leq  c_M\sqrt{\frac{\log M}{M}}\\
				&\Rightarrow  \rho(\hat{m}+2) \leq (C_\beta+c_M) \sqrt{\frac{\log M}{M}}.
			\end{align*}
			Note by definition of $\rho$,
			\begin{align*}
				\rho(\hat{m}+2) = \gamma(0)^{-1} \int \alpha^{\hat{m}+2} F(d\alpha) =\gamma(0)^{-1} \int |\alpha|^{\hat{m}+2} F(d\alpha) =  \int \alpha^{\hat{m}+2} \tilde{F}(d\alpha).
			\end{align*}
			
			We will lower-bound $\int \alpha^{\hat{m}+2} \tilde{F}(d\alpha)$. First, define $\{a_k\}_{k=1}^\infty$ such that
			\begin{align}
				a_k = \sup \{t\ge 0; H_{\tilde{F}}(t) < 1-1/\sqrt{\log(k)}\},
			\end{align}
			where we take the convention of $\sup\{\emptyset\} = -\infty$. 
			By definition of $H_{\tilde{F}}$, we have $H_{\tilde{F}}(1-\delta_\gamma) = 1$ and \begin{align}
				&H_{\tilde{F}}(1-\delta_\gamma) - H_{\tilde{F}}(a_k) \geq 1-(1-1/\sqrt{\log(k)}) = 1/\sqrt{\log(k)}
			\end{align}
			Also, $a_k \ge a_{k+1}$ since $H_{\tilde{F}}$ is an increasing function. 
			Now, we show $\lim_{k\to\infty} a_k = 1-\delta_\gamma$: first, we have
			$a_k \le 1-\delta_\gamma$ for all $k$. Therefore the limit of $a_k$ exists. Now suppose to the contrary that $\lim_k a_k  = c_a < 1-\delta_\gamma$. Then $H_{\tilde{F}}(c_a) <1$ by \eqref{eq:1-delta_as_infimum}. Choose $\{\epsilon_k\}$ such that $\epsilon_k\ge 0$, $a_k+\epsilon_k \le 1-\delta_\gamma$ and $\epsilon_k \to 0$. For each $k$, we have $H_{\tilde{F}}(a_k+ \epsilon_k ) \ge 1-1/\sqrt{\log (k)}$ by the definition of $a_k$. Then taking $k$ limit to both sides, we have $\lim_k H_{\tilde{F}}(a_k +\epsilon_k )  = H_{\tilde{F}}(c_a)\ge 1	$ since $H_{\tilde{F}}$ is a right continuous function, and we have a contradiction. Therefore we have $\lim_k a_k = 1-\delta_\gamma$.
			
			We have,
			\begin{align*}
				\int \alpha^{\hat{m}+2} \tilde{F}(d\alpha) \ge \int_{(a_M, 1-\delta_\gamma]} \alpha^{\hat{m}+2} \tilde{F}(d\alpha) \ge a_M^{\hat{m}+2}\int_{(a_M, 1-\delta_\gamma]} \tilde{F}(d\alpha) = a_M^{\hat{m}+2}(H_{\tilde{F}}(1-\delta_\gamma) - H_{\tilde{F}}(a_M)).
			\end{align*}
			Since 
			\begin{align*}
				a_M^{\hat{m}+2}(H_{\tilde{F}}(1-\delta_\gamma) - H_{\tilde{F}}(a_M)) \ge a_M^{\hat{m}+2}/\sqrt{\log M},
			\end{align*}
			we have,
			\begin{align*}
				&\int \alpha^{\hat{m}+2} \tilde{F}(d\alpha)  \leq (c_M+C_\beta) \sqrt{\frac{\log M}{M}}\\
				&\Rightarrow  a_M^{\hat{m}+2}\leq \sqrt{\log M} (c_M+C_\beta) \sqrt{\frac{\log M}{M}}\\
				&\Rightarrow  (\hat{m}+2)\log a_M \leq \log(\sqrt{\log M}) + \log (c_M+C_\beta) + \frac{1}{2} \{\log(\log M)-\log M\}\\
				&\Rightarrow \frac{\hat{m}}{\log M} \geq \frac{-2\log a_M+\log(\sqrt{\log M}) + \log (c_M+C_\beta) + \frac{1}{2} \{\log(\log M)-\log M\}}{\log a_M \log M}\\
				&\Rightarrow 1-\exp(-\frac{\log M}{2\hat{m}}) \le \\
				&\qquad 1-\exp\left\lbrace-\frac{1}{2}\left[\frac{\log a_M \log M}{-2\log a_M+\log(\sqrt{\log M}) + \log (c_M+C_\beta) + \frac{1}{2} \{\log(\log M)-\log M\}}  \right] \right\rbrace
			\end{align*}
			Since $\log a_M \to \log (1-\delta_\gamma)$ as $M\to \infty$, the RHS converges to $\delta_\gamma$. In other words, there exists a finite $M_2$ such that the RHS is $\delta_\gamma + \epsilon_0$. Therefore, 
			\begin{align}\label{eq:ineq2}
				\hat{\delta}_{M} \leq  \delta_\gamma +\epsilon_0
			\end{align} 
			on $\mathcal{E}_M$ for $M\ge M_2$.
			Therefore, for $M\geq \max\{M_1,M_2\}$, we have $P_x(\hat{\delta}_{M} \leq  \delta_\gamma +\epsilon_0)\geq P(\mathcal{E}_{M})\geq 1-\beta$ for arbitrarily chosen $\beta$ and $\epsilon_0$, i.e., $\hat{\delta}_M$ is asymptotically not greater than $\delta_\gamma$.
			
			Now under the condition that $c_M \to \infty$ and $c_M = O(\log M)$, we show that  $\hat{\delta}_M$ is asymptotically not smaller than $\delta_\gamma$ as well. The first condition in \eqref{eq:mHatdef} implies on the event $\mathcal{E}_M$,
			\begin{align*}
				&\hat{\rho}_M(\hat{m}) \ge c_M\sqrt{\frac{\log M}{M}}\\
				&\Rightarrow \rho(\hat{m}) + |\rho(\hat{m}) - \hat{\rho}_M(\hat{m})| \ge c_M\sqrt{\frac{\log M}{M}}\\
				&\Rightarrow \rho(\hat{m}) + C_\beta \sqrt{\frac{\log M}{M}} \ge c_M\sqrt{\frac{\log M}{M}}\\
				& \Rightarrow \gamma(0)^{-1}\int \alpha^{\hat{m}} F(d\alpha) \ge (c_M - C_\beta)\sqrt{\frac{\log M}{M}}
			\end{align*}
			Note that if we only require $c_M \ge 0$, the RHS can be negative depending on $c_M$ and $C_\beta$. However, with the choice of $c_M \to \infty$, there exists a finite $M_3$ such that $c_M-C_\beta>0$ for $M \ge M_3$. 
			
			We continue to upper bound the LHS. Since $\hat{m}$ is even,
			$\gamma(0)^{-1}\int \alpha^{\hat{m}} F(d\alpha) = \gamma(0)^{-1}\int |\alpha|^{\hat{m}} F(d\alpha) = \int \alpha^{\hat{m}}\tilde{F}(d\alpha)$. In particular, ${\rm Supp} (\tilde{F}) \subseteq [0, 1-\delta_\gamma]$, since $\tilde{F}([0,1-\delta_\gamma]) = \tilde{F}((-\infty, 1-\delta_\gamma]) = \gamma(0)^{-1} F([-1+\delta_\gamma, 1-\delta_\gamma]) = 1$ by the definition of $\delta_\gamma$.
			Therefore,
			\begin{align*}
				&\int \alpha^{\hat{m}}\tilde{F}(d\alpha) \ge (c_M - C_\beta)\sqrt{\frac{\log M}{M}}\\
				&\Rightarrow (1-\delta_\gamma)^{\hat{m}}\ge (c_M - C_\beta)\sqrt{\frac{\log M}{M}}\\
				&\Rightarrow {\hat{m}}\log (1-\delta_\gamma)\ge \log (c_M - C_\beta)+ \frac{1}{2}\{\log(\log M)-\log M\}\\
				&\Rightarrow {\hat{m}}\le \frac{\log (c_M - C_\beta)}{\log (1-\delta_\gamma)}+ \frac{1}{2\log (1-\delta_\gamma)}\{\log(\log M)-\log M\}	
			\end{align*}
			since $\log (1-\delta_\gamma) <0$. By dividing both sides by $\log(M)/2$,
			\begin{align*}
				&\frac{2\hat{m}}{\log M}\le \frac{2\log (c_M - C_\beta)+\log(\log M)-\log M}{\log (1-\delta_\gamma)\log(M)}\\
				&\Rightarrow 1-\exp(-\frac{2\hat{m}}{\log M}) \ge 1-\exp\left(-\frac{\log (1-\delta_\gamma)\log(M)}{2\log (c_M - C_\beta)+\log(\log M)-\log M}\right)
			\end{align*}
			By the condition of $c_M = O(\log(M))$, $c_M/\log(M)\le C$ for some constant $C$ for a sufficiently large $M$.
			\begin{align*}
				\frac{\log(c_M -C_\beta)}{\log(M)}\le\frac{\log(c_M)}{\log(M)} \le \frac{\log(C \log(M))}{\log(M)} = o(1).
			\end{align*}
			Therefore, the RHS converges to $\delta_\gamma$, and we can find a finite $M_4$ such that
			\begin{align}\label{eq:ineq1}
				\hat{\delta}_M \ge \delta_\gamma -\epsilon_0
			\end{align} on $\mathcal{E}_{M}$, for $M \ge \max\{M_3,M_4\}$. Thus, under the additional condition that $c_M \to \infty$ and $c_M = O(\log M)$, combining~\eqref{eq:ineq2} and~\eqref{eq:ineq1} yields $P_x(\{|\hat{\delta}_M - \delta_\gamma|> \epsilon_0\})\leq P(\mathcal{E}_{M}^c)\leq \beta$, for $M \ge \max_{i=1,\dots,4} M_i$, i.e.,  $\hat{\delta}_{M} \to \delta_\gamma$ in probability since $\beta$ and $\epsilon_0$ were arbitrary. This shows the result in the case $\delta_{\gamma}<1$.
			
			Now we consider the case when $\delta_\gamma=1$. In this case, the inequality $\hat{\delta}_M \leq \delta_\gamma=1$ is trivially true with probability 1 from the definition of $\hat{m}$, and therefore, we have $P_x(\hat{\delta}_{M} \leq  \delta_\gamma +\epsilon_0)=1$ for all $M$, for each $\epsilon_0>0$. Thus $\hat{\delta}_{M}$ is not asymptotically larger than $\delta_{\gamma}$. 
			
			Now, under the additional assumption $c_{M}\to \infty$ and $c_M=O(\log(M))$, we show $\hat{\delta}_{M}\overset{p}{\to}\delta_{\gamma}=1$. Let $\epsilon_0$, $\beta>0$ given. As before, let $\Delta_M$ be $\Delta_M = \sup_{k=0,\dots,M-1}|\rho(k) - \hat{\rho}_M(k)|$. Since $\Delta_M = O_{P_x} (\sqrt{\log M/M})$, we have $C_\beta >0$ and a finite $M_1$ such that $\Delta_M \le C_\beta \sqrt{\log M/M}$ with probability at least $1-\beta$ for all $M \ge M_1$. Let $\mathcal{E}_M$ denote the event $\Delta_M\leq C_{\beta}\sqrt{\frac{\log M}{M}}$. We also have a finite $M_2$ such that $c_M\geq C_{\beta}$ for $M\geq M_2$. Therefore, on $\mathcal{E}_{M}$, 
			\begin{align*}
				\hat{\rho}_{M}(2)\leq C_{\beta}\sqrt{\frac{\log M}{M}} \Rightarrow \hat{\rho}_{M}(2)\leq c_M\sqrt{\frac{\log M}{M}}\Rightarrow \hat{m}=0,
			\end{align*} holds for all $M\geq M_2$. Note $\hat{\delta}_{M}=1$ whenever $\hat{m}=0$. Thus for $M\geq \max\{M_1,M_2\}$, 
			\begin{align*}
				P_x(|\hat{\delta}_{M}-\delta_{\gamma}|\geq \epsilon_0)\leq P_x(\hat{\delta}_{M}\neq \delta_{\gamma})=P_x(\hat{m}\neq 0)\leq \beta.
			\end{align*} Since $\beta$ was arbitrary, this proves the result $\hat{\delta}_{M}\overset{p}{\to}\delta_{\gamma}=1$.
		\end{proof}

		\newpage
		\section{Supplementary Tables for Section \ref{sec:emp}}\label{sec:tables}
		Here we present some supplementary tables for Section \ref{sec:emp}.
		\begin{table}[H]
			\caption{\textit{Estimated average $\ell_2$ error (s.e.) for the autocovariance sequence estimators and mean squared error (s.e.) for the asymptotic variance estimators for discrete state space Metropolis-Hastings example}}
			\centering
			\setlength\tabcolsep{3pt}
			\begin{subtable}[t]{\textwidth}
				
				\caption{$\ell_2$ error}
				\centering
				\footnotesize
				\begin{tabularx}{\linewidth}{lXXXXXX}
					\hline
					Estimator & 4000 & 8000 & 16000 & 32000 & 64000 & 128000 \\ 
					\hline
					Empirical & 1.2732 (0.0074) & 1.2639 (0.0051) & 1.2616 (0.0035) & 1.2668 (0.0025) & 1.2683 (0.0018) & 1.2666 (0.0013) \\ 
					Bartlett & 0.0092 (0.0003) & 0.0056 (0.0002) & 0.0034 (0.0001) & 0.0021 (0.0001) & 0.0012 (0.0000) & 0.0007 (0.0000) \\ 
					MomentLS(Orcl,Brtl) & \textbf{0.0054 (0.0003)} & \textbf{0.0029 (0.0001)} & 0.0016 (0.0001) & 0.0009 (0.0000) & 0.0005 (0.0000) & 0.0003 (0.0000) \\ 
					MomentLS(Tune,Emp) & 0.0059 (0.0003) & 0.0030 (0.0002) & 0.0016 (0.0001) & \textbf{0.0008 (0.0000)} & \textbf{0.0004 (0.0000)} & \textbf{0.0002 (0.0000)} \\ 
					MomentLS(Tune-Incr,Emp) & 0.0059 (0.0003) & 0.0030 (0.0002) & 0.0016 (0.0001) & \textbf{0.0008 (0.0000)} & \textbf{0.0004 (0.0000)} & \textbf{0.0002 (0.0000)} \\ 
					MomentLS(Orcl,Emp) & 0.0056 (0.0003) & \textbf{0.0029 (0.0001)} & \textbf{0.0015 (0.0001)} & \textbf{0.0008 (0.0000)} & \textbf{0.0004 (0.0000)} & \textbf{0.0002 (0.0000)} \\ 
					\hline
				\end{tabularx}
			\end{subtable}
			\vspace{.2em}
			
			\begin{subtable}[t]{\textwidth}
				
				\centering
				\footnotesize
				\caption{Asymptotic variance mean squared error}
				\begin{tabularx}{\linewidth}{lXXXXXX}
					\hline
					Estimator & 4000 & 8000 & 16000 & 32000 & 64000 & 128000 \\ 
					\hline
					BM & 0.0917 (0.0056) & 0.0548 (0.0035) & 0.0377 (0.0024) & 0.0234 (0.0017) & 0.0134 (0.0009) & 0.0086 (0.0006) \\ 
					OLBM & 0.0940 (0.0058) & 0.0559 (0.0036) & 0.0332 (0.0022) & 0.0204 (0.0014) & 0.0118 (0.0008) & 0.0070 (0.0005) \\ 
					Empirical & 6.4180 (0.0000) & 6.4180 (0.0000) & 6.4180 (0.0000) & 6.4180 (0.0000) & 6.4180 (0.0000) & 6.4180 (0.0000) \\ 
					Bartlett & 0.0921 (0.0058) & 0.0541 (0.0034) & 0.0325 (0.0022) & 0.0201 (0.0014) & 0.0115 (0.0008) & 0.0069 (0.0005) \\ 
					Init-Positive & 0.1236 (0.0146) & 0.0571 (0.0051) & 0.0332 (0.0041) & 0.0154 (0.0015) & 0.0084 (0.0011) & 0.0038 (0.0003) \\ 
					Init-Decr & 0.0878 (0.0079) & 0.0400 (0.0030) & 0.0230 (0.0022) & 0.0114 (0.0009) & 0.0058 (0.0005) & 0.0030 (0.0002) \\ 
					Init-Convex & 0.0741 (0.0063) & 0.0348 (0.0026) & 0.0195 (0.0019) & 0.0101 (0.0008) & 0.0051 (0.0004) & 0.0026 (0.0002) \\ 
					MomentLS(Orcl,Brtl) & \textbf{0.0521 (0.0038)} & \textbf{0.0274 (0.0018)} & \textbf{0.0148 (0.0010)} & 0.0085 (0.0005) & 0.0048 (0.0003) & 0.0030 (0.0002) \\ 
					MomentLS(Tune,Emp) & 0.0697 (0.0059) & 0.0342 (0.0025) & 0.0183 (0.0017) & 0.0092 (0.0006) & 0.0046 (0.0004) & 0.0023 (0.0002) \\ 
					MomentLS(Tune-Incr,Emp) & 0.0675 (0.0056) & 0.0336 (0.0025) & 0.0179 (0.0017) & 0.0090 (0.0006) & 0.0046 (0.0004) & 0.0023 (0.0002) \\ 
					MomentLS(Orcl,Emp) & 0.0558 (0.0043) & 0.0285 (0.0020) & \textbf{0.0148 (0.0011)} & \textbf{0.0079 (0.0005)} & \textbf{0.0038 (0.0003)} & \textbf{0.0020 (0.0002)} \\ 
					\hline
				\end{tabularx}
			\end{subtable}
			
		\end{table}

		\begin{table}[H]
			\caption{\textit{Estimated average $\ell_2$ error (s.e.) for the autocovariance sequence estimators and mean squared error (s.e.) for the asymptotic variance estimators for AR1 example with $\rho=0.9$}}
			\centering
			\setlength\tabcolsep{3pt}
			\begin{subtable}[t]{\textwidth}
				
				\caption{$\ell_2$ error}
				\centering
				\footnotesize
				\begin{tabularx}{\linewidth}{lXXXXXX}
					\hline
					Estimator & 4000 & 8000 & 16000 & 32000 & 64000 & 128000 \\ 
					\hline
					Empirical & 260.4808 (3.3948) & 262.8923 (2.6730) & 261.6753 (1.8479) & 261.7598 (1.3180) & 263.6259 (0.9379) & 263.6811 (0.6410) \\ 
					Bartlett & 7.1962 (0.2886) & 4.4781 (0.1685) & 2.7572 (0.0995) & 1.6033 (0.0540) & 0.9710 (0.0312) & 0.5964 (0.0175) \\ 
					MomentLS(Orcl,Brtl) & 3.9916 (0.2211) & 2.3139 (0.1283) & 1.3216 (0.0762) & 0.6979 (0.0396) & 0.4162 (0.0241) & 0.2675 (0.0140) \\ 
					MomentLS(Tune,Emp) & 4.8135 (0.2362) & 2.9854 (0.1719) & 1.5289 (0.0819) & 0.7512 (0.0416) & 0.4092 (0.0276) & 0.1970 (0.0117) \\ 
					MomentLS(Tune-Incr,Emp) & 4.7958 (0.2356) & 2.9674 (0.1700) & 1.5241 (0.0816) & 0.7509 (0.0416) & 0.4082 (0.0275) & 0.1973 (0.0117) \\ 
					MomentLS(Orcl,Emp) & \textbf{3.7863 (0.2114)} & \textbf{2.1209 (0.1209)} & \textbf{1.1063 (0.0675)} & \textbf{0.5141 (0.0322)} & \textbf{0.2680 (0.0189)} & \textbf{0.1247 (0.0078)} \\ 
					\hline
				\end{tabularx}
			\end{subtable}
			\vspace{.2em}
			
			\begin{subtable}[t]{\textwidth}
				\footnotesize
				\centering
				\caption{Asymptotic variance mean squared error}
				\begin{tabularx}{\linewidth}{lXXXXXX}
					\hline
					Estimator & 4000 & 8000 & 16000 & 32000 & 64000 & 128000 \\ 
					\hline
					BM & 441.3225 (26.3386) & 310.6322 (18.7844) & 207.5458 (12.0461) & 120.8707 (7.1677) & 75.7108 (4.5795) & 50.9232 (3.1021) \\ 
					OLBM & 530.6842 (30.8907) & 329.0258 (18.6833) & 200.7452 (11.5641) & 113.6880 (6.9584) & 70.3730 (4.6056) & 43.9050 (2.8280) \\ 
					Empirical & 10,000.0000 (0.0000) & 10,000.0000 (0.0000) & 10,000.0000 (0.0000) & 10,000.0000 (0.0000) & 10,000.0000 (0.0000) & 10,000.0000 (0.0000) \\ 
					Bartlett & 480.1497 (28.6448) & 304.9819 (18.0249) & 188.2221 (11.0456) & 108.3490 (6.7378) & 67.6490 (4.5006) & 42.5953 (2.7620) \\ 
					Init-Positive & 727.0617 (102.9961) & 384.9189 (52.4958) & 200.7905 (24.7859) & 94.0732 (9.7885) & 53.6517 (5.0255) & 25.4210 (2.3289) \\ 
					Init-Decr & 442.0562 (37.6246) & 289.4032 (30.6812) & 141.7235 (12.9728) & 70.1543 (5.7635) & 39.3842 (3.5904) & 19.6494 (1.7146) \\ 
					Init-Convex & 349.7933 (23.8527) & 240.7814 (22.5304) & 120.4678 (9.4864) & 59.7319 (4.4887) & 34.3007 (3.1422) & 16.8201 (1.4084) \\ 
					MomentLS(Orcl,Brtl) & 206.4103 (13.2641) & 121.6584 (7.7392) & 73.4146 (4.9726) & 40.7843 (2.6673) & 24.8478 (1.5947) & 16.6819 (0.9671) \\ 
					MomentLS(Tune,Emp) & 317.3013 (21.2726) & 217.8352 (18.9499) & 103.8053 (7.7694) & 52.8779 (4.0613) & 30.1932 (3.0648) & 14.2880 (1.3012) \\ 
					MomentLS(Tune-Incr,Emp) & 313.1335 (20.8190) & 214.4295 (18.4725) & 103.0196 (7.7013) & 52.5359 (4.0221) & 30.0120 (3.0511) & 14.2162 (1.2912) \\ 
					MomentLS(Orcl,Emp) & \textbf{187.2898 (12.1154)} & \textbf{104.0514 (6.7444)} & \textbf{56.0977 (4.0863)} & \textbf{26.4687 (1.9108)} & \textbf{13.3904 (1.0194)} & \textbf{6.2344 (0.4571)} \\ 
					\hline
				\end{tabularx}
			\end{subtable}
			
		\end{table}

		\begin{table}[H]
			\caption{\textit{Estimated average $\ell_2$ error (s.e.) for the autocovariance sequence estimators and mean squared error (s.e.) for the asymptotic variance estimators for AR1 example with $\rho=-0.9$}}
			\centering
			\setlength\tabcolsep{3pt}
			\begin{subtable}[t]{\textwidth}
				\footnotesize
				\caption{$\ell_2$ error}
				\centering
				\begin{tabularx}{\linewidth}{lXXXXXX}
					\hline
					Estimator & 4000 & 8000 & 16000 & 32000 & 64000 & 128000 \\ 
					\hline
					Empirical & 260.2496 (3.5967) & 260.7740 (2.6222) & 260.8336 (1.8960) & 264.7212 (1.2805) & 263.1831 (0.8522) & 263.6014 (0.6256) \\ 
					Bartlett & 6.6264 (0.2658) & 4.1319 (0.1757) & 2.6333 (0.0993) & 1.5941 (0.0591) & 0.9635 (0.0292) & 0.5946 (0.0165) \\ 
					MomentLS(Orcl,Brtl) & 3.7670 (0.2187) & 2.0407 (0.1308) & 1.2248 (0.0719) & 0.6844 (0.0449) & 0.3941 (0.0237) & 0.2585 (0.0146) \\ 
					MomentLS(Tune,Emp) & 4.6199 (0.2543) & 2.5181 (0.1457) & 1.3821 (0.0754) & 0.7286 (0.0432) & 0.3639 (0.0217) & 0.1968 (0.0110) \\ 
					MomentLS(Tune-Incr,Emp) & 4.6113 (0.2541) & 2.5137 (0.1459) & 1.3793 (0.0753) & 0.7305 (0.0436) & 0.3638 (0.0219) & 0.1974 (0.0111) \\ 
					MomentLS(Orcl,Emp) & \textbf{3.5910 (0.2156)} & \textbf{1.8474 (0.1215)} & \textbf{1.0294 (0.0627)} & \textbf{0.5192 (0.0368)} & \textbf{0.2388 (0.0169)} & \textbf{0.1236 (0.0079)} \\ 
					\hline
				\end{tabularx}
		\end{subtable}
		\vspace{.2em}
		
		\begin{subtable}[t]{\textwidth}
			\footnotesize
			\centering
			\caption{Asymptotic variance mean squared error}
				\begin{tabularx}{\linewidth}{lXXXXXX}
					\hline
					Estimator & 4000 & 8000 & 16000 & 32000 & 64000 & 128000 \\ 
					\hline
					BM & 0.0048 (0.0003) & 0.0030 (0.0002) & 0.0017 (0.0001) & 0.0010 (0.0001) & 0.0006 (0.0000) & 0.0005 (0.0000) \\ 
					OLBM & 0.0025 (0.0002) & 0.0018 (0.0001) & 0.0011 (0.0001) & 0.0007 (0.0000) & 0.0005 (0.0000) & 0.0003 (0.0000) \\ 
					Empirical & 0.0767 (0.0000) & 0.0767 (0.0000) & 0.0767 (0.0000) & 0.0767 (0.0000) & 0.0767 (0.0000) & 0.0767 (0.0000) \\ 
					Bartlett & 0.0027 (0.0002) & 0.0019 (0.0002) & 0.0011 (0.0001) & 0.0007 (0.0001) & 0.0005 (0.0000) & 0.0003 (0.0000) \\ 
					Init-Positive & 0.0973 (0.0076) & 0.0526 (0.0036) & 0.0232 (0.0016) & 0.0114 (0.0008) & 0.0059 (0.0004) & 0.0032 (0.0002) \\ 
					Init-Decr & 0.1186 (0.0079) & 0.0663 (0.0041) & 0.0309 (0.0020) & 0.0149 (0.0010) & 0.0078 (0.0005) & 0.0044 (0.0003) \\ 
					Init-Convex & 0.3009 (0.0134) & 0.1664 (0.0076) & 0.0835 (0.0037) & 0.0417 (0.0019) & 0.0209 (0.0009) & 0.0116 (0.0005) \\ 
					MomentLS(Orcl,Brtl) & 0.0028 (0.0002) & 0.0017 (0.0001) & 0.0010 (0.0001) & 0.0007 (0.0000) & 0.0004 (0.0000) & 0.0003 (0.0000) \\ 
					MomentLS(Tune,Emp) & 0.0034 (0.0002) & 0.0019 (0.0001) & 0.0010 (0.0001) & 0.0006 (0.0000) & 0.0003 (0.0000) & \textbf{0.0001 (0.0000)} \\ 
					MomentLS(Tune-Incr,Emp) & 0.0034 (0.0002) & 0.0019 (0.0001) & 0.0010 (0.0001) & 0.0006 (0.0000) & 0.0003 (0.0000) & \textbf{0.0001 (0.0000)} \\ 
					MomentLS(Orcl,Emp) & \textbf{0.0023 (0.0001)} & \textbf{0.0012 (0.0001)} & \textbf{0.0006 (0.0000)} & \textbf{0.0003 (0.0000)} & \textbf{0.0001 (0.0000)} & \textbf{0.0001 (0.0000)} \\ 
					\hline
				\end{tabularx}
			\end{subtable}
			
		\end{table}
		\newpage

		\ifnum\pageoption=2
		\putbib 
	\end{bibunit}\fi
	
	\ifnum\pageoption=3
	\bibliographystyle{plainnat}  
	\bibliography{bib}
	\fi
	
	\fi
\end{document}